\documentclass[11pt]{article}
\usepackage[a4paper,margin=1in]{geometry}
\usepackage{tikz}
\usepackage[T1]{fontenc}
\usepackage{lmodern}
\usepackage{amsmath,amssymb,amsthm,mathtools,microtype,enumitem,booktabs}
\usepackage[colorlinks=true,linkcolor=blue,citecolor=blue,urlcolor=blue]{hyperref}
\newcommand{\Z}{\mathbb Z}
\newcommand{\R}{\mathbb R}
\DeclareMathOperator{\vol}{vol}

\DeclareMathOperator{\conv}{conv}
\DeclareMathOperator{\polylog}{polylog}
\DeclareMathOperator{\supp}{supp}
\newcommand{\softO}[1]{\widetilde O\!\left(#1\right)}
\newcommand{\cT}{\mathcal T}

\newtheorem{theorem}{Theorem}[section]
\newtheorem{lemma}[theorem]{Lemma}

\newtheorem{corollary}[theorem]{Corollary}
\newtheorem{importedtheorem}{Theorem}[section]

\newtheorem{importedlemma}[importedtheorem]{Lemma}
\theoremstyle{definition}
\newtheorem{definition}[theorem]{Definition}
\theoremstyle{remark}
\newtheorem{remark}[theorem]{Remark}
\newtheorem*{remark*}{Remark}
\setlist{nosep}

\title{Subset Sum and Integer Programming Feasibility Test in Near-Optimal Time}
\author{Lin Chen, Yuchen Mao, Guochuan Zhang\footnote{\{chenlin198662,maoyc,zgc\}@zju.edu.cn. Zhejiang University.}}
\date{}

\begin{document}
\maketitle
\begin{abstract}
We give an
$O_{d,\varepsilon}(n+N^{d+\varepsilon})$-time algorithm for multidimensional
subset sum on an explicitly listed multiset of $n$ vectors in
$[-N,N]^d$, for every fixed $d$ and $\varepsilon>0$. This directly implies an $O_{d,\varepsilon}(n+N^{d+\varepsilon})$-time algorithm for feasibility test of
integer programming. Prior to our work, the best-known algorithm has a running time of
$O_{d}(n+N^{O(d^2)})$ [Eisenbrand, Weismantel'19]. Specifically, if $d=1$, our algorithm
runs in $O_{\varepsilon}(n+N^{1+\varepsilon})$-time for the standard subset sum problem, which is near-linear in the largest input integer.
Prior to our work, subset sum can be solved in time near-linear in the target $t$ [Bringmann'17], or near-linear in the mixed parameter $\sqrt{Nt}$ [Chen, Lian, Mao, Zhang'24]. It is a fundamental open problem whether the algorithm can be made near-linear in $N$, which is settled affirmatively in this paper up to an arbitrary small $N^{\varepsilon}$ term.
\end{abstract}

\section*{AI Disclosure} The main proof route for 1D subset sum was designed by the authors, including extracting the additive combinatorics structure of the input integers, lifting 1D subset sum to a suitable constant dimension, decomposing it into searching subproblems in the interior and boundary of a zonotope. The overall framework builds upon a recent paper \cite{CHMZDense} on the multidimensional subset sum problem in the dense regime, coauthored by the authors of this paper. GPT Astra Ultra was used to turn the insights of the authors into complete mathematical proofs.  GPT and Claude was used for drafting the paper and also verified this paper’s main results using the Lean 4 proof. The authors take full responsibility for
the entire content, technical claims, and integrity of the paper.

\tableofcontents

\section{Introduction}\label{sec:introduction}

This paper studies \emph{multidimensional subset sum}: given an explicitly
listed multiset $A=\{\vec a_1,\ldots,\vec a_n\}\subseteq\Z^d$ and a
target $\vec t\in\Z^d$, decide whether
\[
 \sum_{i\in I}\vec a_i=\vec t
 \qquad\text{for some }I\subseteq\{1,\ldots,n\}.
\]
Writing the input vectors as the columns of a matrix $\mathsf A$ identifies
this problem with feasibility of the binary integer program
$\mathsf A\vec x=\vec t$, $\vec x\in\{0,1\}^n$. Finite lower and upper
bounds on the variables, encoded in binary, give bounded integer
programming; we treat it as a corollary, using proximity to cap
multiplicities. The one-dimensional problem is \emph{Subset Sum}.

For a set or multiset $A$, write $\Sigma(A)$ for its total sum and
$\mathcal S(A)=\{\Sigma(B):B\subseteq A\}$ for its set of subset sums.
The size $|A|$ counts all copies, whereas $|\supp(A)|$ counts distinct
values; each input copy may be selected at most once. We study inputs whose
coordinates have absolute value at most $N$, in fixed dimension $d$, and
seek pseudopolynomial algorithms with a small exponent of $N$.
$\widetilde O(\cdot)$ suppresses polylogarithmic factors.

\subsection{Subset sum}
For positive inputs, a nontrivial target can be as large as
$\Sigma(A)\le nN$, so near-linear dependence on $t$ does not give
near-linear dependence on $n+N$. Unless stated otherwise, the prior bounds
below concern positive inputs.

Bellman's dynamic program runs in $O(nt)$ time~\cite{Bellman} and
Pisinger's algorithm in $O(nN)$ time~\cite{Pisinger}. Bringmann~\cite{Bringmann}
obtained a randomized $\softO{n+t}$ algorithm, derandomized by
Chan~\cite{ChanDet}. Bounds independent of the target were obtained by
Polak, Rohwedder, and W\k{e}grzycki~\cite{PRW}, $\softO{n+N^{5/3}}$, and
improved to $\softO{n+N^{3/2}}$ by Chen, Lian, Mao, and
Zhang~\cite{CLMZProximity} and independently by Jin~\cite{JinSubsetSum}.
Chen, Lian, Mao, and Zhang~\cite{CLMZImproved} gave a randomized
$\softO{n+\sqrt{Nt}}$ algorithm, which together with the dense subset sum
algorithm of Bringmann and Wellnitz~\cite{BW} gives $\softO{n+N^{5/4}}$
for distinct integers. Table~\ref{tab:subset-sum-bounds} summarizes these
bounds.

\begin{table}[t]
\centering
\small
\begin{tabular}{@{}p{.51\linewidth}p{.15\linewidth}p{.27\linewidth}@{}}
\toprule
Reference & Running time & Input restriction \\
\midrule
Bellman (1957)~\cite{Bellman} & $O(nt)$ & Arbitrary multiset \\
Pisinger (1999)~\cite{Pisinger} & $O(nN)$ & Arbitrary multiset \\
Bringmann (2017)~\cite{Bringmann} & $\softO{n+t}$ & Arbitrary multiset\\
Polak--Rohwedder--W\k{e}grzycki (2021)~\cite{PRW}
  & $\softO{n+N^{5/3}}$ & Arbitrary multiset \\
Chen--Lian--Mao--Zhang (2024)~\cite{CLMZProximity}; Jin (2023)~\cite{JinSubsetSum}
  & $\softO{n+N^{3/2}}$ & Arbitrary multiset \\
Chen--Lian--Mao--Zhang (2024)~\cite{CLMZImproved}
  & $\softO{n+\sqrt{Nt}}$ & Arbitrary multiset \\
Chen--Lian--Mao--Zhang (2024)~\cite{CLMZImproved} with~\cite{BW}
  & $\softO{n+N^{5/4}}$ & Distinct integers \\
Chan (2026)~\cite{ChanDet}
  & $\softO{n+t}$ & Arbitrary multiset \\
This paper & $O_\varepsilon(n+N^{1+\varepsilon})$ & Arbitrary signed multiset \\
\bottomrule
\end{tabular}
\caption{Selected pseudopolynomial running times for Subset Sum.  Here $n$
is the number of explicitly listed copies, $N$ is the maximum input integer
(maximum absolute value in our result), and $\varepsilon>0$ is fixed.
The multiset algorithms in~\cite{PRW,CLMZProximity} also support
compressed multiplicities; their bounds are stated here for explicitly
listed input to use a common convention. The listed prior algorithms
other than Bellman's, Pisinger's, and Chan's are randomized. Our algorithm is
deterministic and its bound is in the word-RAM model.}
\label{tab:subset-sum-bounds}
\end{table}

\paragraph{A central open problem.}
Obtaining near-linear dependence on the maximum input integer has remained
a major open problem:
\begin{quote}
\emph{Can Subset Sum be solved in $\widetilde O(n+N)$ time?}
\end{quote}
Bringmann and Wellnitz~\cite{BW} identify it as the most crucial open
problem in this line of research, Polak, Rohwedder, and
W\k{e}grzycki~\cite{PRW} describe it as a major open problem, and
Chen, Lian, Mao, and Zhang~\cite{CLMZImproved} and Jin~\cite{JinKnapsack}
also list it. Chen, Hu, Mao, and Zhang~\cite{CHMZDense} argue that
removing the nondegeneracy condition from density-based methods is as hard
as obtaining an $\widetilde O(N^{1+o(1)})$-time algorithm for arbitrary
one-dimensional sets. We obtain $O_\varepsilon(n+N^{1+\varepsilon})$ for
every fixed $\varepsilon>0$.

\subsection{Multidimensional subset sum and integer programming}
Bounded integer programming asks whether
\begin{equation}\label{eq:intro-bounded-ip}
 \mathsf A\vec x=\vec t,\qquad
 \vec\ell\le\vec x\le\vec u,\qquad \vec x\in\Z^n
\end{equation}
is feasible, where $d$ counts the equations (not the variable bounds),
$N$ bounds the absolute coefficients, and $\ell_{\rm in}$ is the binary
input length.

Eisenbrand and Weismantel~\cite{EW} use the Steinitz lemma to bound the
distance from an LP vertex to an integer optimum, which yields coefficient
dependence $N^{O(d^2)}$ for bounded variables. Without upper bounds,
Jansen and Rohwedder~\cite{JR} obtain $N^{O(d)}$ through discrepancy and
convolution, and Gribanov et al.~\cite{GKCKMM} give, for a full-row-rank
matrix, a deterministic feasibility bound
$O(\kappa_d)^d\delta_{\det}\operatorname{poly}(\ell_{\rm in})$, where
$\delta_{\det}\le d^{d/2}N^d$ is the largest full-rank minor. For bounded
variables, Rohwedder and W\k{e}grzycki~\cite[Theorem 6]{RWILP} give
$n^{d+1}O(dN)^d\log(2+\|\vec u-\vec\ell\|_\infty)$, improved by Grage,
Jansen, and Schumacher~\cite[Theorem 8]{GJS} to
$n^{d+1}O(N)^d\log(2+\|\vec u-\vec\ell\|_\infty)$. Randolph and
W\k{e}grzycki~\cite{RW} connect subset sum on sets of small doubling with
integer programming, and Chen, Hu, Mao, and Zhang~\cite{CHMZDense} solve
dense multidimensional subset sum using filled lattice regions inside the
zonotope of the input. Our vector theorem imposes no density or
nondegeneracy hypotheses.

\paragraph{The bounded-variable question.}
Jansen and Rohwedder~\cite[Introduction]{JR} ask whether the linear
dependence of the coefficient exponent on $d$ survives upper bounds on
the variables; Dadush, L\'eonard, Rohwedder, and
Verschae~\cite[Conclusion]{DLRV} and Randolph and
W\k{e}grzycki~\cite[Introduction]{RW} discuss the same gap. The
feasibility part asks:
\begin{quote}
\emph{Can bounded integer-programming feasibility be decided in
$f(d)N^{O(d)}\operatorname{poly}(\ell_{\rm in})$ time?}
\end{quote}
Here the degree of the polynomial in $\ell_{\rm in}$ must be absolute,
which distinguishes the question from the $n^{d+1}N^d$ bounds above.
Corollary~\ref{cor:bounded-ip} answers it with
$O_{d,\varepsilon}(N^{d+\varepsilon}\operatorname{poly}(\ell_{\rm in}))$;
the dependence on $d$ alone is not optimized.

\subsection{Our results}
\begin{theorem}[Deterministic vector multiset subset sum]
\label{thm:main-vector}
Fix $d\ge1$ and $\varepsilon>0$. Let
$A=\{\vec a_1,\ldots,\vec a_n\}$ be an explicitly listed multiset in
\[
 \prod_{j=1}^d[-N_j,N_j]\cap\Z^d,
 \qquad N_j\in\Z_{\ge0},\qquad \nu =\max\{2,N_1,\ldots,N_d\}.
\]
There is a deterministic algorithm that, given $\vec t\in\Z^d$, decides
whether $\vec t\in\mathcal{S}(A)$ in
\[
 O_{d,\varepsilon}\left(n+\nu ^\varepsilon
                    \prod_{j=1}^d(2N_j+1)\right)
\]
word-RAM time. In particular, for $A\subseteq[-N,N]^d$ with $N\ge2$,
the bound is $O_{d,\varepsilon}(n+N^{d+\varepsilon})$.
There is no restriction on the number of distinct vectors or on
their multiplicities.
\end{theorem}

The initial words have $O_d(\log(n\nu +2))$ bits; an arbitrarily long
target encoding adds its input-reading cost. Theorem~\ref{thm:main-vector}
is proved in Section~\ref{sec:deterministic}.

\begin{corollary}[Deterministic subset sum with bounded integers]
\label{thm:main-multiset}
For every fixed $\varepsilon>0$, subset sum for an explicitly listed
multiset $A=\{a_1,\ldots,a_n\}\subseteq[-N,N]\cap\Z$, $N\ge2$, can be
decided deterministically in $O_\varepsilon(n+N^{1+\varepsilon})$
word-RAM time, with unrestricted multiplicities.
\end{corollary}

\begin{corollary}[Bounded integer-programming feasibility]
\label{cor:bounded-ip}
Fix $d\ge1$ and $\varepsilon>0$. Let
$\mathsf A\in\Z^{d\times n}$ have entries of absolute value at most
$N\ge2$. Let $\vec t\in\Z^d$ be a target and
$\vec\ell,\vec u\in\Z^n$ be finite variable bounds, all encoded
in binary, with $\vec\ell\le\vec u$. Feasibility of
\eqref{eq:intro-bounded-ip} can be decided deterministically, and a
feasible vector returned when one exists, in
\[
 O_{d,\varepsilon}\!\left(
     N^{d+\varepsilon}\operatorname{poly}(\ell_{\rm in})\right)
\]
bit operations, where $\ell_{\rm in}$ is the total binary input length.
The degree of the polynomial in $\ell_{\rm in}$ is absolute.
\end{corollary}

Section~\ref{sec:ip-consequences} proves Corollary~\ref{cor:bounded-ip}
and derives extensions to inequalities, prescribed cardinalities, vector
partition and discrepancy, bounded-coefficient objectives, Partition, and
two-machine scheduling $P2\Vert C_{\max}$.

\paragraph{A constructive John-type theorem.}
As a by-product we obtain a deterministic description of the subset sums
of an arbitrary signed multiset of integers, up to a controlled boundary
region and a small exceptional multiset (Theorem~\ref{thm:john} in
Section~\ref{sec:john}). After removing $\widetilde O(h)$ exceptional items
$E$, the remaining items $B$ are images $\rho(U)$ of vectors in a
bounded-rank box $C$, and a reserve $R\subseteq U$ of $O(h)$ items makes
every lattice point of the zonotope $\vec z_0+Z(U\setminus R)$ a subset
sum, while $\mathcal S(U)\subseteq Z(U)\subseteq\vec z_0+Z(U\setminus R)+|R|\cdot C$.
This builds on the common-GAP framework of Conlon, Fox, and
Pham~\cite[Theorem 1.5]{CFP} and the rounding-and-reserve principle of
Chen, Hu, Mao, and Zhang~\cite[Lemmas 30--31]{CHMZDense}; related
constructive results of Chen, Mao, and Zhang~\cite{CMZSparse} find long
arithmetic progressions in sparse subset sums.

\paragraph*{Independent and Concurrent Work} Independently and concurrently to our work, Jin and Mao~\cite{JM26} also obtained an $O(n + N^{1 + o(1)})$-time algorithm for Subset Sum. Both papers use the Conlon--Fox--Pham theorem and Macbeath regions in the design of the algorithms.

\subsection{Technical overview}\label{sec:set-overview}
The heart of the paper is the following randomized warm-up result for
sets of integers. Sections~\ref{sec:carrier}--\ref{sec:recursion} prove
it in full; Section~\ref{sec:deterministic} then lists the changes
needed for multiplicities, vector input, and determinism.

\begin{theorem}[Randomized set case]\label{thm:main}
For every fixed $\varepsilon,c>0$, there is a randomized algorithm that,
given a set $A\subseteq[-N,N]\cap\Z$ and an integer target $t$, decides
whether $t\in\mathcal{S}(A)$ in $O_{\varepsilon,c}(N^{1+\varepsilon})$ time,
with error probability at most $N^{-c}$ and no false positives.
\end{theorem}

\paragraph{Additive structure becomes geometry.}
Dense subset sum algorithms~\cite{GM,BW} use that sufficiently dense
inputs have an interval of subset sums in the middle of their range.
Sparse inputs need not, but long sumsets still contain generalized
arithmetic progressions (GAPs) of bounded rank~\cite{SV,CFP}. A rank-$r$
GAP is the image of an integer box under a linear map $\rho:\Z^r\to\Z$, so
holes in one dimension may be the shadow of a simple filled object in
higher dimension. For example, for the set $B=\{11,12,22,24\}$ and
$\rho(x,y)=11x+12y$, the subset sums are exactly $\rho([0,3]^2\cap\Z^2)$:
a filled square that projects to seven clusters separated by holes
(Figure~\ref{fig:holes-clusters}). The work of Chen, Hu, Mao, and
Zhang~\cite{CHMZDense} shows that the right analogue of a filled interval
is a filled lattice region inside the zonotope $Z(U)=\sum_{\vec u\in U}[\vec 0,\vec u]$.
Our algorithm lifts a scalar instance to such coordinates, after
removing a few exceptional items $E$ whose subset sums are transferred to
the target side:
\[
 \mathcal{S}(A)=\mathcal{S}(E)+\rho\bigl(\mathcal{S}(U)\bigr).
\]

\begin{figure}[t]
\centering
\begin{tikzpicture}[font=\small]
  \begin{scope}[shift={(1.2,2.4)},x=0.85cm,y=0.85cm]
    \fill[blue!6] (0,0) rectangle (3,3);
    \draw[step=1,gray!35,thin] (0,0) grid (3,3);
    \draw[black!60] (0,0) rectangle (3,3);
    \draw[->] (-0.3,0) -- (3.5,0) node[right] {$x$};
    \draw[->] (0,-0.3) -- (0,3.5) node[above] {$y$};
    \draw[orange!85!black,very thick] (0,3) -- (3,0);
    \foreach \x in {0,...,3} {
      \foreach \y in {0,...,3} {
        \pgfmathtruncatemacro{\k}{\x+\y}
        \ifnum\k=3
          \fill[orange!85!black] (\x,\y) circle (2.4pt);
        \else
          \fill[blue!75!black] (\x,\y) circle (2.1pt);
        \fi
      }
    }
    \node[below left] at (0,0) {$0$};
    \node[below] at (3,0) {$3$};
    \node[left] at (0,3) {$3$};
    \node[above] at (1.5,4.2) {(a) Filled lattice square};
  \end{scope}
  \node[align=center] at (8.7,4.7)
    {$B=\{11,12,22,24\}$\\[3pt]
     $\rho(\vec{z})=11x+12y$};
  \node[align=center,text width=6cm] at (8.7,3.45)
    {Every lattice point of $[0,3]^2$ is attainable.\\
     The orange diagonal $x+y=3$ maps to\\
     the central cluster $\{33,34,35,36\}$.};
  \draw[->,thick] (8.7,2.65) -- (8.7,1.65)
    node[midway,right] {linear projection};
  \begin{scope}[x=0.185cm,y=1cm]
    \draw[->,black!60] (-1,0.8) -- (71,0.8);
    \foreach \s in {0,...,69} {
      \draw[gray!45,thin] (\s,0.76) -- (\s,0.84);
    }
    \foreach \s in {0,11,12,22,23,24,45,46,47,57,58,69} {
      \fill[blue!75!black] (\s,0.8) circle (2.1pt);
    }
    \foreach \s in {33,34,35,36} {
      \fill[orange!85!black] (\s,0.8) circle (2.4pt);
    }
    \foreach \s in {0,11,22,33,45,57,69} {
      \node[below] at (\s,0.66) {\scriptsize $\s$};
    }
    \draw[<->,black!70] (25,1.12) -- (32,1.12)
      node[midway,above] {\scriptsize hole};
    \draw[<->,black!70] (37,1.12) -- (44,1.12)
      node[midway,above] {\scriptsize hole};
    \node[below] at (34.5,0.12)
      {(b) Clusters and holes in $\mathcal{S}(B)\subseteq[0,69]$};
  \end{scope}
\end{tikzpicture}
\caption{A filled higher-dimensional region can project to clusters separated
by holes, even in the middle of the scalar range. The four distinct items
$11,12,22,24$ generate all $16$ lattice points of $[0,3]^2$ in the lifted
coordinates, but their images under
$\rho(\vec{z})=11x+12y$ form seven clusters. Large dots mark attainable sums;
faint ticks mark all integers. The orange diagonal and orange cluster
correspond to the coordinate sum $x+y=3$. This example is exact and
illustrates projection-induced gaps; the general construction also retains
an uncertified boundary and exceptional-item shifts.}
\label{fig:holes-clusters}
\end{figure}
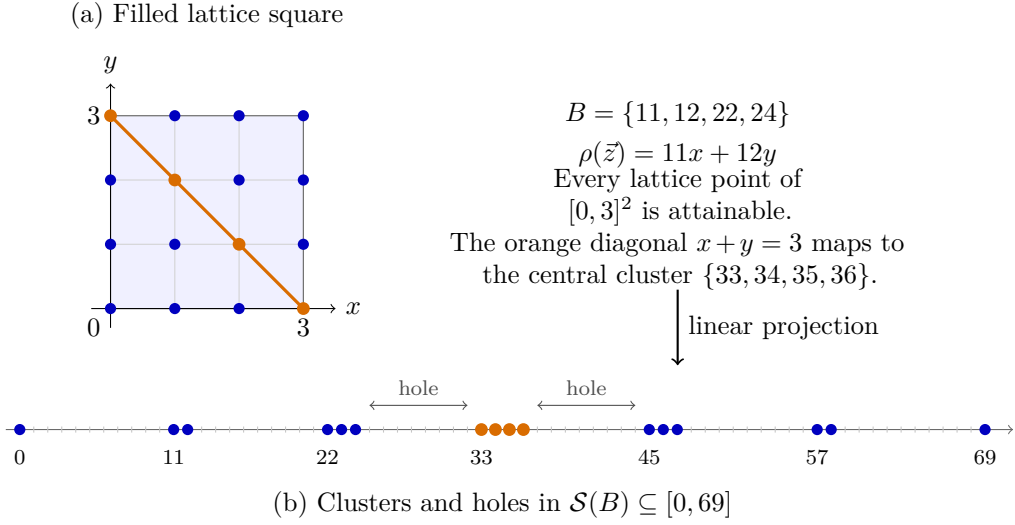

\paragraph{The recursive problem.}
An instance consists of integer vectors $V=\{\vec v_i:i\in I\}$ in a
symmetric box $P=\prod_{j=1}^d[-p_j,p_j]$ and an exact target set
$\mathcal T$ contained in a translate of $\tau\cdot P$; we ask whether
$\mathcal S(V)\cap\mathcal T\neq\varnothing$. The targets are stored as a
Boolean array whose \emph{allocated size}
$w$ counts the whole enclosing box. Every item is an
original integer, possibly with its selection complemented, and a linear
map recovers its oriented value. We start with $d=1$, $P=[-N,N]$,
$\mathcal T=\{t\}$, and $\tau_0=\lfloor N^\delta\rfloor$ for a small
$\delta=\delta(\varepsilon)$, so $w_0=O(N^{1+\delta})$. Throughout,
$N^{\delta}/\polylog N\le\tau\le N^\delta$.

Two exact tools act on items (Section~\ref{sec:boolean-operations}).
Removing $O(\tau)$ items $E$ and replacing $\mathcal T$ by
$\mathcal T-\mathcal S(E)$ preserves the answer at cost $\widetilde O(w)$.
Signed color coding computes, for signed vectors and coordinate budgets
$\Delta_j$, a set of actual subset sums containing every sum that has a
witness of total \emph{absolute} contribution at most $\Delta_j$ in every
coordinate, in time $\widetilde O(n+\prod_j(1+\Delta_j))$. Absolute
budgets bound every partial sum, so cancellations cause no difficulty.

Geometric queries split the target window into a constant number of cells:
cells outside the zonotope $Z(V)$ are discarded; cells with a buffer of
order $\tau\cdot P$ inside $Z(V)$ are \emph{interior}; the rest are near
the \emph{boundary}. Each recursive step either halves the total allocated
size (boundary) or adds one coordinate (interior), and the second kind of
step happens only $O_\delta(1)$ times along any branch.

\paragraph{A common GAP and a filling reserve (Section~\ref{sec:carrier}).}
For items bounded by $L$ and a parameter $h\ge L^{\eta}$, we construct, after removing
$O_\eta(h\log^2 L)$ exceptional items, coordinates
$\vec u_i\in C\cap\Z^r$ for all other items in a common GAP of rank
$r=O_\eta(1)$, and a reserve $R$ of $O_\eta(h)$ items whose subset sums
fill a translated box $\vec z_0+(ch\cdot C\cap\Z^r)$. The reserve corrects
rounding errors: an extreme fractional selection of the unreserved items
has at most $r$ fractional coefficients, and rounding them leaves an error
in $r\cdot C$, which a subset of $R$ supplies. Hence every lattice point of
\[
 J=\vec z_0+\sum_{i\notin R}[\vec 0,\vec u_i]
\]
is an actual subset sum (Lemma~\ref{lem:filled-zonotope}). For vector items
we encode them injectively as integers, and a linear map $F$ with
$F\vec u_i=\vec v_i$ recovers every old coordinate
(Lemma~\ref{lem:carrier-old-map}).

\paragraph{Boundary targets yield smaller instances (Section~\ref{sec:normal}).}
For a supporting direction $\vec a$ of the zonotope, the endpoint selection
maximizing $\vec a$ gives a reference subset. A target whose deficit
$h_Z(\vec a)-\vec a\cdot\vec t$ is small can be reached only by changing a
bounded total absolute amount in direction $\vec a$, because each changed
item lowers the functional by $|\vec a\cdot\vec v_i|$. With several
directions we choose, item by item, the reference bit minimizing the
combined normalized deficit. An item is \emph{heavy} if one of its
contributions exceeds the corresponding budget divided by $\tau$; a relevant
solution changes only $O_d(\tau)$ heavy items. Signed color coding computes
all relevant heavy-change sums $S_H$, we replace the targets by
$(\mathcal T-\vec b)-S_H$, where $\vec b$ is the reference sum, and we
delete \emph{all} heavy items. The light items then fit, at scale $\tau$,
in a box of the same shape as the residual targets.

The total size of the children is controlled by a packing argument. Around
each boundary center we use a Macbeath region, which adapts to the local
boundary; small disjoint copies lie in a boundary shell of relative
thickness $O_d(\vartheta)$, and fixed enlargements cover all boundary
candidates. Since lattice boxing loses only a dimension-dependent factor,
choosing the fixed constant $\vartheta$ small enough gives
$\sum w_{\mathrm{child}}\le w/2$.

Figure~\ref{fig:walkthrough-boundary} shows a small example:
$A=\{3,4,8,10,14,20,28\}$ with outlier $E=\{3\}$ and the remaining items
represented by $U=\{(1,0),(2,0),(0,1),(1,1),(0,2),(2,2)\}$ under
$\rho(x,y)=4x+10y$. Every lattice point of
$J=(1,1)+Z(\{(2,0),(0,2),(2,2)\})$ is a subset sum of $U$. For $t=57$ the
residual lines $\rho=57$ and $\rho=54$ contain no lattice point of $J$;
the only lattice candidate in $Z(U)$ is $(6,3)$, on the face $x=6$. The
endpoint selection maximizing $x$ chooses $4,8,14,28$, which sum to $54$;
relative to it the candidate has zero deficit, so only the items $10,20$
remain free, and $57=3+4+8+14+28$. For $t=45$, by contrast, the line
$\rho=42$ passes through the interior lattice point $(3,3)$ of $J$, giving
$45=3+14+28$.

\begin{figure}[tbp]
\centering
\begin{tikzpicture}[font=\small]
  \node[align=center] at (2.35,7.35)
    {Input after preprocessing\\distinct lifted items};
  \begin{scope}[shift={(1.1,5.1)},x=0.65cm,y=0.65cm]
    \draw[step=1,gray!25,thin] (0,0) grid (2,2);
    \draw[->,black!60] (-0.25,0) -- (2.65,0) node[right] {$x$};
    \draw[->,black!60] (0,-0.25) -- (0,2.45) node[left] {$y$};
    \foreach \x/\y in {1/0,2/0,0/1,0/2,1/1,2/2}
      \fill[blue!75!black] (\x,\y) circle (2.4pt);
  \end{scope}
  \node[align=center] at (2.35,4.5)
    {$\rho(U)=A\setminus E$\\$\rho(x,y)=4x+10y$};
  \node[draw=orange!80!black,fill=orange!8,rounded corners,
        align=center,inner sep=5pt] at (2.35,3.6)
    {Scalar outlier: $E=\{3\}$\\$3\notin\rho(\Z^2)$};
  \node[align=left,text width=4.5cm] at (2.35,2.3)
    {Target $t=57$ gives two lines:\\[3pt]
     \textcolor{orange!85!black}{solid: $\rho=54$} (use $3$),\\
     \textcolor{orange!85!black}{dashed: $\rho=57$} (use $0$).};
  \node[align=left,text width=4.5cm] at (2.35,0.9)
    {Boundary patch: nearby candidates.\\[3pt]
     Dashed box: allocated array.};

  \begin{scope}[shift={(5.7,1.45)},x=0.87cm,y=0.87cm]
    \fill[blue!4] (0,0)--(3,0)--(6,3)--(6,6)--(3,6)--(0,3)--cycle;
    \draw[step=1,gray!20,thin] (0,0) grid (6,6);
    \fill[green!14] (1,1)--(3,1)--(5,3)--(5,5)--(3,5)--(1,3)--cycle;
    \draw[green!45!black,thick]
      (1,1)--(3,1)--(5,3)--(5,5)--(3,5)--(1,3)--cycle;
    \draw[blue!65!black,very thick]
      (0,0)--(3,0)--(6,3)--(6,6)--(3,6)--(0,3)--cycle;
    \foreach \x in {0,...,6} {
      \foreach \y in {0,...,6} {
        \pgfmathtruncatemacro{\inside}{(\y>=\x-3)&&(\y<=\x+3)}
        \ifnum\inside=1 \fill[gray!65] (\x,\y) circle (1.6pt);\fi
      }
    }
    \draw[->,black!60] (-0.3,0) -- (6.55,0) node[right] {$x$};
    \draw[->,black!60] (0,-0.3) -- (0,6.55) node[above] {$y$};
    \node[below left] at (0,0) {$0$};
    \node[below] at (3,0) {$3$};
    \node[below] at (6,0) {$6$};
    \node[left] at (0,3) {$3$};
    \node[left] at (0,6) {$6$};
    \node[above,align=center] at (3,6.8)
      {Boundary case: $t=57$};
    \node[blue!65!black,fill=white,inner sep=1pt] at (5.65,5.72) {$Z$};
    \node[green!45!black,fill=green!14,inner sep=1pt] at (3.8,4.65) {$J$};
    \draw[orange!85!black,very thick] (0.2,5.32) -- (6.85,2.66);
    \draw[orange!85!black,thick,dashed] (0.2,5.62) -- (6.85,2.96);
    \draw[violet!80!black,thick] (6,3) ellipse [x radius=0.48,y radius=0.52];
    \draw[violet!80!black,dashed,thick] (5.45,2.42) rectangle (6.55,3.58);
    \draw[blue!65!black,densely dotted] (6,3.8) -- (6,6.4);
    \draw[->,blue!65!black,thick] (6,4.75) -- (7,4.75)
      node[above,midway] {$\vec n=(1,0)$};
    \fill[white] (6,3) circle (4pt);
    \fill[orange!85!black] (6,3) circle (3pt);
    \draw[violet!80!black] (6.55,3.58) -- (6.95,3.88);
    \node[violet!80!black,anchor=west,align=left] at (6.98,3.9)
      {enclosing\\box};
    \draw[violet!80!black] (6.28,2.6) -- (6.95,2.2);
    \node[violet!80!black,anchor=west] at (6.98,2.18) {patch};
    \node[below,align=center,text width=6cm] at (3,-0.55)
      {Only integer candidate: $(6,3)$.};
  \end{scope}
\end{tikzpicture}
\caption{A boundary target in the two-dimensional walkthrough. The six blue
points at left are the distinct input vectors; the gray points at right are
candidate subset-sum lattice points in the outer zonotope $Z$.
The green zonotope $J$ is a filled inner region in this exact toy example.
The target lines have real intersections with $J$, but neither contains an
integer point of $J$. The sole integer candidate is $(6,3)$, on the boundary
of $Z$. Section~\ref{sec:normal} groups boundary candidates into patches,
boxes each patch for its array computation, and uses a supporting normal
(here $\vec n=(1,0)$) to bound changes from an endpoint selection.}
\label{fig:walkthrough-boundary}
\end{figure}
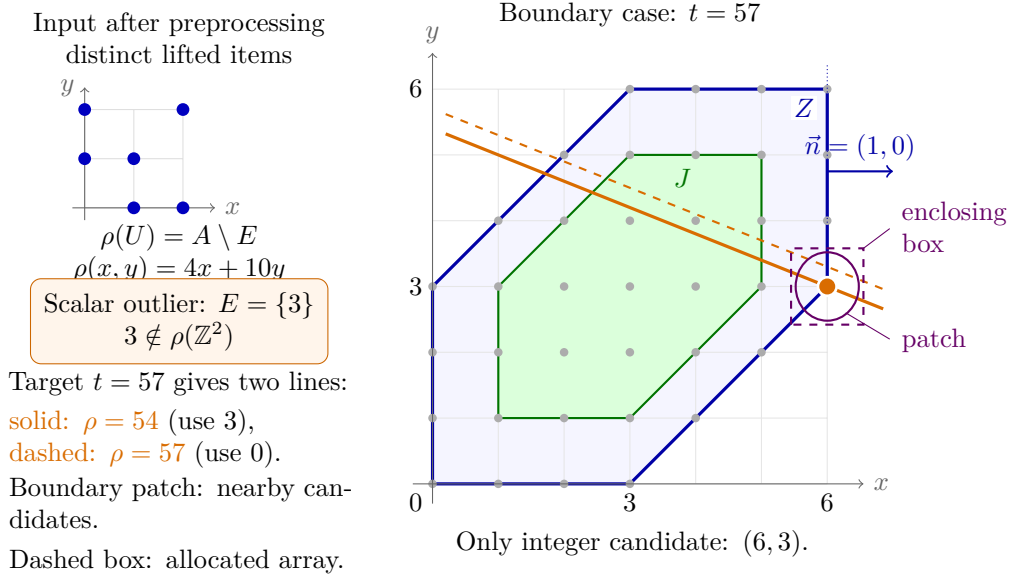

\paragraph{Interior targets are resolved or expose a new coordinate (Section~\ref{sec:deep}).}
Interior is relative to the \emph{current} $d$ coordinates, but the common
GAP may have rank $r>d$ and $F$ forgets the extra directions. Pick a target
$\vec y_0$ that passes the lattice filter. If the fiber
$J\cap F^{-1}(\vec y_0)$ contains a lattice point, the target is a subset
sum. Otherwise fixed-dimensional integer geometry provides an integer
functional $\psi$ of width $O_r(1)$ on the fiber
(Figure~\ref{fig:overview-core-fibers}). A concavity and reflection
argument extends this thin slice to a constant-width strip
$0\le\psi(\vec u)-\alpha-\vec\beta\cdot F\vec u\le O_r(1)$ over the whole
target patch, for every reserve choice. We remove the reserve and the $h$
items with largest prices $|\psi(\vec u_i)-\vec\beta\cdot\vec v_i|$, which
forces all surviving prices to be $O_r(1/h)$, and append $\psi$ as a new
coordinate, keeping \emph{every} permitted value of it
(Figure~\ref{fig:overview-fiber-lift}). The new instance has dimension
$d+1$, allocated size $O(w)$, and scale $\Theta(h)$.

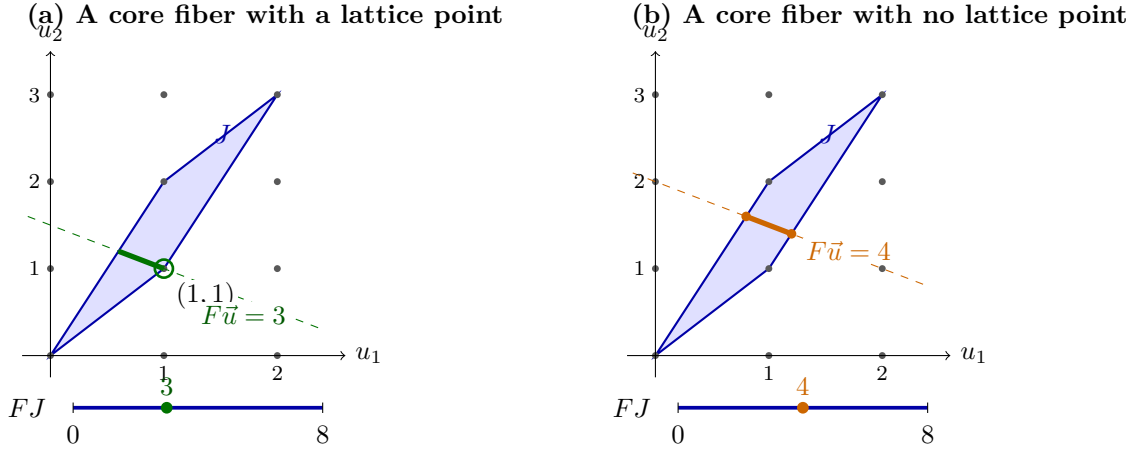
\begin{figure}[tbp]
\centering
\begin{tikzpicture}[font=\small,x=1.5cm,y=1.15cm]
 \begin{scope}
  \node[font=\small\bfseries,anchor=west] at (-.3,3.9)
    {(a) A core fiber with a lattice point};
  \fill[blue!12] (0,0)--(1,1)--(2,3)--(1,2)--cycle;
  \draw[blue!65!black,thick] (0,0)--(1,1)--(2,3)--(1,2)--cycle;
  \foreach \x in {0,1,2} {
   \foreach \y in {0,1,2,3} {
    \fill[black!65] (\x,\y) circle (1.3pt);
   }
  }
  \draw[->] (-.25,0)--(2.6,0) node[right] {$u_1$};
  \draw[->] (0,-.25)--(0,3.5) node[above] {$u_2$};
  \draw[green!45!black,dashed] (-.2,1.6)--(2.4,.3);
  \draw[green!45!black,line width=2pt] (.6,1.2)--(1,1);
  \draw[green!45!black,line width=1pt] (1,1) circle (3.4pt);
  \node[anchor=north west,fill=white,inner sep=1pt] at (1.08,.88)
    {$(1,1)$};
  \node[green!35!black,anchor=north west,fill=white,inner sep=1pt]
    at (1.3,.62) {$F\vec u=3$};
  \node[blue!65!black,anchor=west] at (1.35,2.55) {$J$};
  \foreach \x in {1,2} {\node[below] at (\x,0) {\scriptsize $\x$};}
  \foreach \y in {1,2,3} {\node[left] at (0,\y) {\scriptsize $\y$};}
  \node[anchor=east] at (.05,-.6) {$FJ$};
  \draw[blue!65!black,line width=1.4pt] (.2,-.6)--(2.4,-.6);
  \draw (.2,-.66)--(.2,-.54) node[below=5pt] {$0$};
  \draw (2.4,-.66)--(2.4,-.54) node[below=5pt] {$8$};
  \fill[green!45!black] (1.025,-.6) circle (2.2pt);
  \node[above,green!35!black] at (1.025,-.57) {$3$};
 \end{scope}
 \begin{scope}[xshift=8cm]
  \node[font=\small\bfseries,anchor=west] at (-.3,3.9)
    {(b) A core fiber with no lattice point};
  \fill[blue!12] (0,0)--(1,1)--(2,3)--(1,2)--cycle;
  \draw[blue!65!black,thick] (0,0)--(1,1)--(2,3)--(1,2)--cycle;
  \foreach \x in {0,1,2} {
   \foreach \y in {0,1,2,3} {
    \fill[black!65] (\x,\y) circle (1.3pt);
   }
  }
  \draw[->] (-.25,0)--(2.6,0) node[right] {$u_1$};
  \draw[->] (0,-.25)--(0,3.5) node[above] {$u_2$};
  \draw[orange!80!black,dashed] (-.2,2.1)--(2.4,.8);
  \draw[orange!80!black,line width=2pt] (.8,1.6)--(1.2,1.4);
  \fill[orange!80!black] (.8,1.6) circle (1.8pt);
  \fill[orange!80!black] (1.2,1.4) circle (1.8pt);
  \node[orange!80!black,anchor=north west,fill=white,inner sep=1pt]
    at (1.3,1.35) {$F\vec u=4$};
  \node[blue!65!black,anchor=west] at (1.35,2.55) {$J$};
  \foreach \x in {1,2} {\node[below] at (\x,0) {\scriptsize $\x$};}
  \foreach \y in {1,2,3} {\node[left] at (0,\y) {\scriptsize $\y$};}
  \node[anchor=east] at (.05,-.6) {$FJ$};
  \draw[blue!65!black,line width=1.4pt] (.2,-.6)--(2.4,-.6);
  \draw (.2,-.66)--(.2,-.54) node[below=5pt] {$0$};
  \draw (2.4,-.66)--(2.4,-.54) node[below=5pt] {$8$};
  \fill[orange!80!black] (1.3,-.6) circle (2.2pt);
  \node[above,orange!80!black] at (1.3,-.57) {$4$};
 \end{scope}
\end{tikzpicture}
\caption{The two outcomes of a core-fiber test, illustrated by the exact
lattice zonotope $J=[\vec{0},(1,1)]+[\vec{0},(1,2)]$ in $\Lambda=\Z^2$,
with $F(u_1,u_2)=u_1+2u_2$. Both targets lie strictly inside $FJ=[0,8]$.
In (a), $J\cap F^{-1}(3)$ contains $(1,1)$; when $J$ is a certified
core, such a point gives an actual subset sum. In (b),
$J\cap F^{-1}(4)$ joins $(4/5,8/5)$ to $(6/5,7/5)$ and contains no
lattice point. The small example illustrates the geometric alternatives,
rather than the quantitative reserve and buffer assumptions of the algorithm.}
\label{fig:overview-core-fibers}
\end{figure}
\begin{figure}[tbp]
\centering
\begin{tikzpicture}[font=\small,x=.85cm,y=1.2cm]
 \fill[orange!10]
  (-.35,{3*(-.35)/8-1/8})--(8.35,{3*8.35/8-1/8})
  --(8.35,{3*8.35/8+1/8})--(-.35,{3*(-.35)/8+1/8})--cycle;
 \draw[orange!80!black,dashed]
  (-.35,{3*(-.35)/8-1/8})--(8.35,{3*8.35/8-1/8});
 \draw[orange!80!black,dashed]
  (-.35,{3*(-.35)/8+1/8})--(8.35,{3*8.35/8+1/8});
 \fill[blue!18] (0,0)--(3,1)--(8,3)--(5,2)--cycle;
 \draw[blue!65!black,thick] (0,0)--(3,1)--(8,3)--(5,2)--cycle;
 \foreach \x in {0,1,2,3,4,5,6,7,8} {
  \foreach \y in {0,1,2,3} {
   \fill[black!65] (\x,\y) circle (1.2pt);
  }
 }
 \draw[->] (-.5,0)--(8.8,0) node[right] {$y=F\vec u$};
 \draw[->] (0,-.35)--(0,3.5) node[above] {$z=\psi(\vec u)$};
 \foreach \x in {2,4,6,8} {\node[below] at (\x,0) {\scriptsize $\x$};}
 \foreach \y in {1,2,3} {\node[left] at (0,\y) {\scriptsize $\y$};}
 \draw[orange!80!black,dashed] (4,.2)--(4,2.75);
 \draw[orange!80!black,line width=2pt] (4,1.4)--(4,1.6);
 \fill[orange!80!black] (4,1.4) circle (1.6pt);
 \fill[orange!80!black] (4,1.6) circle (1.6pt);
 \node[orange!80!black,above] at (4,2.75) {$y=4$};
 \node[blue!65!black,anchor=south east] at (7.9,2.9) {$(F,\psi)J$};
 \node[orange!80!black,anchor=north west] at (5.3,.6)
   {$|z-\tfrac38y|\le\tfrac18$};
\end{tikzpicture}
\caption{Recording the previously hidden direction in the same toy
example. With $\psi(u_1,u_2)=u_2$, the map $(F,\psi)$ is an invertible
integer change of coordinates. The old target equation $y=4$ remains;
the extra coordinate records position along its fiber. The image of
$J$ is contained in the slanted band $|z-(3/8)y|\le1/8$, and its slice
at $y=4$ has $7/5\le z\le8/5$. Thus the obstruction is invisible in
the old coordinate but visible in the new one. In the algorithm, a
dimension-dependent constant-width band retains every permitted integer
value of the added coordinate, followed by a child-lattice filter.}
\label{fig:overview-fiber-lift}
\end{figure}
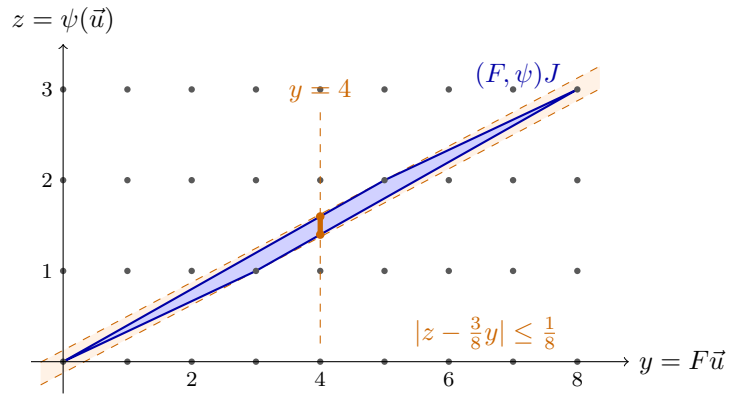

\paragraph{Why dimensions stay bounded (Section~\ref{sec:invariant}).}
Let $k=\lfloor N^{\delta/3}\rfloor$. All GAPs are proper on a box of scale
$\Omega(h)>2k$, so every representation preserves equalities between
signed sums of at most $k$ original integers, in both directions. Hence the
number $\mu(I)$ of distinct signed $k$-sums of the active original integers,
which is at most $2kN+1$, also counts the distinct short vector sums. If
$d$ current vectors are independent, then $\binom{k+d}{d}\le\mu(I)$, which
bounds $d$. If a deletion lowers the rank by $q$, the lost directions give
$\Omega(k^q)$ disjoint translates of the surviving signed half-sumset, and
the polynomial-scale doubling estimate~\eqref{eq:carrier-doubling} gives
$\mu(I')\le k^{-q/2}\mu(I)$; so the total rank loss is also $O_\delta(1)$.
Since each unresolved interior step adds a genuinely independent
coordinate, only $O_\delta(1)$ such steps occur on a branch.

\paragraph{Cost.}
Boundary children have total size at most $w/2$, and the $O_\delta(1)$
exceptional transitions per branch enlarge sizes by polylogarithmic
factors, so $\sum_\omega w_\omega=\widetilde O(w_0)$. The GAP construction
costs $\widetilde O(hw)$, giving total time
$\widetilde O(\tau_0w_0)=\widetilde O(N^{1+2\delta})$.

\paragraph{Multiplicities, vectors, and determinism (Section~\ref{sec:deterministic}).}
All steps act on distinct types with integer counts; only $O(\tau)$ copies
of a type can be changed at a boundary patch, so only that many are
instantiated. Vector input starts the same recursion in dimension $d$, with
the invariant stated for original vectors. The two randomized subroutines are
replaced by exact computations: boundary corrections involve $O(\tau)$
copies, so Newton's identities over a truncated polynomial ring compute
them exactly; and every constructive Freiman call lies in an interval of
length $O(hL)$, which permits a deterministic sieve for a
modelling multiplier.

\paragraph{Organization.}
Section~\ref{sec:preliminaries} collects notation, imported results, the
instance representation, and the exact Boolean operations.
Sections~\ref{sec:carrier}, \ref{sec:normal}, and~\ref{sec:deep} give the
common GAP, the boundary reduction, and the interior reduction.
Section~\ref{sec:invariant} bounds dimension changes and
Section~\ref{sec:recursion} assembles the recursion and proves
Theorem~\ref{thm:main}. Section~\ref{sec:deterministic} proves
Theorem~\ref{thm:main-vector}, Section~\ref{sec:john} proves the John-type
theorem, and Section~\ref{sec:ip-consequences} derives the consequences.
Standard Boolean, lattice, and bit-complexity details are in the appendices.

\section{Preliminaries}\label{sec:preliminaries}

\subsection{Notation}\label{sec:imported-results}

Vectors carry an arrow, as in $\vec v$, with scalar coordinates $v_j$;
$0$ also denotes the zero vector. Write $[a,b]_{\Z}=[a,b]\cap\Z$. For sets
$X,Y$ in an abelian group, $X\pm Y=\{x\pm y:x\in X,y\in Y\}$. For a positive
integer $k$, the \emph{repeated sumset} $kX=X+\cdots+X$ may reuse elements,
whereas $\lambda\cdot X=\{\lambda\vec x:\vec x\in X\}$ is geometric scaling
(for $X=\{0,1\}$, $2X=\{0,1,2\}$ but $2\cdot X=\{0,2\}$; for convex $X$ the
two agree). For a finite multiset $A$ of integers or vectors,
\[
 \Sigma(A)=\sum_{a\in A}a,
 \qquad \mathcal{S}(A)=\{\Sigma(B):B\subseteq A\},
 \qquad Z(A)=\sum_{a\in A}[0,a].
\]
Each occurrence is an \emph{item} (or \emph{copy}) and may be selected at most
once; $|A|$ counts items and $|\supp(A)|$ distinct values. An indexed family
in braces keeps every indexed occurrence. $A\uplus B$ adds multiplicities and
multiset difference removes the indicated copies. A box
$P=\prod_j[-p_j,p_j]$ means $|x_j|\le p_j$ in coordinate $j$, and for a box or convex set $B$ we write $|B|=\#(B\cap\Z^d)$ for its number
of integer points (for multisets, $|A|$ still counts items).
For a convex body $K$, its support function, width, and Macbeath regions are
\[
 h_K(\vec{u})=\max_{\vec{x}\in K}\vec{u}\cdot \vec{x},
 \qquad \operatorname{width}_{\vec{u}}(K)=h_K(\vec{u})+h_K(-\vec{u}),
 \qquad M_K^\lambda(\vec{x})=\vec{x}+\lambda\cdot\bigl((K-\vec{x})\cap(\vec{x}-K)\bigr).
\]

\paragraph{Progressions and lattices.}
A generalized arithmetic progression (GAP) of rank $r$ is a described set
\[
 Q=q_0+\rho(B\cap\Z^r),\qquad
 B=\prod_{j=1}^r[\ell_j,u_j],\qquad
 \rho(\vec{z})=\sum_{j=1}^r z_jq_j,
\]
in a torsion-free abelian group. Its \emph{coefficient volume} is
$|B\cap\Z^r|$. It is \emph{proper} if $\rho$ is injective on $B\cap\Z^r$;
we then also say that $B$ is proper under $\rho$. It is \emph{centred} if
$q_0=0\in B$ and \emph{symmetric} if moreover $B=-B$. For a centred
description put $Q_\lambda=\rho(\lambda\cdot B\cap\Z^r)$; the GAP is
\emph{$\lambda$-proper} if $Q_\lambda$ is proper. With integer endpoints,
$Q_k=kQ$. A \emph{containing GAP} for $X$ is a GAP $Q\supseteq X$.
A full-rank lattice $\Lambda\subset\R^r$ is given by a basis matrix $D$,
with $\det\Lambda=|\det D|$; lattices of smaller rank are considered within
their real span, and $\operatorname{span}_{\Z}X$ is the lattice generated by $X$.
Subscripts on $O(\cdot)$ list the parameters on which constants may depend;
likewise $\gamma$ with subscripts (e.g.\ $\gamma_d$, $\gamma_{d,{a_{\mathrm{buf}}}}$)
denotes a positive constant depending only on its subscripts, whose value may
change between occurrences. Numbered constants such as $\gamma_1,\gamma_2$ in
Section~\ref{sec:carrier} are fixed once defined.
\paragraph{Main symbols.}
Table~\ref{tab:notation} lists the symbols used in more than one section;
all other symbols are local to a section.

\begin{table}[ht]
\centering\small
\begin{tabular}{@{}lp{.7\linewidth}@{}}
\toprule
Symbol & Meaning \\
\midrule
$n$, $m$ & number of items (copies), number of distinct types\\
$N$, $\nu$ & input coordinate bound; $\nu=\max\{2,N_1,\ldots,N_d\}$ for vectors\\
$\varepsilon,\delta,\eta$ & target exponent slack; $\delta=\delta(\varepsilon)$; exponent of the GAP construction\\
$d$ & current dimension (in Section~\ref{sec:deterministic}: input dimension, current one $\bar d$)\\
$V=\{\vec v_i\}$, $c_i$ & current items, and their counts when stored as types\\
$a_i$ & oriented original value (vector) of item $i$; $\varphi(\vec v_i)=a_i$ (Sections~\ref{sec:representation}, \ref{sec:invariant}, \ref{sec:deterministic})\\
$P=\prod_j[-p_j,p_j]$ & item box; $|B|$ is the number of integer points of a box $B$\\
$\mathcal T$, $\vec o$, $\tau$, $w$ & target set $\subseteq\vec o+\tau\cdot P$; scale; allocated size $w=|\vec o+\tau\cdot P|$\\
$Z$, $h_K$ & zonotope of the current items; support function of a convex body $K$\\
$b_i$, $\vec b$, $\vec v_i'$ & reference bits, reference sum, signed items $(1-2b_i)\vec v_i$ (eq.~\eqref{eq:reanchor})\\
$h$, $L$ & GAP parameter $h=\lfloor\tau/\log^{c_{\mathrm{log}}}N\rfloor$; bound on the scalar inputs of the GAP construction\\
$r$, $C$, $\rho$, $\vec u_i$ & GAP rank, coefficient box, GAP map, coefficient vectors\\
$E$, $R$, $\vec z_0$ & exceptional items; reserve; centre of the filled box\\
$F$, $J$ & old-vector map ($F\vec u_i=\vec v_i$); certified core $\vec z_0+Z(U\setminus R)$\\
$k$, $\mu(I)$ & short-relation scale $\lfloor N^{\delta/3}\rfloor$ (or $\lfloor\nu^{\delta/3}\rfloor$); number of signed $k$-sums\\
$\vartheta$, $\varpi$, $\sigma$ & boundary threshold; boundary buffer $c_d\vartheta$; radius of an interior target patch\\
$c_{\mathrm{log}}$, $d_*$, $r_*$ & exponent in $h$; bounds on the rank and on exceptional transitions (Section~\ref{sec:invariant})\\
$\xi$ & failure allowance of randomized subroutines\\
\bottomrule
\end{tabular}
\caption{Main notation.}\label{tab:notation}
\end{table}

\subsection{Imported results}

\begin{importedtheorem}[Long sumsets; Conlon--Fox--Pham
  {\cite[Lemma 2.22 and Corollary 2.24]{CFP}}]
\label{thm:import-long-sumsets}
For every $\alpha>0$, there are constants $r_\alpha$ and
$c_\alpha>0$ such that the following holds for sufficiently large $L$.
If $X\subseteq[0,L]_{\Z}$ contains $0$ and $s$ is an integer with
$s\ge L^{\alpha}$, then there is a centred GAP $Q$ of rank at most
$r_\alpha$ containing $X$ such that $Q_{c_\alpha s}$ is proper and
\[
 z+Q_{c_\alpha s}\subseteq sX
 \qquad\text{for some }z\in\Z.
\]
\end{importedtheorem}

This is the polynomial-scale form of the cited result (take the parameter
$\beta>\max\{1,1/\alpha\}$ there and use the rank bound of its preceding
lemma); no density assumption on $X$ is needed.

\begin{importedtheorem}[Constructive Freiman theorem; Randolph--W\k{e}grzycki
  {\cite[Theorem 3.2, full version]{RW}}]
\label{thm:import-freiman}
Fix $K\ge1$ and $\gamma>0$. Given a set $X\subset\Z$ of $m\ge2$
elements satisfying $|X+X|\le K|X|$, a randomized algorithm uses
$\widetilde O_{K,\gamma}(m)$ arithmetic operations and, with probability
at least $1-m^{-\gamma}$, returns an explicit GAP $Q\supseteq X$ of rank
$O_K(1)$ and coefficient volume $O_K(m)$.
\end{importedtheorem}

\begin{importedtheorem}[Properization; Tao--Vu
  {\cite[Corollary 1.11]{TV}}]
\label{thm:import-properization}
Fix $r$ and $t\ge1$. Every symmetric GAP $P$ of rank at most $r$
in a torsion-free abelian group is contained in a symmetric $t$-proper
GAP $Q$ of rank at most $r$. If $P$ has coefficient box $B$, then
$Q$ has coefficient volume $O_{r,t}(|B|)$.
\end{importedtheorem}

Lemma~\ref{lem:carrier-properization} below implements this statement with
polynomial bit complexity, following the quotient argument
of~\cite[Sections 3--4]{TV}.

\begin{importedlemma}[Dense-box filling; Conlon--Fox--Pham
  {\cite[Lemma 2.15]{CFP}}]
\label{lem:import-filling}
For every fixed $r\ge1$ and $0<c\le1$, there are constants
$\gamma>0$, $\ell_0$, and $w_0$ with the following property.
Let $B=\prod_{j=1}^r[a_j,b_j]$ have integer endpoints and
$b_j-a_j+1\ge w_0$ for every $j$. Suppose that $\ell\ge\ell_0$ and
$A_1,\ldots,A_\ell\subseteq B\cap\Z^r$ satisfy, for every $i$,
\[
 |A_i|\ge c|B\cap\Z^r|,\qquad 0\in A_i,
 \qquad \operatorname{span}_{\Z}A_i=\Z^r.
\]
Then, for some $\vec{z}\in\Z^r$,
\[
 \vec{z}+\bigl(\gamma\ell\cdot B\cap\Z^r\bigr)
 \subseteq A_1+\cdots+A_\ell.
\]
\end{importedlemma}

The source assumes that each $A_i$ is \emph{reduced}: its image modulo every
subgroup $t_1\Z\times\cdots\times t_r\Z$ lies in no coset of a proper
subgroup. This follows from $0\in A_i$ and generation: such a coset would
contain $0$, hence be a proper subgroup containing all generators.

\begin{importedlemma}[Greedy growth; Conlon--Fox--Pham
  {\cite[Lemma 2.41]{CFP}}]
\label{lem:import-greedy}
Let $S,U\subset\Z^r$ be finite nonempty sets and let $k\ge1$ be an
integer. If $|kU|\ge2|S|$, then some $\vec{u}\in U$ satisfies
\[
 |(S+\vec{u})\setminus S|\ge\frac{|S|}{2k}.
\]
\end{importedlemma}

The source states it for integers; the proof (subadditivity of
$\vec a\mapsto|(S+\vec a)\setminus S|$ and counting representations in
$S-S$) applies verbatim in $\Z^r$, and repetitions in $kU$ are allowed.

\paragraph{Geometric and lattice tools.}
We use the following oracle algorithms only in fixed dimension. A rational
polytope has encoding bound $b$ if its vertex coordinates have numerators
and denominators of at most $b$ bits. A strong separation oracle either
certifies membership of a rational point or returns a strictly separating
rational hyperplane.

\begin{importedtheorem}[Convex optimization and integer feasibility;
  Gr\"otschel--Lov\'asz--Schrijver
  {\cite[Section 1 and Theorems 2.13--2.14]{GLS}}]
\label{thm:import-geometry}
Fix $r$. Given a bounded rational polytope $K\subset\R^r$ by a
strong separation oracle and an encoding bound, the following tasks
require polynomially many bit operations and oracle calls in the
encoding bound and the rational input lengths.
\begin{enumerate}[label=(\roman*)]
\item Decide emptiness, determine the affine hull, and optimize a
given rational linear objective. Strong separation and optimization
are polynomial-time equivalent under this representation.
\item If $K$ is full-dimensional and nonempty, find a rational point
$\vec{x}$ and an origin-centred ellipsoid $E$ such that
$\vec{x}+E\subseteq K\subseteq \vec{x}+\gamma _r\cdot E$, where $\gamma _r$ depends only on $r$.
\item If $K$ is full-dimensional and nonempty, either find a point of
$K\cap\Z^r$ or find a nonzero $\vec{u}\in\Z^r$ with
$\operatorname{width}_{\vec{u}}(K)\le f_r$, where $f_r$ depends only on $r$.
\end{enumerate}
All returned rational data have polynomial encoding length. The same
statements hold relative to a full-rank rational lattice supplied by a
basis, after changing to lattice coordinates.
\end{importedtheorem}

The thin-direction outcome of (iii) does not certify the absence of an
integer point. Radius bounds and oracle costs in our applications are
established in Lemma~\ref{lem:bits}.

\begin{importedtheorem}[Lattice basis reduction;
  Gr\"otschel--Lov\'asz--Schrijver
  {\cite[Theorems 2.9--2.10]{GLS}}]
\label{thm:import-lattice-reduction}
Given a rational basis of a full-rank lattice $\Lambda\subset\R^r$,
one can compute in polynomial bit time a basis $\vec{b}_1,\ldots,\vec{b}_r$ of
the same lattice satisfying
\[
 \prod_{j=1}^r\|\vec{b}_j\|_2
 \le 2^{r(r-1)/4}\det\Lambda.
\]
\end{importedtheorem}

\begin{importedlemma}[Macbeath overlap;
  {\cite[Lemma 2.2]{AFM}}]
\label{lem:import-macbeath}
Let $K\subset\R^r$ be a convex body and let $\vec{x},\vec{y}\in K$. If
$M_K^{1/5}(\vec{x})\cap M_K^{1/5}(\vec{y})\ne\varnothing$,
then $M_K^{1/5}(\vec{y})\subseteq M_K^1(\vec{x})$.
\end{importedlemma}

This lemma is due to Ewald, Larman, and Rogers. Consequently a maximal
family of pairwise disjoint small Macbeath regions covers all others after
a fixed enlargement.

\subsection{Recursive instances}\label{sec:representation}
Every recursive call asks whether $\mathcal{S}(V)\cap\mathcal T\ne\varnothing$
for a multiset of integer vectors $V=\{\vec v_i:i\in I\}\subseteq\Z^d$ and
an exact target set $\mathcal T$ stored as a Boolean array:
\[
 P=\prod_{j=1}^d[-p_j,p_j],\qquad p_j\in\Z_{\ge1},\qquad
 \vec v_i\in P,\qquad \mathcal T\subseteq\vec o+\tau\cdot P,
 \qquad \tau\ge1.
\]
The array has one entry for every integer point of the enclosing box, so its
\emph{allocated size} is
\[
 w=|\vec o+\tau\cdot P|=\Theta_d\Bigl(\prod_{j=1}^d(1+\tau p_j)\Bigr).
\]
Thus $w$ counts zero entries too; the one entries specify the decision problem.

\paragraph{Reanchoring.}
Fix reference bits $b_i\in\{0,1\}$ and write a selection as
$x_i=b_i+(1-2b_i)y_i$ with flip bits $y_i\in\{0,1\}$. Then
\begin{equation}\label{eq:reanchor}
 \sum_i x_i\vec v_i=\sum_i b_i\vec v_i+\sum_i y_i(1-2b_i)\vec v_i.
\end{equation}
With the \emph{reference sum} $\vec b=\sum_ib_i\vec v_i$ and the \emph{signed
items} $\vec v_i'=(1-2b_i)\vec v_i$, replacing each item by $\vec v_i'$ and
translating $\mathcal T$ by $-\vec b$ is therefore an exact change of variables: a flip adds
item $i$ if $b_i=0$ and removes it if $b_i=1$. The reference is an actual
subset, but its sum need not be a target.

\paragraph{Original values.}
In the set recursion each active item keeps its \emph{oriented original
value} $a_i\in[-N,N]$: the input integer, with its sign flipped each time the
item is reanchored with $b_i=1$. The reductions preserve a linear map
$\varphi(\vec z)=\vec\alpha\cdot\vec z$ with $\varphi(\vec v_i)=a_i$.
It is recomputed from a basis of current items, never as a product of
earlier coordinate changes; it need not be injective, so we always keep
every current coordinate. If two current vectors coincide, their oriented
values coincide, so they come from $a$ and $-a$; hence vector multiplicity
is at most two and
\begin{equation}\label{eq:nscan}
 |I|\le2|P|=O_d(w).
\end{equation}

\paragraph{Span restriction.}
We always restrict to the real span of the active vectors: targets outside it
are discarded, and projection onto a suitable set of integer coordinates is
injective on the span and preserves the exact problem. Before projecting we
must discard off-span targets, since an off-span target can project onto an
attainable point. When we pass to a basis of a sublattice, convolutions
already performed are charged in the lattice where they were performed.

\paragraph{Parameters.}
The set recursion starts with $d=1$, $P=[-N,N]$, $\tau_0=\lfloor N^\delta\rfloor$,
and uses $h=\lfloor\tau/\log^{c_{\mathrm{log}}}(N+2)\rfloor$.
Lemma~\ref{lem:bootstrap} proves that, for sufficiently large $N$, every
nonterminal call has
\[
 d=O_\delta(1),\qquad
 \frac{N^\delta}{\log^{O_\delta(1)}N}\le h\le\tau\le N^\delta,
 \qquad w\le N^{1+2\delta}.
\]
We reserve $\vartheta$ for the small fixed threshold of the boundary
reduction.

\subsection{Exact Boolean operations and signed color coding}
\label{sec:boolean-operations}

Fix a multiset $W=\{\vec w_1,\ldots,\vec w_n\}\subseteq\Z^d$ and budgets
$\Delta_j\ge0$. A subset $X\subseteq[n]$ is \emph{absolute-bounded} if
\begin{equation}\label{eq:absolutequota}
 \sum_{i\in X}|w_{ij}|\le\Delta_j\qquad(j=1,\ldots,d).
\end{equation}
Then every subset of $X$ has its sum in $B=\prod_j[-\Delta_j,\Delta_j]$, which
permits truncating intermediate arrays. Let $\mathcal S_{\mathrm{abs}}$ be the
set of sums of absolute-bounded subsets.

\begin{lemma}[Signed absolute-budget sums]\label{lem:quota}
For $0<\xi<1$, a randomized algorithm computes a set
$S\subseteq\mathcal{S}(W)\cap B$ in
\[
 \widetilde O_d(n+|B|)
\]
time, with polylogarithmic dependence on $1/\xi$, such that
$\Pr[\mathcal S_{\mathrm{abs}}\subseteq S]\ge1-\xi$.
\end{lemma}

Every reported sum has an actual witness; $S$ may contain other subset sums
in $B$. The proof (Appendix~\ref{app:boolean-operations}) combines a balanced
Boolean product tree (Lemma~\ref{lem:boolean-product}) with the two-stage
color coding of~\cite{Bringmann}.

\begin{lemma}[Removing a subset of items]\label{lem:transfer}
Let $E\subseteq I$ contain $e$ items in
$P$, and let $\cT\subseteq \vec{o}+\tau\cdot P$.
The exact target set $\cT-\mathcal{S}(\{\vec{v}_i:i\in E\})$ can be computed in
$\widetilde O_d(e+|(\tau+e)\cdot P|)$
time, which is $\widetilde O_d(e+w)$ if $e=O(\tau)$. Replacing the target set
by this difference and removing $E$ preserves the answer.
\end{lemma}

\begin{proof}
Lemma~\ref{lem:boolean-product} on the factors $\{0,\vec v_i\}$, $i\in E$,
computes $\mathcal S(\{\vec v_i:i\in E\})\subseteq e\cdot P$; one convolution
with the target indicator gives the difference inside
$\vec o+(\tau+e)\cdot P$. Disjointness of $E$ and $I\setminus E$ gives the
equivalence.
\end{proof}

\begin{lemma}[Small residual instance]\label{lem:terminal}
An instance with $n\le\tau$ items can be solved exactly in
$\widetilde O_d(n+w)$ time.
\end{lemma}

\begin{proof}
All subset sums lie in $n\cdot P\subseteq\tau\cdot P$; compute them by
Lemma~\ref{lem:boolean-product} and intersect with $\mathcal T$.
\end{proof}

\section{A common GAP and a filling reserve}\label{sec:carrier}

This section extracts additive structure from the items. After removing a
small exceptional multiset $E$, all remaining items receive coordinates in
one bounded-rank GAP, and a small reserve $R$ is chosen whose subset sums
fill a box in those coordinates. We state the construction directly for
multisets stored by counts: the set recursion produces multiplicities up to
two anyway, and the John-type theorem of Section~\ref{sec:john} needs
arbitrary multiplicities.

\subsection{Growth of repeated sumsets}
Let $X\subseteq[0,L]$ contain $0$, and let $s\ge L^{\alpha}$.
Theorem~\ref{thm:import-long-sumsets} gives a GAP $Q\supseteq X$ of rank
$r\le r_\alpha$ with coefficient box $B$, such that $Q_{c_\alpha s}$ is proper
and a translate of it lies in $sX$. For $t\ge s$, $tX\subseteq Q_t$ and, after
discarding zero-width coordinates, $|t\cdot B\cap\Z^r|=O_\alpha((t/s)^{r})\,|c_\alpha s\cdot B\cap\Z^r|$.
Hence, for fixed $b\ge1$ and with ${a_{\mathrm{gr}}}=r_\alpha$,
\begin{equation}\label{eq:carrier-doubling}
 |bsX|=O_{\alpha,b}(|sX|),
\end{equation}
\begin{equation}\label{eq:carrier-polynomial-growth}
 |sX|\ge c(s/t)^{a_{\mathrm{gr}}}|tX|\qquad(L^{\alpha}\le s\le t),
\end{equation}
for a constant $c=c(\alpha)>0$. No lower bound on $|X|$ is needed.

\subsection{Statement}
For multiplicities $\vec y=(y_a)$ on distinct integers $a$, write
$X(\vec y)=\{0\}\cup\{a:y_a>0\}$ for the support and
$\mathcal S(\vec y)=\{\sum_a j_aa:0\le j_a\le y_a\}$ for the subset sums.
Repeated sumsets $sX(\vec y)$ ignore multiplicities; they serve only as a
progress measure. Every actual selection respects the counts.

\begin{theorem}[Constructing a common GAP]\label{thm:carrier}
Fix $\eta>0$. The input consists of $m$ distinct positive integers
$a\le L$ with multiplicities $c_a\in\Z_{\ge1}$, and an integer
$h\ge L^{\eta}$. In
\begin{equation}\label{ms:time}
 \softO{mh+h^2L}
\end{equation}
time, an algorithm whose only randomized steps are calls of
Theorem~\ref{thm:import-freiman} computes exception multiplicities
$0\le e_a\le c_a$ with
\begin{equation}\label{ms:exceptions}
 \sum_a e_a=O_{\eta}\bigl(h\log^2(hL+2)\bigr),
\end{equation}
and, unless all residual multiplicities $c_a-e_a$ vanish:
\begin{enumerate}[label=(\arabic*)]
\item \textbf{Coordinates.} A rank $r=O_\eta(1)$, a box
$C=\prod_{j=1}^r[-b_j,b_j]$ with $b_j\in\Z_{\ge1}$, an integer homomorphism
$\rho:\Z^r\to\Z$, and $\vec{u}_a\in C\cap\Z^r$ with $\rho(\vec{u}_a)=a$ for
every $a$ with $c_a-e_a>0$.
\item \textbf{Reserve.} Multiplicities $0\le r_a\le c_a-e_a$ with
$\sum_a r_a=O_\eta(h)$, a point $\vec z_0\in\Z^r$, and a constant
$\gamma=\gamma(\eta)>0$ such that $\rho$ is injective on $\gamma h\cdot C\cap\Z^r$ and
\begin{equation}\label{ms:filling}
 \vec{z}_0+(\gamma h\cdot C\cap\Z^r)
 \ \subseteq\
 \Bigl\{\sum_a j_a \vec{u}_a:0\le j_a\le r_a\Bigr\},
\end{equation}
together with an explicit reserve witness for every point of the filled box.
The reserve vectors generate $\Z^r$.
\item \textbf{Remaining items.} The vectors $\vec u_a$ with
$c_a-e_a-r_a>0$ generate $\Z^r$, and they still span $\R^r$ after the
deletion of any further $h$ copies.
\end{enumerate}
For any fixed $b>0$ the output is correct with probability at least
$1-(hL+2)^{-b}$, within the same time bound.
\end{theorem}

The time bound counts word operations; $m$ is the number of distinct
values, so copies are never expanded. Condition~(3) is a quantitative
non-degeneracy condition, of the same form as in~\cite[Definitions 18
and 49]{CHMZDense}: for every hyperplane $\mathcal H\ni0$ more than $h$ remaining
copies lie outside $\mathcal H$.

The reserve is used as follows (Lemma~\ref{lem:filled-zonotope}). Take a
lattice point $\vec z$ of the zonotope $\vec z_0+Z(U\setminus R)$ and an
extreme fractional representation of $\vec z-\vec z_0$ by the unreserved
items. At most $r$ coefficients are fractional, and rounding them leaves an
error $\vec e\in r\cdot C\cap\Z^r$; once $\gamma h\ge r$ a reserve subset
sums to $\vec z_0+\vec e$, and the two disjoint selections sum to $\vec z$.

\subsection{Lattice tools}

\begin{lemma}[Constructive full-lattice boxing; cf.\ Tao--Vu and GLS]
\label{lem:carrier-full-lattice-box}
Fix $r\ge1$ and a constant bound on polytope complexity. Let
$K\subseteq\R^r$ be an origin-symmetric rational polytope, and let
$\Lambda$ be a full-rank rational lattice supplied by a basis, such that
$K$ contains $r$ linearly independent points of $\Lambda$. In polynomial
time in the description bit length one can compute a basis
$\vec{b}_1,\ldots,\vec{b}_r$ of $\Lambda$ and integers $n_j\ge1$ such that
$C=\{\sum_j x_j\vec{b}_j:|x_j|\le n_j\}$ satisfies
\[
 K\subseteq C\subseteq \gamma_r\cdot K,
 \qquad |a\cdot C\cap\Lambda|\le\gamma_r|a\cdot K\cap\Lambda|
 \quad\text{for every real }a>0,
\]
where $\gamma_r$ depends only on $r$.
\end{lemma}

\begin{lemma}[Fixed-rank constructive properization]\label{lem:carrier-properization}
Fix $r\ge1$ and $t\ge1$. From a symmetric integer GAP $P$ of rank at most $r$
and coefficient box $B$, one can construct a
symmetric $t$-proper GAP $Q\supseteq P$ of rank at most $r$ and
coefficient volume $O_{r,t}(|B|)$, in time polynomial in the bit
length of the description of $P$.
\end{lemma}

The first lemma is a constructive form of the discrete John theorem
\cite[Theorem 1.6]{TV} using lattice reduction~\cite[Theorems 2.9--2.10]{GLS};
the second implements Theorem~\ref{thm:import-properization}. Both proofs are
in Appendix~\ref{app:lattice-tools}.

\subsection{Proof of Theorem~\ref{thm:carrier}}
All constants depend on $\eta$ (and on $b$ for the error probability).
Put $s=a_\eta h$, where $a_\eta$ is a sufficiently large integer chosen after
all other constants; we may assume $sL\ge2$, and since $h\ge L^\eta$,
$\log(sL)=O_\eta(\log s)$. The inflation
$s=a_\eta h$ makes every numerical threshold below hold for all values of
$h,L$, so we never discard the whole (possibly huge) multiset.
We maintain live multiplicities $y_a$, initially $c_a$, and exception
multiplicities $e_a$, initially $0$. The construction runs in phases; each
unsuccessful phase halves the quantity $|sX(\vec y)|$ and puts
$O(s\log(sL))$ copies into $E$.

\paragraph{Growth.}
Fix $\alpha<\min(1/2,\eta/2)$ and enlarge $a_\eta$ so that
$(L+1)^\alpha\le\lceil\sqrt s\rceil$. Then
\eqref{eq:carrier-doubling}--\eqref{eq:carrier-polynomial-growth} with
this $\alpha$ give constants ${a_{\mathrm{gr}}},c_*>0$ such that for every
$Y\subseteq[0,L]$ with $0\in Y$,
\begin{equation}\label{ms:growth}
 |qY|\ge c_*(q/s)^{a_{\mathrm{gr}}}|sY|\quad(\lceil\sqrt{s}\rceil\le q\le s),
 \qquad |20sY|=O_\eta(|sY|).
\end{equation}

\paragraph{The box of a phase.}
If all $y_a=0$, stop with the empty outcome. Otherwise put $X=X(\vec y)$ and
${m_{\mathrm{ph}}}=|sX|\in[s+1,sL+1]$. Compute $10sX$ by Boolean repeated
squaring; it has bounded doubling and size $O_\eta({m_{\mathrm{ph}}})$ by
\eqref{ms:growth}. Theorem~\ref{thm:import-freiman} gives a containing GAP
of bounded rank and volume $O_\eta({m_{\mathrm{ph}}})$; symmetrize it (difference with
itself) and make it $2$-proper by Lemma~\ref{lem:carrier-properization}.
Write it as $\rho_0(B\cap\Z^r)$ with $B=\prod_j[-\beta_j,\beta_j]$.

Let $\vec u_a$ be the coordinate vector of a live value $a$ (found by enumerating
the $O_\eta({m_{\mathrm{ph}}})$ points of the $2$-proper GAP once and looking up each
value). Since $ja\in10sX$
for $0\le j\le10s$, the identities $(j-1)a+a=ja+0$ and $2$-properness give
inductively that $ja$ has coordinates $j\vec u_a$. Hence every $\vec u_a$
lies in
\[
 C_0=\prod_j\bigl[-\lfloor \beta _j/(10s)\rfloor,\lfloor \beta _j/(10s)\rfloor\bigr],
\]
and we delete zero-width coordinates, so $C_0$ has integer half-widths
$b^{(0)}_j\ge1$. As $10s\cdot C_0\subseteq B$, the map $\rho_0$ is injective
on $10s\cdot C_0\cap\Z^r$; in particular repeated sums of at most $10s$
live vectors are determined by their scalar values. Moreover
$s\cdot C_0\subseteq10s\cdot C_0\subseteq B$, so
\begin{equation}\label{ms:phase-volume}
 |s\cdot C_0\cap\Z^r|\le|B\cap\Z^r|\le \gamma_1 {m_{\mathrm{ph}}},
\end{equation}
where $\gamma_1=\gamma_1(\eta)$ is the product of the volume factor of
Theorem~\ref{thm:import-freiman} (applied to $10sX$, whose size and doubling are
$O_\eta({m_{\mathrm{ph}}})$ and $O_\eta(1)$ by \eqref{ms:growth}), the factor $2^r$ lost in
symmetrizing, and the factor of Lemma~\ref{lem:carrier-properization} with $t=2$;
it depends only on $\eta$ because $r=O_\eta(1)$. The constant $a_\eta$ is chosen
after $\gamma_1$, so $2s+1>2\gamma_1$. Every coordinate width of $s\cdot C_0$ is at
least $2s+1$, so a lattice hyperplane through $0$ meets $s\cdot C_0\cap\Z^r$ in at
most $\gamma_1{m_{\mathrm{ph}}}/(2s+1)<{m_{\mathrm{ph}}}/2$ points. Consequently the live vectors
span $\R^r$ (otherwise $sX$ would correspond to such a hyperplane section),
and $\prod_jb^{(0)}_j\le s^{-r}|s\cdot C_0\cap\Z^r|\le\gamma_1{m_{\mathrm{ph}}}$.

\paragraph{Protecting the rank.}
Let ${k_{\mathrm{pack}}}=\lceil \gamma_{\mathrm{pack}}s\log(sL)\rceil$, where
\[
 \gamma_{\mathrm{pack}}=b_*\ell\gamma_2+2(2b_*\ell+1)
\]
and $b_*$, $\ell$, $\gamma_2$ are the constants introduced in the next three
paragraphs (the block-size bound, the number of blocks fed to
Lemma~\ref{lem:import-filling}, and the bound on lattice changes). They depend only
on $\eta$ and not on $\gamma_{\mathrm{pack}}$, so this is not circular. On temporary counters, greedily pack
${k_{\mathrm{pack}}}+1$ disjoint real bases of $r$ copies each (scan the types with
positive temporary count, adding a vector when it enlarges the span, and
debit one copy). If only $j\le{k_{\mathrm{pack}}}$ complete bases are found, move their
$rj$ copies into $E$. The types that remain live are exactly those with
positive temporary count after the $j$ complete bases; the failed scan found
a maximal independent subset of them of size less than $r$, so they span a
proper subspace $\mathcal H$, and
\begin{equation}\label{ms:rank-restart}
 |sX(\vec{y}_{\rm new})|
 \le |s\cdot C_0\cap \mathcal H\cap\Z^r|<{m_{\mathrm{ph}}}/2,
\end{equation}
and we restart. Otherwise restore the counters: deleting any ${k_{\mathrm{pack}}}$
copies now leaves a complete basis.

\paragraph{One block.}
For the rest of the phase, after every removal of one copy we test whether
$|sX(\vec y)|<{m_{\mathrm{ph}}}/2$; if so, all copies removed in the phase go to $E$
and we restart. (The support changes only when a type dies, so
$sX(\vec y)$ is recomputed only then.)
At the start of a block let $\Lambda$ be the lattice generated by the live
vectors; it has full rank while at most ${k_{\mathrm{pack}}}$ copies have been removed.
A \emph{seed} consists of one copy of each of $r$ independent live types,
followed by further live types outside the current generated lattice until
it equals $\Lambda$. The first determinant is at most
$r!\prod_jb_j^{(0)}=O_\eta({m_{\mathrm{ph}}})$ and each addition halves the index,
so the seed has $O_\eta(\log(sL))$ copies. Remove the seed and let $S$
be its set of subset sums in GAP coordinates. While $|S|<{m_{\mathrm{ph}}}/4$, remove
one copy of a live type maximizing $|(S+\vec u_a)\setminus S|$ and update
$S\leftarrow S\cup(S+\vec u_a)$.

We claim the block has at most $b_*s$ copies, where
\[
 b_*=\frac{2\ln2\cdot2^{-1/{a_{\mathrm{gr}}}}}{1-2^{-1/{a_{\mathrm{gr}}}}}+2\ln\frac{2}{c_*}+1
\]
depends only on $\eta$ (through the constants ${a_{\mathrm{gr}}},c_*$ of \eqref{ms:growth}). Apply
Lemma~\ref{lem:import-greedy} with $U=\{0\}\cup\{\vec u_a:y_a>0\}$; a positive
increment is attained only by a live type. Since $s\cdot C_0$ is proper,
$|qU|=|qX(\vec y)|$, and by \eqref{ms:growth} and the restart rule,
\begin{equation}\label{ms:residual-growth}
 |qU|\ge c_*(q/s)^{a_{\mathrm{gr}}}{m_{\mathrm{ph}}}/2\qquad(q_0=\lceil\sqrt{s}\rceil\le q\le s).
\end{equation}
Write $z=|S|$. While $z\le(c_*/4){m_{\mathrm{ph}}}(q_0/s)^{a_{\mathrm{gr}}}$, use $q=q_0$: each
step multiplies $z$ by at least $1+1/(2q_0)$, so this range takes
$O(q_0\log{m_{\mathrm{ph}}})=O(\sqrt s\,\log(sL))=O_\eta(s)$ steps. While
$z\le c_*{m_{\mathrm{ph}}}/8$, choose $q=O_\eta(s(z/{m_{\mathrm{ph}}})^{1/{a_{\mathrm{gr}}}})$ in $[q_0,s]$
with right side of \eqref{ms:residual-growth} at least $2z$; doubling $z$
from $z_*$ takes $O_\eta(s(z_*/{m_{\mathrm{ph}}})^{1/{a_{\mathrm{gr}}}})$ steps, a geometric series
of total $O_\eta(s)$. The final constant-factor range up to ${m_{\mathrm{ph}}}/4$ uses
$q=s$ and $|sU|\ge{m_{\mathrm{ph}}}/2$, taking at most $2s\ln(2/c_*)+1$ steps. The
middle range takes at most $2\ln2\cdot2^{-1/{a_{\mathrm{gr}}}}s/(1-2^{-1/{a_{\mathrm{gr}}}})$ steps plus
$O(\log(sL))$, since its largest $q$ is at most $2^{-1/{a_{\mathrm{gr}}}}s$; the first range and
the seed take $O(\sqrt s\log(sL))$ copies. As $\log(sL)=O_\eta(\log s)$ and
$a_\eta$ is chosen after $b_*$, these lower-order terms are at most $s$, which
proves the claim. The completed block
support contains $0$, generates $\Lambda$, has at least ${m_{\mathrm{ph}}}/4$ points, and
lies in $b_*s\cdot C_0$; its sums use distinct copies.

\paragraph{Blocks with a common lattice.}
The lattices $\Lambda_j$ remaining after $j$ blocks form a descending chain
of full-rank lattices of covolume at most $r!\prod_jb^{(0)}_j\le r!\,\gamma_1{m_{\mathrm{ph}}}$
(each contains an independent set from $C_0$), and a strict inclusion at least
doubles the covolume. So there are at most
$\log_2(r!\,\gamma_1(sL+1))\le\gamma_2\log(sL)$ strict changes, for a constant
$\gamma_2=\gamma_2(\eta)$, and among
${j_{\mathrm{tot}}}=\lceil\ell(\gamma_2\log(sL)+1)\rceil$ blocks some $\ell$
consecutive ones have equal starting and ending lattice $\Lambda$. Here $\ell$
and a width threshold $w_*$ are the constants $\ell_0$ and $w_0$ of
Lemma~\ref{lem:import-filling} for the density
\[
 c_{\rm dens}=\frac{1}{4\gamma_1\gamma_r(b_*+1)^r}
\]
(where $\gamma_r$ is the constant of Lemma~\ref{lem:carrier-full-lattice-box}),
maximized over the finitely many possible ranks $r=O_\eta(1)$ (the rank bound of
Theorem~\ref{thm:import-freiman} for the doubling constant given by
\eqref{ms:growth}); thus $\ell$ and $w_*$ depend only on $\eta$. Stop at the first such run, let the
reserve $R$ be these $\ell$ blocks, and move all earlier blocks into $E$.
By the choice of $\gamma_{\mathrm{pack}}$, ${k_{\mathrm{pack}}}\ge b_*s{j_{\mathrm{tot}}}+h$: indeed
$b_*s{j_{\mathrm{tot}}}+h\le s\bigl(b_*\ell\gamma_2\log(sL)+2b_*\ell+1\bigr)$ since $h\le s$,
and $1\le2\log(sL)$ since $sL\ge2$.
Then every rank assertion above is covered by the packed bases, and so is the
deletion of $R$, the earlier blocks, and $h$ further copies; this gives
the robustness in~(3). The remaining live vectors generate the final
lattice $\Lambda$ of the run, and the reserve generates $\Lambda$ because it
contains seeds.

\paragraph{Filling.}
Apply Lemma~\ref{lem:carrier-full-lattice-box} to $C_0$ and $\Lambda$. It
returns a basis of $\Lambda$, which identifies $\Lambda$ with $\Z^r$, and an
integer box $C$ with $C_0\subseteq C\subseteq \gamma_r\cdot C_0$. By the dilation
bound of that lemma, $\Lambda\subseteq\Z^r$, a coordinatewise count, and
\eqref{ms:phase-volume},
\[
 |b_*s\cdot C\cap\Lambda|\le\gamma_r|b_*s\cdot C_0\cap\Z^r|
 \le\gamma_r(b_*+1)^r|s\cdot C_0\cap\Z^r|\le\gamma_1\gamma_r(b_*+1)^r{m_{\mathrm{ph}}}.
\]
Each block support lies in $b_*s\cdot C$, contains $0$, generates $\Z^r$, and has
at least ${m_{\mathrm{ph}}}/4$ points, hence density at least $c_{\rm dens}$ in this box.
The box has integer half-widths $b_*sb_j$, which exceed $w_*$ because $a_\eta$ is
chosen after $w_*$. Lemma~\ref{lem:import-filling} with these $\ell$ supports
gives a translate of $\gamma_{\rm fill}\ell b_*s\cdot C$ inside $S_1+\cdots+S_\ell$,
where $\gamma_{\rm fill}$ is the constant $\gamma$ of that lemma. With
$\gamma_0=\min(\gamma_{\rm fill}\ell b_*,\,10/\gamma_r)$ this gives
\[
 \vec{z}_0+(\gamma_0s\cdot C\cap\Z^r)\subseteq S_1+\cdots+S_\ell,
 \qquad \gamma_0s\cdot C\subseteq\gamma_0\gamma_r s\cdot C_0\subseteq 10s\cdot C_0.
\] Thus $\rho$, which is $\rho_0$ written in the new
basis, is injective on this box. Since the blocks use disjoint copies,
$S_1+\cdots+S_\ell$ is exactly the set of reserve subset sums in GAP
coordinates; this proves \eqref{ms:filling} with $\gamma=\gamma_0a_\eta$.
To make the filling explicit, compute the reserve sums with predecessor
pointers in a box of size $O_\eta({m_{\mathrm{ph}}})$, find a fully covered translate of
the box by one Boolean correlation, and read off witnesses.

\paragraph{Counting and time.}
A failed packing or failed greedy phase moves $O(s\log(sL))$ copies into
$E$ and halves $|sX(\vec y)|$, which starts at most $sL+1$ and never
increases. Hence there are $O(\log(sL))$ failed phases; the successful
phase adds $O(s\log(sL))$ copies to $E$ and $O(s)$ to $R$. This proves
\eqref{ms:exceptions} and the reserve bound.

Only $\widetilde O(s)$ copies are removed overall. Each recomputation of a
repeated sumset costs $\widetilde O(sL)$; each greedy step uses one FFT
correlation in a fixed dilate of $s\cdot C_0$, of size
$O_\eta({m_{\mathrm{ph}}})=O(sL)$, followed by an $O(m)$ scan. Basis packing uses
$\widetilde O(ms)$ counter operations, seeds and lattice chains use
fixed-rank Hermite normal forms, and there are $O(\log(sL))$ Freiman and
properization calls. The final witness computation costs
$\widetilde O(s^2L)$. Altogether this is
$\widetilde O(mh+h^2L)$. The output (containment, properness on the
box, and the filled translate) is verified explicitly, and the Freiman calls
are repeated on failure, which gives the error bound.
\qed

\subsection{Vector items and the old-vector map}\label{sec:vector-preprocessing}
The GAP construction takes positive scalars. To apply it to vector items in
$P=\prod_{j=1}^d[-p_j,p_j]$, first delete zero items and coordinates with
$p_j=0$ (after checking the corresponding target coordinates). Fix
$\gamma_3$ larger than four times one plus the reserve-size constant, and use
the mixed-radix encoding
\[
 t_j=2\lceil\gamma _3hp_j\rceil+1\quad(1\le j<d),\qquad
 \operatorname{enc}(\vec v)=v_1+t_1v_2+\cdots+(t_1\cdots t_{d-1})v_d,
\]
which is injective on $\gamma _3h\cdot P\cap\Z^d$ by successive reduction
modulo the radices. Replace every item with $\operatorname{enc}(\vec v_i)<0$ by
$-\vec v_i$ and complement its selection; this translates the targets, their
window, and the zonotope by the same vector $-\vec s$, where $\vec s$ is the
sum of the replaced items, and is an exact change of variables. Equal
positive encodings then mean equal reoriented vectors; merge them into types
with counts, recording constituents for witness recovery. In the set
recursion each count is at most two, because equal oriented vectors have
equal oriented original values. The scalar inputs are the values
$\operatorname{enc}(\vec v_i)$.

\begin{lemma}[Recovering the full old-vector map]\label{lem:carrier-old-map}
For a preprocessed vector input, the nonempty output of
Theorem~\ref{thm:carrier} admits an explicitly computable integer matrix $F$
with $F\vec{u}_a=\vec{v}_a$ for every live type and $\operatorname{enc}\circ F=\rho$.
\end{lemma}

\begin{proof}
For $\vec q$ in the filled box, let $f(\vec q)$ be the old-vector sum of its
reserve witness. Two witnesses have the same encoded value and use $O(h)$
items, so injectivity of $\operatorname{enc}$ on $\gamma_3h\cdot P$ makes $f$ well
defined. For a unit square $\vec q,\vec q+\vec e_j,\vec q+\vec e_l,\vec q+\vec e_j+\vec e_l$
in the box the encoded sums of opposite corners agree, so the same
injectivity gives
$f(\vec q)+f(\vec q+\vec e_j+\vec e_l)=f(\vec q+\vec e_j)+f(\vec q+\vec e_l)$;
three collinear points give vanishing second differences. Hence
$f(\vec q)=\vec w+F\vec q$ on the box with $F$ integral. The box has widths
$\Omega(h)$ times those of $C$, so for each live $\vec u_a\in C$ there is
$\vec q$ with $\vec q,\vec q+\vec u_a$ in the box; comparing encodings,
$\operatorname{enc}(f(\vec q+\vec u_a)-f(\vec q))=\rho(\vec u_a)=\operatorname{enc}(\vec v_a)$, and
injectivity gives $F\vec u_a=\vec v_a$. Adjacent points give
$\operatorname{enc}(F\vec e_j)=\rho(\vec e_j)$. Witnesses for $r+1$ points determine $F$,
and one scan verifies all identities.
\end{proof}

The filling in GAP coordinates (not merely a scalar progression inside the
projected reserve sums) is what makes this argument work.

\subsection{From a filled box to a filled zonotope}\label{sec:filled-zonotope}

The rounding step is the two-point case of the Shapley--Folkman
theorem~\cite{ZhouSF}; the correction by a reserve follows Chen, Hu, Mao, and
Zhang~\cite[Section 4, Lemmas 30--31]{CHMZDense}, and
Conlon--Fox--Pham~\cite[Lemmas 4.2 and 4.6]{CFP} use a related passage from
subset sums plus a box to zonotope neighbourhoods.

\begin{lemma}[Filling a zonotope using reserved items]
\label{lem:filled-zonotope}
Let $\Lambda\subset\R^r$ be a full-rank lattice, let $C\subset\R^r$
be a nonempty compact origin-symmetric convex set, and let $U$ be a finite
multiset in $C\cap\Lambda$. Suppose a submultiset $R\subseteq U$ and a point
$\vec{z}_0\in\Lambda$ satisfy
$\vec{z}_0+(r\cdot C\cap\Lambda)\subseteq\mathcal{S}(R)$.
Put $J=\vec{z}_0+Z(U\setminus R)$. Then
\[
 J\cap\Lambda\subseteq\mathcal{S}(U)\subseteq Z(U)\cap\Lambda,
 \qquad J\subseteq Z(U)\subseteq J+|R|\cdot C,
\]
and every $\vec{z}\in\Lambda$ with $\vec{z}+|R|\cdot C\subseteq Z(U)$ lies in
$J\cap\Lambda$. The same holds when $U$ is stored by distinct vectors with
counts, and the witnesses respect the counts.
\end{lemma}

\begin{proof}
For $\vec{z}\in J\cap\Lambda$, write $\vec{z}-\vec{z}_0$ as a fractional
subset sum of $U\setminus R$ with one variable in $[0,1]$ per item, and take
an extreme point of this polytope. At most $r$ variables are fractional,
since otherwise a linear dependence among their vectors would allow a
perturbation in both directions. Rounding them gives
$V\subseteq U\setminus R$ with
$\vec{z}-\Sigma(V)\in \vec{z}_0+(r\cdot C\cap\Lambda)$, by symmetry and
convexity of $C$. This difference is a subset sum of $R$, disjoint from $V$.
For counts, use one variable in $[0,c_i]$ per distinct vector (with $c_i$ its count in $U\setminus R$) and round the
at most $r$ nonintegral ones to adjacent integers.

Since $\vec{z}_0\in\mathcal{S}(R)\subseteq Z(R)$, $J\subseteq Z(U)$.
Subtracting the incidence vector of a subset of $R$ with sum $\vec z_0$ from
any fractional selection of $R$ leaves coefficients in $[-1,1]$, so
$Z(R)-\vec{z}_0\subseteq |R|\cdot C$ and $Z(U)\subseteq J+|R|\cdot C$. Finally, if
$\vec{z}+|R|\cdot C\subseteq Z(U)$, then for every $\vec a$,
$\vec{a}\cdot \vec{z}+|R| h_C(\vec{a})\le h_{Z(U)}(\vec{a})\le h_J(\vec{a})+|R| h_C(\vec{a})$,
so $\vec a\cdot\vec z\le h_J(\vec a)$ for all $\vec a$ and $\vec z\in J$.
\end{proof}

\section{The boundary reduction}\label{sec:normal}

\paragraph{Setting.}
\label{par:normal-input}\label{par:normal-representation}
Fix $d$. The items $V=\{\vec{v}_1,\ldots,\vec{v}_n\}\subseteq\Z^d$ span $\R^d$
and lie in $P=\prod_{j=1}^d[-p_j,p_j]$, $p_j\in\Z_{\ge1}$, so
$h_P(\vec a)=\sum_jp_j|a_j|$. The scale $\tau$ exceeds a sufficiently large
constant depending on $d$, and the exact target set
$\mathcal T\subseteq\vec{o}+\tau\cdot P$ is stored in an array of allocated
size $w=|\vec o+\tau\cdot P|$. All numbers have polynomially
bounded bit length, and $0<\xi<1$ is a failure allowance. The working
lattice is $\Z^d$; when an operation chooses a basis of a lattice $\Lambda$
we use its coefficients as new integer coordinates, measure volumes in
lattice units $\vol(\cdot)/\det\Lambda$, and count array entries in the basis
actually used. Write
\[
 Z=\sum_{i=1}^n[\vec{0},\vec{v}_i],\qquad
 h_Z(\vec a)=\sum_{i=1}^n\max(0,\vec a\cdot\vec v_i).
\]
For a selection $x\in\{0,1\}^n$ and a direction $\vec a$, the
\emph{deficit} $h_Z(\vec a)-\vec a\cdot\sum_ix_i\vec v_i$ is a sum of the
nonnegative item losses
\begin{equation}\label{eq:item-loss}
 g_{i}^{x_i}(\vec a)=\max(0,\vec a\cdot\vec v_i)-x_i\,(\vec a\cdot\vec v_i),
 \qquad g_i^0(\vec a)+g_i^1(\vec a)=|\vec a\cdot\vec v_i|.
\end{equation}
This identity is the source of all bounds on how a solution can differ from
a reference selection near the boundary.

\paragraph{Classification.}
\label{par:normal-infeasible}\label{par:normal-interior}\label{par:normal-boundary}
Fix small constants $\vartheta>0$ and $\varpi=c_d\vartheta$, chosen in the
proof. Partition $\vec o+\tau\cdot P$ into a constant number of grid cells,
so small that differences of two points of a cell lie in
$c'_d\varpi\tau\cdot P$, where $c'_d\le\lambda_d$ (the constant of Step~1 of
the proof) and $c'_dc_d<1/100$; thus the targets of a cell lie in
$\vec x+\sigma\tau\cdot P$ with $\sigma<\vartheta/100$ for any point $\vec x$ of
the cell. For each cell:
\begin{enumerate}
\item If the cell does not meet $Z$ (a convex feasibility test), it is
\emph{infeasible} and its targets are discarded.
\item Otherwise pick a rational $\vec x\in Z$ in the cell. If all $2^d$
corners of $\vec x+\vartheta\tau\cdot P$ lie in $Z$, then by convexity
$\vec x+\vartheta\tau\cdot P\subseteq Z$; the cell with its targets is a
\emph{buffered interior patch}, handled in Section~\ref{sec:deep}. Each of
its targets $\vec t$ satisfies $\vec t+(\vartheta-c'_d\varpi)\tau\cdot P\subseteq Z$.
\item Otherwise a failed corner yields a support normal $\vec a$ with
\begin{equation}\label{eq:normal-center}
 0\le h_Z(\vec a)-\vec a\cdot\vec x\le\vartheta \tau h_P(\vec a),
\end{equation}
and $\vec x$ is a \emph{normal center}; the cell is a \emph{boundary cell}.
\end{enumerate}
Membership and separation queries for $Z$ need no facet list: the support
function above has an $O(n)$ evaluation, and
Theorem~\ref{thm:import-geometry}(i) applies.

\begin{lemma}[Boundary reduction]\label{lem:normal}
In the setting above, for a sufficiently small fixed choice of
$\vartheta,\varpi>0$, there is an algorithm that partitions $\mathcal T$ into
infeasible cells, buffered interior patches, and normal patches, and replaces
the normal patches by child instances with the following properties. Every
child uses a subset of the original items, with signs possibly changed,
keeps the scale $\tau$, and has dimension equal to the rank of its items.
Every child solution gives a parent solution; with probability at least
$1-\xi$, every parent solution for a target in a normal patch survives in
some child. The children have total allocated size at most $w/2$, and the
classification and reduction cost
$\softO{n+w}$, with polynomial dependence on $\log(1/\xi)$ and on the
bit length.
\end{lemma}

\begin{lemma}[Local lattice boxing]\label{lem:boxing}
Fix $d\ge1$. There is a constant $\gamma_d$ such that the
following holds. Let $M\subset\R^d$ be an origin-symmetric convex body and
$\Lambda$ a full-rank rational lattice, and suppose that for some real
$\ell\ge1$ there are linearly independent $\vec u_1,\ldots,\vec u_d\in\Lambda$
with $\ell\vec u_i\in M$. Then there are a basis
$\vec b_1,\ldots,\vec b_d$ of $\Lambda$ and integers $n_j\ge1$ such that
$D=\{\sum_j z_j\vec b_j:|z_j|\le n_j\}$ satisfies
\[
 M\subseteq D\subseteq\gamma_d\cdot M,\qquad
 |D\cap\Lambda|\le\gamma_d\frac{\vol(M)}{\det\Lambda}.
\]
For fixed $a>0$, every translate of $a\cdot M$ has an enclosing box of this
kind with $O_{d,a}(\vol(M)/\det\Lambda)$ entries. Given a separation oracle and
polynomial radius bounds, the basis and box are computable in polynomial
time.
\end{lemma}

$D$ is a box in the coefficients of the chosen basis, i.e.\ an oblique
parallelepiped whose integer points index an array. The lower bound $\ell\ge1$
ensures that rounding the half-widths to integers costs only a constant
factor. The proof is in Appendix~\ref{app:local-boxing}.

\begin{proof}[Proof of Lemma~\ref{lem:normal}]
\emph{Step 1: packing.}
Fix ${a_{\mathrm{buf}}}$ large depending on $d$, and put
\[
 K=(Z+\varpi\tau\cdot P)\cap
      \bigl(\vec o+(1+{a_{\mathrm{buf}}})\tau\cdot P\bigr),\qquad
 M_{\vec x}=(K-\vec x)\cap(\vec x-K).
\]
So $\vec u\in M_{\vec x}$ iff $\vec x\pm\vec u\in K$, and
$M_{\vec x}\supseteq\varpi\tau\cdot P$ for $\vec x\in Z\cap(\vec o+\tau\cdot P)$.
The extra window ${a_{\mathrm{buf}}}\tau\cdot P$ leaves room for heavy-item
shifts. Fix a small $\lambda=\lambda_d$. For a normal center $\vec x$ and
$\vec y\in\vec x+\lambda\cdot M_{\vec x}$, \eqref{eq:normal-center} gives
$h_K(\vec a)-\vec a\cdot\vec y\le(1+\lambda)(\vartheta+\varpi)\tau h_P(\vec a)$,
so these small regions lie in the shell
$K\setminus(K\ominus a_0\tau\cdot P)$, $a_0=2(1+\lambda)(\vartheta+\varpi)$.
After scaling $\tau\cdot P$ to $[-1,1]^d$, $K$ lies in a cube of side
$2(1+{a_{\mathrm{buf}}})$; eroding one coordinate at a time removes at most $2a_0$
from each line section, so by Fubini
\begin{equation}\label{eq:normal-shell}
 \vol\bigl(K\setminus(K\ominus a_0\tau\cdot P)\bigr)
 \le d(1+{a_{\mathrm{buf}}})^{d-1}a_0\vol(\tau\cdot P).
\end{equation}
Select normal centers greedily so that the regions
$\vec x+\lambda\cdot M_{\vec x}$ have disjoint interiors. By
Lemma~\ref{lem:import-macbeath}, the fixed dilations
$\vec x+\gamma_d\cdot M_{\vec x}$ of the selected regions cover all candidate
regions, and since a boundary cell lies in a small fixed dilate of
$M_{\vec x}\supseteq\varpi\tau\cdot P$ around its center, they cover every
boundary cell. Assign each boundary target to one covering region containing
it; these are the \emph{normal patches}, and every target is assigned once.
By disjointness and \eqref{eq:normal-shell},
\begin{equation}\label{eq:normal-macbeath-sum}
 \sum_{\vec{x}\text{ selected}}\vol(M_{\vec{x}})
 \le \gamma _{d,{a_{\mathrm{buf}}}}\vartheta\vol(\tau\cdot P).
\end{equation}
Figure~\ref{fig:normal-macbeath} illustrates the construction. Constants in
\eqref{eq:normal-macbeath-sum} do not depend on $\vartheta,\varpi$.

\begin{figure}[t]
\centering
\begin{minipage}[t]{0.485\linewidth}
\centering
\begin{tikzpicture}[x=1.15cm,y=1cm,font=\small]
  \path[use as bounding box] (-1.3,-1.92) rectangle (4.3,3.95);
  \node[font=\small\bfseries] at (1.5,3.70)
    {(a) Symmetry about a center};

  \fill[black!4] (0,0) rectangle (4,3);
  \fill[blue!16] (0,0) rectangle (3,2);
  \draw[black!80,thick] (0,0) rectangle (4,3);
  \draw[orange!85!black,dashed,thick] (-1,-1) rectangle (3,2);
  \draw[blue!65!black,thin] (0,0) rectangle (3,2);
  \node[anchor=west] at (4.05,2.87) {$K$};

  \node[blue!65!black] at (1.50,2.53) {$\vec x+M_{\vec x}$};
  \draw[blue!65!black,->] (1.50,2.35) -- (1.50,2.03);
  \draw[black!70,thin] (0.60,0.55) -- (2.40,1.45);
  \fill[blue!75!black] (0.60,0.55) circle (1.7pt);
  \fill[black] (1.50,1.00) circle (2pt);
  \fill[blue!75!black] (2.40,1.45) circle (1.7pt);
  \node[below right,inner sep=3pt] at (1.50,1.00) {$\vec x$};
  \node[below,inner sep=4pt] at (0.60,0.55) {$\vec x-\vec u$};
  \node[above,inner sep=3pt] at (2.40,1.45) {$\vec x+\vec u$};
  \node[orange!85!black] at (1.5,-1.27)
    {reflected body $2\vec x-K$};
  \node[font=\scriptsize] at (1.5,-1.72)
    {$\vec x+M_{\vec x}=K\cap(2\vec x-K)$};
\end{tikzpicture}
\end{minipage}\hfill
\begin{minipage}[t]{0.485\linewidth}
\centering
\begin{tikzpicture}[x=1.15cm,y=1cm,font=\small]
  \path[use as bounding box] (-0.8,-1.92) rectangle (4.8,3.95);
  \node[font=\small\bfseries] at (2,3.70)
    {(b) Packing and covering};

  \fill[black!3] (0,0) rectangle (4,3);
  \fill[yellow!20] (0,1.95) rectangle (4,3);
  \draw[yellow!50!black,densely dotted] (0,1.95) -- (4,1.95);
  \draw[black!80,thick] (0,0) rectangle (4,3);
  \draw[teal!75!black,dashed,thick] (0.2,0.2) rectangle (3.8,2.8);
  \node[anchor=west] at (4.06,0.13) {$K$};
  \node[teal!75!black,anchor=west] at (4.06,1.12) {$Z$};
  \draw[teal!75!black] (3.80,1.12) -- (4.00,1.12);
  \node[font=\scriptsize] at (2,3.29) {one side of the boundary shell};

  \draw[orange!85!black,thick] (0.4,2.25) rectangle (2,2.85);
  \draw[orange!85!black,thick] ({2/3},2.25) rectangle ({10/3},2.85);
  \draw[orange!85!black,thick] (2,2.25) rectangle (3.6,2.85);
  \filldraw[fill=blue!65,draw=blue!75!black]
     (1,2.475) rectangle (1.4,2.625);
  \filldraw[fill=blue!65,draw=blue!75!black]
     ({5/3},2.475) rectangle ({7/3},2.625);
  \filldraw[fill=blue!65,draw=blue!75!black]
     (2.6,2.475) rectangle (3,2.625);
  \foreach \x/\i in {1.2/1,2/2,2.8/3} {
    \fill[black] (\x,2.55) circle (1.4pt);
    \node[font=\scriptsize] at (\x,2.08) {$\vec x_{\i}$};
  }
  \foreach \x in {1.5,2.5} {
    \draw[black,fill=white] (\x,2.55) circle (1.7pt);
  }

  \filldraw[fill=blue!65,draw=blue!75!black]
    (-0.1,-0.48) rectangle (0.2,-0.34);
  \node[anchor=west,font=\scriptsize] at (0.32,-0.41)
    {selected $\vec x_i+\lambda\cdot M_{\vec x_i}$};
  \draw[orange!85!black,thick] (-0.1,-0.87) rectangle (0.2,-0.70);
  \node[anchor=west,font=\scriptsize] at (0.32,-0.785)
    {covering $\vec x_i+\gamma\cdot M_{\vec x_i}$};
  \draw[black,fill=white] (0.05,-1.14) circle (1.7pt);
  \node[anchor=west,font=\scriptsize] at (0.32,-1.14)
    {an unselected candidate center};
  \node[font=\scriptsize] at (2,-1.72)
    {$\lambda=1/6,\qquad\gamma=2/3=4\lambda$};
\end{tikzpicture}
\end{minipage}
\caption{Macbeath regions and the packing used for the boundary estimate.
The rectangular $K$ is a special case; in general neither $K$ nor its
Macbeath regions are axis parallel.
(a) For $K=[0,4]\times[0,3]$ and $\vec x=(3/2,1)$, the blue rectangle
is exactly $\vec x+M_{\vec x}$. A displacement $\vec u$ is permitted precisely
when both displayed points $\vec x\pm\vec u$ belong to $K$.
(b) A finite example of the packing and cover, using the same $K$.
The solid outline is the buffered body $K$ and the dashed rectangle is
the smaller zonotope $Z$; all five candidate
centers lie in $Z$. The three blue regions have disjoint interiors.
Each omitted small candidate region intersects the middle blue region
and lies in its orange enlargement. The orange covering regions overlap.
Only the top portion of the boundary shell is shaded. In the proof,
containment of the selected small regions in the shell and disjointness
give $\sum_i\operatorname{area}(\lambda\cdot M_{\vec x_i})
\le\operatorname{area}(\text{shell})$.
This charge concerns the small regions; enlarged covering regions are
not required to be disjoint or to remain in the shell.}
\label{fig:normal-macbeath}
\label{fig:normal-packing}
\end{figure}
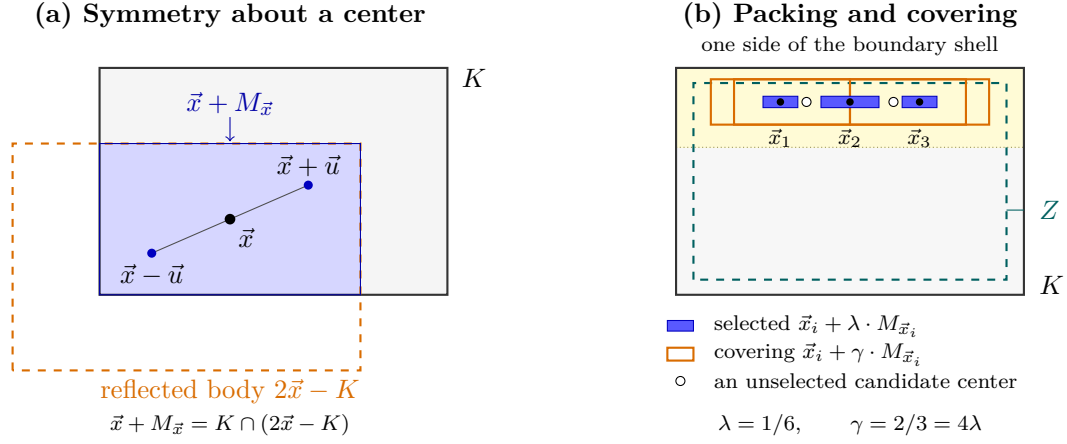

\emph{Step 2: local parallelepipeds.}
We replace $M_{\vec x}$ by a parallelepiped
$D_{\vec x}=\{\vec u:|\vec a_j\cdot\vec u|\le 1,\ 1\le j\le d\}$ with
independent rows and
\begin{equation}\label{eq:normal-row-box}
 M_{\vec{x}}\subseteq D_{\vec{x}}\subseteq 2d\cdot M_{\vec{x}},
\end{equation}
using only \emph{permissible rows}: supporting directions of one of the two
bodies defining $K$. A row from $Z+\varpi\tau\cdot P$ is a \emph{zonotope row};
a row from the window is an \emph{artificial row}. Normalize a supporting
direction $\vec a$ of such a body $B$ by the smaller of its two slacks
$h_B(\vec a)-\vec a\cdot\vec x$ and $h_B(-\vec a)+\vec a\cdot\vec x$ (both positive),
oriented toward the smaller one, so that the row reads $|\vec a\cdot\vec u|\le1$
on $M_{\vec x}$. Start with the $d$ artificial coordinate
rows, with row matrix $A$, and put $\vec u_j=A^{-1}\vec e_j$. If
$\vec x\pm\vec u_j/2\notin K$ for some $j$, the separation oracle returns a
supporting inequality of one defining body; normalizing it gives a row with
$|\vec a\cdot\vec u_j|>2$, which replaces row $j$ and multiplies $|\det A|$ by
more than $2$. All normalized rows are bounded in the dual norm because
$M_{\vec x}\supseteq\varpi\tau\cdot P$, so every possible $|\det A|$ is at most
$\gamma_d\varpi^{-d}/\vol(\tau\cdot P)$, whereas the initial value is at least a
fixed multiple of $(1+{a_{\mathrm{buf}}})^{-d}/\vol(\tau\cdot P)$; hence only
$O_d(\log((1+{a_{\mathrm{buf}}})/\varpi))$ replacements occur. At termination every
normalized defining row of $M_{\vec x}$ has $|\vec a\cdot\vec u_j|\le2$ for all
$j$, which gives \eqref{eq:normal-row-box}. For the greedy packing we may use
the parallelepipeds $\vec x+(\lambda/2d)\cdot D_{\vec x}$, whose pairwise
intersections are fixed-dimensional linear programs; intersection implies
intersection of the Macbeath regions at scale $\lambda$, so Step~1 is
unchanged.

By this normalization a zonotope row satisfies
$h_Z(\vec{a}_j)+\varpi \tau h_P(\vec{a}_j)-\vec{a}_j\cdot \vec{x}=1$, and for an
artificial row $\vec x\in\vec o+\tau\cdot P$ gives
\begin{equation}\label{eq:artificial-width}
 {a_{\mathrm{buf}}} \tau h_P(\vec{a}_j)\le1.
\end{equation}
Let $q\le d$ be the number of zonotope rows. Every feasible target $\vec y$ of
the patch satisfies, for every zonotope row,
\begin{equation}\label{eq:normal-target-deficits}
 0\le h_Z(\vec{a}_j)-\vec{a}_j\cdot \vec{y}\le \beta,
 \qquad \beta=1+\gamma _d,
\end{equation}
since the patch lies in $\vec x+\gamma_d\cdot M_{\vec x}\subseteq\vec x+\gamma_d\cdot D_{\vec x}$
and $h_Z(\vec a_j)-\vec a_j\cdot\vec x\le 1$,
so we may clear all patch targets violating \eqref{eq:normal-target-deficits}.
Let $\mathcal T_{\vec x}$ be the remaining patch targets.

\emph{Step 3: reference selection and heavy items.}
Write $g_{ij}^e=g_i^e(\vec a_j)$. For each item choose the reference bit
$b_i\in\{0,1\}$ minimizing $\sum_{j\le q}g_{ij}^{b_i}$ (sums over $j\le q$
run over zonotope rows). A feasible selection $x$ for a target in
$\mathcal T_{\vec x}$ has $\sum_ig_{ij}^{x_i}\le\beta$ for each zonotope row,
hence total loss at most $q\beta$; so the reference has total loss at most
$q\beta$, and if the computed minimum exceeds $q\beta$ the patch
has no solution and is discarded. Reanchor at $\vec b=\sum_ib_i\vec v_i$ as in
\eqref{eq:reanchor}, with signed items $\vec v_i'=(1-2b_i)\vec v_i$.

Call item $i$ \emph{heavy} if $|\vec{a}_j\cdot \vec{v}_i|>\beta/\tau$ for some
zonotope row, and let $H$ be the heavy items. Fix a feasible selection $x$
for a target in $\mathcal T_{\vec x}$ (used only in the analysis) and let
$X=\{i\in H:x_i\ne b_i\}$ be its changed heavy items; its heavy part after
reanchoring is $\sum_{i\in X}\vec v_i'$. For $i\in X$,
$g_{ij}^{b_i}+g_{ij}^{x_i}=|\vec a_j\cdot\vec v_i|$ by \eqref{eq:item-loss}, so
summing the losses of the reference and of $x$,
\[
 \sum_{i\in X}\sum_{j\le q}|\vec a_j\cdot\vec v_i|\le 2q\beta .
\]
Each heavy item contributes more than $\beta/\tau$ on the left. In a single
zonotope row $j$, the same identity bounds $\sum_{i\in X}|\vec a_j\cdot\vec v_i|$
by the row-$j$ losses of the reference and of $x$, which are at most
$q\beta$ and $\beta$. Using $\vec v_i\in P$ and
\eqref{eq:artificial-width} for artificial rows, we obtain
\begin{gather}
 |X|\le2q\tau,\label{eq:normal-heavy-cardinality}\\
 \sum_{i\in X}|\vec{a}_j\cdot \vec{v}_i'|\le(q+1)\beta\quad(\text{zonotope rows}),
 \label{eq:normal-heavy-absolute}\\
 \sum_{i\in X}|\vec{a}_j\cdot \vec{v}_i'|\le2q\tau h_P(\vec{a}_j)\le2q/{a_{\mathrm{buf}}}\quad(\text{artificial rows}).
 \label{eq:normal-heavy-artificial}
\end{gather}
These are absolute bounds, so they bound all partial sums.

\emph{Step 4: transferring the heavy choices.}
Apply Lemma~\ref{lem:boxing} to a fixed dilation of $D_{\vec x}$ in the parent
lattice. Each lattice coordinate is a linear combination of the rows, so
\eqref{eq:normal-heavy-absolute}--\eqref{eq:normal-heavy-artificial} give
absolute coordinate budgets $\Delta_k$; they are the coordinate bounds of a
fixed dilation of $D_{\vec x}$, so $\prod_k(1+\Delta_k)=O(w_{\vec x}^{\rm loc})$, and Lemma~\ref{lem:quota} applied to the signed
heavy items computes a set $S_{\vec x}$ of actual heavy-flip sums that, with
the prescribed probability, contains the heavy part of every feasible
selection, in time $\softO{n+w_{\vec x}^{\rm loc}}$ where
\begin{equation}\label{eq:normal-local-allocation}
 w_{\vec{x}}^{\rm loc}\le \gamma _{d,{a_{\mathrm{buf}}}}\frac{\vol(M_{\vec{x}})}{\det\Lambda}.
\end{equation}
The lattice-coordinate budgets may admit extra sums, so we also clear every
$\vec s\in S_{\vec x}$ violating the row bounds
$|\vec a_j\cdot\vec s|\le(q+1)\beta$ (zonotope rows) and
$|\vec a_j\cdot\vec s|\le2q/{a_{\mathrm{buf}}}$ (artificial rows); required sums
pass by \eqref{eq:normal-heavy-absolute}--\eqref{eq:normal-heavy-artificial}.
Then compute, by one Boolean convolution in the parent lattice (both operands
lie in fixed dilations of translates of $D_{\vec x}$, so it costs
$\softO{w_{\vec x}^{\rm loc}}$),
\begin{equation}\label{eq:normal-transfer-overview}
 \mathcal U_{\vec x}=(\mathcal T_{\vec x}-\vec b)-S_{\vec x},
\end{equation}
and delete \emph{all} heavy items: their reference choices are in $\vec b$
and their retained changes in $S_{\vec x}$, so keeping them would allow a
copy to be used twice. Finally, with
$\alpha_j=\sum_{i\notin H}g_{ij}^{b_i}$, keep only the residual targets with
\begin{equation}\label{eq:normal-residual-clip}
 0\le\alpha_j-\vec{a}_j\cdot \vec{r}\le \beta
 \quad\hbox{for all zonotope rows }j.
\end{equation}
This filter loses no relevant solution: if $\vec r=\sum_{i\notin H}y_i\vec v_i'$
for the flips $y$ of a feasible selection $x$, then
$\alpha_j-\vec a_j\cdot\vec r=\sum_{i\notin H}g_{ij}^{x_i}$, the light part of the
deficit of $x$, which lies in $[0,\beta]$.

\emph{Step 5: the child instance.}
After \eqref{eq:normal-residual-clip}, each zonotope-row value of a residual
target lies in $[\alpha_j-\beta,\alpha_j]$, and each artificial-row value
lies in an interval centred at $\vec a_j\cdot(\vec x-\vec b)$ of radius
$\gamma_d+2q/{a_{\mathrm{buf}}}$. Solving for the point $\vec r_0$ with the $d$
midpoints as row values, the residual targets lie in
$\vec r_0+\gamma_{d,{a_{\mathrm{buf}}}}\cdot D_{\vec x}$. Every light item satisfies
$|\vec a_j\cdot\vec v_i'|\le\beta/\tau$ in zonotope rows and
$\le h_P(\vec a_j)\le1/({a_{\mathrm{buf}}}\tau)$ in artificial rows, so $\tau\vec v_i'$ lies in the same
symmetric shape $\gamma_{d,{a_{\mathrm{buf}}}}\cdot D_{\vec x}$. Constants depend on
$d,{a_{\mathrm{buf}}}$ only.

Now, in the lattice basis of Step~4, restrict to the real span $\mathcal H$ of the light items: discard residual
targets outside $\mathcal H$ and, if $d'=\dim\mathcal H<d$, project onto $d'$ integer
coordinates on which projection is injective on $\mathcal H$; this preserves the exact
problem and does not enlarge the coordinate box. If $d'=0$ the child accepts
iff $0$ is a residual target. Otherwise choose a coordinate box
$\prod_j[-\omega_j,\omega_j]$ containing both the translated residual targets
and $\tau$ times every light item. The projected light items have full rank,
so each coordinate is nonzero on some item and $\omega_j\ge\tau$; with an
integer centre within distance one and $p_j'=\lceil(\omega_j+1)/\tau\rceil\le3\omega_j/\tau$,
the child has item box $P'=\prod_j[-p'_j,p'_j]$, targets in a translate of
$\tau\cdot P'$, and allocated size $w_{\vec x}=O_d(w_{\vec x}^{\rm loc})$. Restricting to
the span \emph{before} rounding is essential: rounding a width below one in a
degenerate direction could cost a factor depending on $1/\varpi$. The
convolution \eqref{eq:normal-transfer-overview} was already charged in the
parent lattice.

\emph{Step 6: correctness.}
Soundness: a child solution, after undoing the injective projection, is a light
subset $Y$ with sum $\vec r\in\mathcal U_{\vec x}$, so $\vec r=\vec t-\vec b-\vec s$
for some $\vec t\in\mathcal T_{\vec x}$ and $\vec s\in S_{\vec x}$ with a heavy
witness $X'$. As $X'$ and $Y$ are disjoint,
$\vec b+\sum_{i\in X'}\vec v_i'+\sum_{i\in Y}\vec v_i'=\vec t$, and undoing the
reference bits gives a parent subset. This needs no random event.
Completeness: the heavy part of a parent solution for a patch target obeys
the absolute budgets, so it is retained on the success event of
Lemma~\ref{lem:quota}; then its light part is a residual target that passes
\eqref{eq:normal-residual-clip}, lies in $\mathcal H$, and survives the projection.
A union bound over the constantly many patches gives probability $1-\xi$.

\emph{Step 7: total size and cost.}
Since $M_{\vec x}\supseteq\varpi\tau\cdot P$ and the items span the parent
lattice, Lemma~\ref{lem:boxing} applies with $\ell=\varpi\tau$, and a large
enough lower threshold on $\tau$ absorbs lattice endpoint effects. As
$w=\Theta_d(\vol(\tau\cdot P)/\det\Lambda)$,
\eqref{eq:normal-macbeath-sum} and \eqref{eq:normal-local-allocation} give
\[
 \sum_{\vec{x}} w_{\vec{x}}\le \gamma _{d,{a_{\mathrm{buf}}}}\vartheta w,
\]
with all fixed enlargements and rounding included and no inverse power of
$\vartheta$ or $\varpi$. Choose $\vartheta$ so that this is at most $w/2$, then
$\varpi=c_d\vartheta$; only the number of geometric queries and the lower
threshold on $\tau$ depend on them. The numbers of cells, patches, and
geometric queries are constants, each query costs $\softO{n}$, one scan
assigns the targets, and the local Boolean work is
$\softO{\sum_{\vec x}w_{\vec x}^{\rm loc}}=\softO{w}$.
\end{proof}

\section{The interior reduction}
\label{sec:deep}

This section handles a target patch with a buffer of order $\tau\cdot P$
inside the zonotope. Using the common GAP we either certify a target, or
find a thin direction that becomes one new coordinate.

\paragraph{Hypotheses.}
The items are $\vec v_i\in P\cap\Lambda_0$, $i\in I$, where $\Lambda_0\subset\R^d$
is the current lattice, and $Z=\sum_i[0,\vec v_i]$. For a fixed constant
$\sigma>0$ and a point $\vec o$,
\begin{equation}
 \mathcal T\subseteq\vec{o}+\sigma\tau\cdot P,
 \qquad
 \vec{o}+10\sigma\tau\cdot P\subseteq Z.
 \label{eq:deep-buffer}
\end{equation}
After the exceptional items of Theorem~\ref{thm:carrier} have been
transferred, the remaining items have GAP coordinates
$\vec u_i\in C\cap\Lambda$, where $\Lambda\subset\R^r$ has a supplied basis and
$C$ is a symmetric box, and Lemma~\ref{lem:carrier-old-map} gives a linear map
$F:\Lambda\to\Lambda_0$ with $F\vec u_i=\vec v_i$. The reserve $R\subseteq I$ has
$|R|=O(h)$ items, explicit witnesses, and
\begin{equation}
 \vec{z}_0+((2r+1)\cdot C\cap\Lambda)\subseteq\mathcal{S}(\{\vec{u}_i:i\in R\});
 \label{eq:reserve-core-fill}
\end{equation}
in particular $F\vec z_0\in|R|\cdot P$. Removing $R$ and any $h$ further items
leaves vectors spanning $\R^r$. Here
$h=\lfloor\tau/\log^{c_{\mathrm{log}}}N\rfloor$ and $N$ is sufficiently
large, so $|R|+h=O(h)=o(\tau)$. Let $w$ be the allocated size of the parent
box of shape $\tau\cdot P$.

\begin{lemma}[Interior reduction]\label{lem:deep}
Under these hypotheses, an exact reduction either resolves the instance or
produces one equivalent instance on surviving original items, of vector rank
$d+1$, with an exact target set in a lattice box of $O(w)$ entries whose
symmetric shape divided by $h$ contains every child item. Excluding the GAP
construction, it takes $\softO{n+w}$ time and is deterministic.
\end{lemma}

\begin{proof}
\emph{The certified core.} Put
$J=\vec{z}_0+\sum_{i\in I\setminus R}[0,\vec{u}_i]$.
Lemma~\ref{lem:filled-zonotope} and \eqref{eq:reserve-core-fill} give
\begin{equation}
 F(J\cap\Lambda)\subseteq\mathcal{S}(\{\vec{v}_i:i\in I\}).
 \label{eq:certified-core}
\end{equation}
Discard targets outside $F\Lambda$ (no subset sum lies there); if none remain,
answer \emph{no}. Otherwise fix a remaining target $\vec y_0$. Replacing the
reserve zonotope by $F\vec z_0$ changes support values by at most
$2|R|h_P$, and $\vec y_0-\vec o\in\sigma\tau\cdot P$, so
\begin{equation}
 \vec{y}_0+(9\sigma\tau-2|R|)\cdot P\subseteq FJ,
 \label{eq:projected-core-buffer}
\end{equation}
a buffer of radius $\Omega(\tau)$.

\emph{Lattice point or thin direction.} If $r=d$ then $F$ is injective and the
unique fiber point over $\vec y_0$ is a lattice point, so
\eqref{eq:certified-core} answers \emph{yes}. Let $r>d$. The unreserved
vectors span $\R^r$, so $J$ has interior, and by
\eqref{eq:projected-core-buffer} $\vec y_0$ is interior to $FJ$; hence the
fiber $J\cap F^{-1}(\vec y_0)$ is full-dimensional in its affine kernel space,
whose lattice is nonempty because $\vec y_0\in F\Lambda$.
Theorem~\ref{thm:import-geometry}(iii) returns either a lattice point of the
fiber, answering \emph{yes} by \eqref{eq:certified-core}, or a nonzero
integer functional on $\ker F\cap\Lambda$ of width at most $f_r$ on the fiber;
dividing by the gcd of its coefficients makes it primitive without increasing
the width.
As $\ker F\cap\Lambda$ is saturated in $\Lambda$, it extends to an integer
functional $\psi$ on $\Lambda$.

\emph{A strip over the whole patch.} For $\vec y\in FJ$ let
$\psi_-(\vec y)$ and $\psi_+(\vec y)$ be the minimum and maximum of $\psi$ on
$J\cap F^{-1}(\vec y)$. Choose a rational $\vec\beta$ and put
$\alpha=\min_{\vec u\in J}(\psi(\vec u)-\vec\beta\cdot F\vec u)$ so that
\begin{equation}
 0\le \psi_-(\vec{y}_0)-\alpha-\vec{\beta}\cdot \vec{y}_0\le1.
 \label{eq:dual-slack}
\end{equation}
By LP duality $\psi_-(\vec y_0)=\max_{\vec\beta}G(\vec\beta)$ for the concave function
\[
 G(\vec{\beta})=\psi(\vec{z}_0)+\vec{\beta}\cdot(\vec{y}_0-F\vec{z}_0)
          +\sum_{i\notin R}\min(0,\psi(\vec{u}_i)-\vec{\beta}\cdot \vec{v}_i),
\]
whose hypograph has an $O(n)$ separation oracle. With $r_{\rm buf}$ a Euclidean buffer
radius from \eqref{eq:projected-core-buffer} and
$b_\psi=|\psi(\vec z_0)|+\sum_{i\notin R}|\psi(\vec u_i)|$, evaluating the
minimum at a preimage of $\vec y_0+r_{\rm buf}\vec\beta/\|\vec\beta\|$ gives
$G(\vec\beta)\le b_\psi-r_{\rm buf}\|\vec\beta\|$, while $G(0)\ge-b_\psi$; so a
maximizer lies in a ball of radius $1+2b_\psi/r_{\rm buf}$. Optimizing to accuracy
$1/2$ and rounding as in Lemma~\ref{lem:bits} gives \eqref{eq:dual-slack} in
$\softO{n}$ time.

Put $\kappa=f_r+1$. The function $g(\vec y)=\psi_+(\vec y)-\alpha-\vec\beta\cdot\vec y$
is concave and nonnegative on $FJ$, and $g(\vec y_0)\le\kappa$ by
\eqref{eq:dual-slack} and the width bound $\psi_+(\vec y_0)-\psi_-(\vec y_0)\le f_r$. If $\vec y$ and $2\vec y_0-\vec y$ lie in $FJ$, then
$g(\vec y_0)\ge(g(\vec y)+g(2\vec y_0-\vec y))/2\ge g(\vec y)/2$, so
\begin{equation}
 0\le\psi(\vec{u})-\alpha-\vec{\beta}\cdot F\vec{u}\le2\kappa
 \quad(\vec{u}\in J,\ F\vec{u}=\vec{y}).
 \label{eq:reflection-band}
\end{equation}
For any target $\vec t$ and any reserve subset sum $\vec s_R$ in old
coordinates, the point $\vec y=\vec t-\vec s_R+F\vec z_0$ is within
$(2\sigma\tau+2|R|)\cdot P$ of $\vec y_0$, so by \eqref{eq:projected-core-buffer}
it and its reflection lie in $FJ$; thus \eqref{eq:reflection-band} applies to
every target and every reserve choice.

\emph{Prices and the largest ones.} Put $q_i=\psi(\vec{u}_i)-\vec{\beta}\cdot \vec{v}_i$
for $i\notin R$, and choose the reference bits (the \emph{cheap endpoints})
$b_i=0$ if $q_i\ge0$ and $b_i=1$ if $q_i<0$. Since
$\alpha=\psi(\vec{z}_0)-\vec{\beta}\cdot F\vec{z}_0+\sum_{i\notin R}\min(0,q_i)$,
the quantity in \eqref{eq:reflection-band} at the point
$\vec z_0+\sum_{i\notin R}x_i\vec u_i$ equals $\sum_{i\notin R}|q_i|\,[x_i\ne b_i]$,
a sum of nonnegative deficits. Let $D$ consist of $h$ items with largest $|q_i|$
(all unreserved items if fewer remain). Forcing the expensive endpoint of
every item of $D$ changes each support value of the projected core by at most
$hh_P$, so by \eqref{eq:projected-core-buffer} the forced core still contains
$\vec y_0$, and a fractional completion there has deficit at least
$\sum_{i\in D}|q_i|$ and at most $g(\vec y_0)$. Hence
\begin{equation}
 \sum_{i\in D}|q_i|\le\kappa,
 \qquad\text{so}\qquad
 |q_i|\le\kappa/h\quad(i\notin R\cup D).
 \label{eq:light-new-coordinate}
\end{equation}
If no item survives, an exact support computation resolves the instance.

\emph{The child.} Reanchor as in \eqref{eq:reanchor}: put
$\vec{b}=\sum_{i\notin R}b_i\vec{v}_i$, $\vec{v}_i'=(1-2b_i)\vec{v}_i$, and
$\zeta_i=(1-2b_i)\psi(\vec{u}_i)$, and compute
\begin{equation}
 \mathcal T'
 =\mathcal T-\mathcal{S}(\{\vec{v}_i:i\in R\})-\vec{b}-\mathcal{S}(\{\vec{v}_i':i\in D\})
 \label{eq:deep-old-target-update}
\end{equation}
by Lemma~\ref{lem:transfer}. The child items are
$\widehat{\vec{v}}_i=(\vec{v}_i',\zeta_i)\in\Lambda_0\times\Z$ for $i\notin R\cup D$,
and the child targets are
\begin{equation}
 \widehat{\mathcal T}
 =\{(\vec{y},z):\vec{y}\in\mathcal T',\ z\in\Z,
                 \ 0\le z-\vec{\beta}\cdot\vec{y}\le2\kappa\},
 \label{eq:deep-label-mask}
\end{equation}
filtered to the lattice $\widehat\Lambda$ generated by the child items. Note
$\zeta_i-\vec\beta\cdot\vec v_i'=|q_i|\ge0$.

\emph{Equivalence.} A child subset realizing $(\vec y,z)$ realizes
$\vec y\in\mathcal T'$, and \eqref{eq:deep-old-target-update} supplies
disjoint reserve and $D$ witnesses which, with the reference bits, give a
parent subset. Conversely, take a parent solution with its actual reserve
choice. Its unreserved part gives a point of $J$ to which
\eqref{eq:reflection-band} applies, so its deficits sum to at most $2\kappa$;
dropping the nonnegative deficits of $D$, the surviving part has
$z-\vec\beta\cdot\vec y\in[0,2\kappa]$ for its lifted sum $(\vec y,z)$. So
$(\vec y,z)\in\widehat{\mathcal T}$, and it passes the lattice filter.

\emph{Rank.} $\psi$ is nonzero on $\ker F$, so $(F,\psi)$ has rank $d+1$; the
vectors $\vec u_i$, $i\notin R\cup D$, span $\R^r$ by robustness, and their
images (with signs) span $\R^{d+1}$.

\emph{Size.} Each $\vec y\in\mathcal T'$ has at most $2\kappa+1$ admissible
values of $z$. With $s_{\mathrm{lift}}=\sigma\tau+|R|+h$ and a centre $\vec y_c$,
$\mathcal T'\subseteq\vec y_c+s_{\mathrm{lift}}\cdot P$, and the lifted targets
lie in a translate of
\[
 B=\{(\vec{x},z):\vec{x}\in {s_{\mathrm{lift}}}\cdot P,\ |z-\vec{\beta}\cdot\vec{x}|\le\kappa\}.
\]
The shear has determinant one and $\widehat\Lambda\subseteq\Lambda_0\times\Z$,
so $\vol(B)/\det\widehat\Lambda\le2\kappa\vol(s_{\mathrm{lift}}\cdot P)/\det\Lambda_0=O(w)$.
By \eqref{eq:light-new-coordinate} and $h=o(\tau)$, $h\widehat{\vec v}_i\in B$
for every child item, and the child items contain $d+1$ independent lattice
vectors. Lemma~\ref{lem:boxing} with $\ell=h$ gives an enclosing basis box
with $O(w)$ entries and half-widths $\omega_j\ge h$ (some child vector has a
nonzero $j$th coordinate and its $h$-fold multiple lies in the box); as in Step~5
of the boundary reduction, $p'_j=\lceil(\omega_j+1)/h\rceil\le3\omega_j/h$ gives an
item box $P'$ with targets in a translate of $h\cdot P'$ and size $O(w)$.

\emph{Cost.} Enumerate the old targets and their $O(1)$ values of $z$, filter
into $\widehat\Lambda$, and only then change basis, so that no old-lattice
convolution is charged in a coarser lattice. The $O(h)$ removed items cost
$\softO{n+w}$ by Lemma~\ref{lem:transfer}. Support evaluations for $J$ cost
$O(n)$, and the fiber oracle, the lattice-point-or-direction call, and the
computation of $\vec\beta$ cost $\softO{n}$ by Theorem~\ref{thm:import-geometry}
and Lemma~\ref{lem:bits}, which also bounds all bit lengths; coordinates are
stored afresh rather than composing earlier transformations.
\end{proof}

\section{Bounding dimension changes}\label{sec:invariant}

Interior reductions add coordinates and deletions may lower the dimension.
We bound both using short additive relations among the original integers.
Fix a rational $0<\delta\le1/20$, set $k=\lfloor N^{\delta/3}\rfloor$, and put
\begin{equation}\label{eq:rankconstants}
 d_*=\lfloor1+3/\delta\rfloor+1,\qquad
 q_*=\lceil4(1+3/\delta)\rceil,\qquad
 r_*=d_*+2q_*.
\end{equation}
$N$ is assumed large enough that $k\ge4$ and all inequalities below hold.
For the active original items $I$ let
\[
 S_I=\{0\}\cup\{a_i,-a_i:i\in I\},\qquad \mu(I)=|kS_I|,
 \qquad 1\le\mu(I)\le2kN+1.
\]
$\mu$ is used only in the analysis.

\begin{definition}[Short-relation invariant]
The current representation \emph{preserves signed $k$-sums} if the map
$\varphi$ induces a bijection from
$k(\{0\}\cup\{\vec{v}_i,-\vec{v}_i:i\in I\})$ onto $kS_I$.
\end{definition}

\begin{lemma}\label{lem:freiman}
The recursion preserves signed $k$-sums provided every GAP used in an
interior reduction is proper on a box of scale exceeding $2k$.
\end{lemma}

\begin{proof}
The root representation is the identity, and deletion, reorientation,
restriction to the span, and injective coordinate changes preserve the
property. For an interior reduction, an equality of two signed $k$-sums of
original integers gives, by induction, an equality of old vector sums, hence of
their scalar encodings; both GAP-coordinate sums lie in the box on which the GAP
map is injective, so they agree, and so do their values of $\psi$.
Conversely an equality after the lift holds in the retained old coordinates.
\end{proof}

\begin{lemma}[Dimension and total rank loss]\label{lem:rank}
Under the short-relation invariant the current rank is less than $d_*$, and
along any path the total rank lost by deletions is at most $q_*$. Hence a path
starting in rank one has at most $d_*-1+q_*$ interior reductions that add a
coordinate, and at most $r_*$ such reductions and rank-losing deletion events
together.
\end{lemma}

\begin{proof}
If $d$ current vectors are independent, their nonnegative combinations with
total coefficient at most $k$ are distinct, so
$\binom{k+d}{d}\le\mu(I)\le2kN+1$; this fails for $d=d_*>1+3/\delta$ and large $N$.

Suppose a deletion leaves $I'$ and lowers the rank by $q>0$. Choose $q$ deleted
vectors independent modulo the surviving span and put $m=\lfloor k/2\rfloor$.
Their nonnegative combinations of total at most $m$ give $\binom{m+q}{q}$
translates of the surviving signed $m$-sumset; they are disjoint because the
chosen vectors are independent modulo the surviving span, and they lie in the
old signed $k$-sumset. By the invariant,
\begin{equation}\label{eq:ranktranslate}
 \mu(I)\ge\binom{m+q}{q}\,|mS_{I'}|.
\end{equation}
Translating $S_{I'}$ into $[0,2N]$ and using $m\ge(2N)^{\delta/4}$,
\eqref{eq:carrier-doubling} gives $|kS_{I'}|\le\gamma_\delta|mS_{I'}|$ (as
$kS_{I'}\subseteq4mS_{I'}$). Since $q<d_*$,
\[
 \mu(I')\le \gamma _\delta\binom{m+q}{q}^{-1}\mu(I)\le k^{-q/2}\mu(I).
\]
Other deletions do not increase $\mu$ and coordinate changes do not affect it,
so telescoping gives $\sum q\le2\log(2kN+1)/\log k\le q_*$, the middle term
tending to $2(1+3/\delta)$. The final rank is one plus the number $u$ of added
coordinates minus the total loss, so $u\le d_*-1+q_*$; there are at most $q_*$
rank-losing events, and $d_*-1+2q_*<r_*$.
\end{proof}

\section{The recursive algorithm}\label{sec:recursion}\label{sec:complexity}

\paragraph{Constants.}
For the given $\varepsilon$, fix a rational
$0<\delta\le\min(\varepsilon/4,1/20)$ and the constants
\eqref{eq:rankconstants}; fix the GAP exponent $\eta=\delta/4$. Choose all
geometric and lattice constants uniformly over dimensions at most $r_*+2$ and
the resulting bounded GAP ranks (this does not presuppose the invariant), then
$\vartheta$ for Lemma~\ref{lem:normal} in all these dimensions, then the grid
constants, and finally a fixed ${c_{\mathrm{log}}}>2$ large enough for all buffers.
At a node of scale $\tau$ set
\[
 h=\left\lfloor \tau/\log^{c_{\mathrm{log}}}(N+2)\right\rfloor,
\]
and at the root $\tau_0=\lfloor N^\delta\rfloor$. The threshold on $N$ is
chosen so that $h\ge1$ on the prefix analysed in Lemma~\ref{lem:bootstrap}
and so that every new scale $\tau'=\Theta_\delta(h)$ satisfies $\tau'\le\tau$.

\paragraph{The algorithm.}
The root has the original integers as one-dimensional items, target set
$\{t\}$, $P_0=[-N,N]$, and allocated size
\begin{equation}\label{eq:rootalloc}
 w_0=O(\tau_0N)=O(N^{1+\delta}).
\end{equation}
At a node:
\begin{enumerate}[label=\arabic*.]
\item Remove zero items, restrict to the real span of the items, reject an empty
target set, and accept an empty item set iff $0$ is a target. Clip targets to
the coordinatewise attainable intervals.
\item If $n\le\tau$, apply Lemma~\ref{lem:terminal}.
\item Classify the targets as in Lemma~\ref{lem:normal}.
\item On each normal patch, apply the boundary reduction and recurse on its
children.
\item On a buffered interior patch, apply the vector preprocessing of
Section~\ref{sec:vector-preprocessing} and Theorem~\ref{thm:carrier} with
parameter $h$, and transfer the exceptional items by
Lemma~\ref{lem:transfer}. If no items remain, the transferred target set
decides the patch. If the old rank drops, restrict to the new span and recurse
(a \emph{rank-loss transition}).
\item Otherwise apply Lemma~\ref{lem:carrier-old-map} and
Lemma~\ref{lem:deep}, and recurse on the lifted child with scale
$\tau'=\Theta_\delta(h)$.
\end{enumerate}
A node accepts iff some child accepts. Steps 5--6 and rank losses in Step~4
are the \emph{exceptional transitions}.

\paragraph{Interior hypotheses.}
By the choice of grid in Section~\ref{sec:normal}, an interior cell has a
centre $\vec x$ with $\vec x+\vartheta\tau\cdot P\subseteq Z$ and targets in
$\vec x+\sigma\tau\cdot P$, $\sigma<\vartheta/100$. Then
transferring $e=O(h\log^2N)=o(\tau)$ exceptional items changes each support
value by at most $eh_P$ and each target by an element of $e\cdot P$. So the
surviving zonotope contains $\vec x+(\vartheta\tau-e)\cdot P$ and the targets lie
in $\vec x+(\sigma\tau+e)\cdot P$. For large $N$, $e\le\vartheta\tau/200$, so
Lemma~\ref{lem:deep} applies with $\vec o=\vec x$ and target radius $\vartheta/50$
in place of $\sigma$ (the buffer $10\cdot\vartheta/50<\vartheta/2$ is available). The filled box of Theorem~\ref{thm:carrier} has
scale $\gamma h\ge2r+1$, and condition~(3) gives the required robustness. Sign
reorientation translates targets, zonotope, and centre together.

\begin{lemma}[Correctness]\label{lem:correct}
Conditional on the success events of the color-coding calls and of the
(verified) GAP constructions, the recursion answers every instance correctly.
Unconditionally it has no false positive.
\end{lemma}

\begin{proof}
Terminal computations are exact, discarding cells outside $Z$ and targets
outside the span or the attainable intervals loses no subset sum, transfers are
exact by Lemma~\ref{lem:transfer}, and Lemmas~\ref{lem:normal}
and~\ref{lem:deep} give both directions for their children (the former on its
success event). Induction on the terminating recursion proves correctness;
soundness of every step does not depend on any random event.
\end{proof}

\begin{remark}[Randomness and witnesses]\label{rem:random}
Number the calls of Lemma~\ref{lem:quota} in execution order and give call $j$
fresh randomness and failure allowance $\xi/(2j^2)$; conditional on the past
its input is fixed, so all calls succeed with probability at least $1-\xi$.
Retaining the product trees and, at each union node, one contributing child,
lets us recover an actual subset for any retained sum within the same time.
\end{remark}

\begin{lemma}[Bit size]\label{lem:bits}
At a node of fixed dimension whose item coordinates and allocated size are
$N^{O_\delta(1)}$, the item and target coordinates, the scalar maps, $F$, and
all lattice descriptions can be stored with $O_\delta(\log N)$ bits per entry,
and all geometric and lattice operations take $n\polylog N$ bit time apart
from Boolean array operations. In the count representation, $n$ is replaced
by the number of distinct types, provided the total multiplicity is
$N^{O_\delta(1)}$.
\end{lemma}

The proof (Appendix~\ref{app:bit-details}) recomputes all descriptions from
the current vectors, so that coordinate changes are never composed.

\begin{lemma}[Noncircular choice of scales]\label{lem:bootstrap}
Every GAP call satisfies $h\ge L^{\eta}$, every GAP used in a
lift is proper on a scale exceeding $2k$, and every path has at most $r_*$
exceptional transitions.
\end{lemma}

\begin{proof}
Consider a path only up to its $(r_*+1)$st exceptional transition. Each lift
adds one coordinate and deletions never increase the rank, so dimensions on
this prefix are at most $r_*+2$, covered by the uniform choice of constants.
Normal steps keep $\tau$, lifts divide it by $O(\log^{c_{\mathrm{log}}}N)$, and
rank-loss restrictions keep it, so on the prefix
$\tau,h\ge N^\delta/\log^{O_\delta(1)}N$ and $h\le\tau\le\tau_0$.

Normal steps shrink the summed allocated size; a lift has size $O(w)$ by
Lemma~\ref{lem:deep}; a rank-loss branch transfers $e=o(\tau)$ items in the
old lattice, giving targets in a translate of $(\tau+e)\cdot P$, after which
restriction and injective projection do not enlarge the box, so it too has size
$O(w)$. Hence $w\le N^{1+\delta}\log^{O_\delta(1)}N\le N^{1+2\delta}$ on the prefix.
The encoding of
Section~\ref{sec:vector-preprocessing} has
\begin{equation}\label{eq:encoding}
 L=O_\delta(h^{d-1}|P|),\qquad w\ge c_d\tau^d|P|,
\end{equation}
so $L\le\gamma_\delta w/\tau\le N^2$, and $h\ge L^{\delta/4}$
for large $N$. The properness scale $\gamma h$ of Theorem~\ref{thm:carrier} exceeds $2k$. These bounds are
established for each call from its parent before the call is made. Hence
Lemma~\ref{lem:freiman} holds on the prefix, and Lemma~\ref{lem:rank} excludes
the $(r_*+1)$st exceptional transition.
\end{proof}

\begin{proof}[Proof of Theorem~\ref{thm:main}]
Normal children have total size at most $w/2$, and an
exceptional transition multiplies the summed size by at most
$a_N=\log^{O_\delta(1)}N$ (including the constant number of patches). Grouping
nodes by the number of exceptional transitions above them,
\begin{equation}\label{eq:totalalloc}
 \sum_{\text{nodes }\omega}w_\omega
 \le w_0\sum_{j=0}^{r_*}2^{j+1}{a_N}^j
 =\softO{w_0}.
\end{equation}
In particular the recursion terminates and has polynomially many nodes.

Normal processing costs $\softO{n+w}$ and exceptional transfers $\softO{w}$.
By \eqref{eq:nscan}, \eqref{eq:encoding}, and $h\le\tau$, the GAP construction
costs $\softO{nh+h^2L}=\softO{hw+h^{d+1}|P|}=\softO{hw}$,
and other operations are covered by Lemma~\ref{lem:bits}. Since $h\le\tau_0$,
the total time is $\softO{\tau_0w_0}=\softO{N^{1+2\delta}}=O_\varepsilon(N^{1+\varepsilon})$.

For error $N^{-c}$, give summable failure allowances to the color-coding and
Freiman calls (Remark~\ref{rem:random}); GAP outputs are verified within their
budget, and a failed verification or time-out ends the run with answer \emph{no}.
Lemma~\ref{lem:correct} completes the proof.
\end{proof}

\section{Multisets, vectors, and determinism}
\label{sec:deterministic}

We now prove Theorem~\ref{thm:main-vector} by running the recursion of
Section~\ref{sec:recursion} with three changes: items are stored as distinct
types with counts (Sections~\ref{sec:counts}--\ref{sec:count-reductions}),
the recursion starts at the original vectors with the invariant stated for
them (Section~\ref{sec:vector-invariant}), and the two randomized subroutines
are replaced by exact ones (Section~\ref{sec:derandomize}). Every other
ingredient (classification, packing, lattice boxing, the GAP construction,
the filled zonotope, the strip argument, the accounting) is used unchanged,
and we only point out what differs. Throughout, $d$ is the input dimension and
$\bar d$ the dimension of the current node; the local lemmas of the previous
sections are applied with their dimension parameter equal to $\bar d$.

\subsection{The count representation}\label{sec:counts}
For distinct vectors $\vec v_1,\ldots,\vec v_m$ and counts
$\vec c=(c_1,\ldots,c_m)$ let $V_{\vec c}$ contain $c_i$ copies of $\vec v_i$, so
$\mathcal{S}(V_{\vec{c}})=\{\sum_i x_i\vec{v}_i:x_i\in\Z,\ 0\le x_i\le c_i\}$.
Selecting or deleting copies decreases counts, and disjoint selections have
counts adding to at most $\vec c$.

\begin{lemma}[Exact count operations]\label{ms:representation}
Let the distinct nonzero types lie in $P$ and let
$\mathcal T\subset\vec{o}+\tau\cdot P$ have allocated size $w$. Then:
\begin{enumerate}[label=(\roman*)]
\item $m\le|P|$.
\item $\conv\mathcal{S}(V_{\vec c})=Z(V_{\vec{c}})=\sum_i[0,c_i\vec{v}_i]$ and
$h_{Z(V_{\vec{c}})}(\vec{a})=\sum_i c_i\max(0,\vec{a}\cdot \vec{v}_i)$.
\item For reference bits $b_i$, with $\vec{b}=\sum_i c_ib_i\vec{v}_i$ and
$\vec{v}_i'=(1-2b_i)\vec{v}_i$, $\mathcal{S}(V_{\vec{c}})=\vec{b}+\mathcal{S}(V'_{\vec{c}})$
via $x_i=c_ib_i+(1-2b_i)y_i$.
\item For counts $0\le e_i\le c_i$ with $e=\sum_ie_i=O(\tau)$,
$\mathcal{S}(V_{\vec{c}})=\mathcal{S}(V_{\vec{e}})+\mathcal{S}(V_{\vec{c}-\vec{e}})$,
and $\mathcal T-\mathcal{S}(V_{\vec{e}})$ can be computed in $\softO{m+w}$ time.
\end{enumerate}
\end{lemma}

\begin{proof}
(ii) holds because $\prod_i[0,c_i]$ is the convex hull of its integer points.
In (iv), a count $x_i$ splits as $\min(x_i,e_i)+(x_i-\min(x_i,e_i))$; expand
only the $e$ designated copies and apply Lemma~\ref{lem:transfer}.
\end{proof}

When types become equal (after reorientation or a coordinate change), we merge
them, adding their counts and recording the constituents $c_1,\ldots,c_s$; an
aggregate count $x\le\sum_jc_j$ is split back as
$x_j=\min(c_j,x-\sum_{i<j}x_i)$, and reanchoring is then undone constituentwise.
This is exact, costs $\softO{m}$ by sorting, and never increases the total
number of copies, which therefore stays at most $n$. Merged types have equal
oriented original vectors, so merging does not change the original support
used in the invariant. Support functions, reference sums, generated lattices,
and attainable intervals are computed by scans over types.

\subsection{The local reductions with counts}\label{sec:count-reductions}

\paragraph{The common GAP.}
Theorem~\ref{thm:carrier} is already stated for counts, and the proof of
Lemma~\ref{lem:carrier-old-map} applies verbatim (witnesses use $O(h)$ reserve
copies regardless of the multiplicities). In the vector preprocessing, equal
types are merged \emph{after} reorientation.

\begin{lemma}[Boundary reduction with counts]\label{ms:normal}
In the setting of Lemma~\ref{lem:normal}, let the items be $m$ distinct nonzero
types $\vec v_i$ with counts $c_i$. The same conclusions hold, with children
consisting of whole surviving types with their counts, and with cost
$\softO{m+w}$. If Lemma~\ref{lem:det-bounded-sums} replaces
Lemma~\ref{lem:quota}, the reduction is deterministic, every parent solution
survives, and the cost is $\softO{m+\tau^3w}$.
\end{lemma}

\begin{proof}
We list the changes to the proof of Lemma~\ref{lem:normal}.

\emph{Geometry.} $Z=\sum_i[0,c_i\vec v_i]$ has the support function of
Lemma~\ref{ms:representation}(ii), evaluated in $O(m)$ time. Steps~1--2 use only
this oracle, $M_{\vec x}\supseteq\varpi\tau\cdot P$, and the window, so they and
\eqref{eq:normal-macbeath-sum} hold with constants independent of the counts.

\emph{Reference and heavy types.} With the per-copy losses $g_{ij}^{x}$ of
\eqref{eq:item-loss}, choose one bit $b_i$ for all copies of type $i$,
minimizing $\sum_{j\le q}g^{b}_{ij}$, and let $\vec b=\sum_ic_ib_i\vec v_i$ be the
\emph{full} reference sum. A count vector $\vec x$ ($0\le x_i\le c_i$) has row-$j$
deficit $g_j(\vec x)=\sum_i\bigl((c_i-x_i)g^0_{ij}+x_ig^1_{ij}\bigr)$, which is at
most $\beta$ if $\vec x$ realizes a patch target. Since
$\sum_{j\le q}g_j(\vec x)\ge\sum_ic_i\sum_{j\le q}g^{b_i}_{ij}$, the reference has
total loss at most $q\beta$ whenever a solution exists; otherwise the patch is
discarded. Let $f_i\in[0,c_i]$ be the flip counts of $\vec x$ ($f_i=x_i$ if $b_i=0$,
else $c_i-x_i$), so that
$g_j(\vec x)=\sum_i\bigl((c_i-f_i)g^{b_i}_{ij}+f_ig^{1-b_i}_{ij}\bigr)$. Then
$\sum_if_ig^{b_i}_{ij}$ is at most the reference row loss $\sum_ic_ig^{b_i}_{ij}$,
and $\sum_if_ig^{1-b_i}_{ij}\le g_j(\vec x)$. Hence Step~3 applies with $|X|$
replaced by $\sum_{i\in H}f_i$, giving
\begin{equation}\label{ms:heavy-count}
 \sum_{i\in H}f_i\le2q\tau
\end{equation}
and \eqref{eq:normal-heavy-absolute}--\eqref{eq:normal-heavy-artificial} with
multiplicities $f_i$.

\emph{Heavy support.} By \eqref{ms:heavy-count}, instantiating
$\min(c_i,k_0)$ copies of $\vec v_i'=(1-2b_i)\vec v_i$ for each heavy type, where
$k_0=\lceil2\bar d\tau\rceil$, represents every relevant heavy change, and every
subset of these copies is available. Lemma~\ref{lem:quota} on them costs
$\softO{k_0m+w_{\vec x}^{\rm loc}}$. All copies of every heavy type are then
deleted, while $\vec b$ keeps the full counts. The residual filter uses
$\alpha_j=\sum_{i\notin H}c_ig^{b_i}_{ij}$: if a residual target $\vec r$ is the sum
of light flips with counts $f_i$, then
$\alpha_j-\vec a_j\cdot\vec r=\sum_{i\notin H}\bigl((c_i-f_i)g^{b_i}_{ij}+f_ig^{1-b_i}_{ij}\bigr)\in[0,\beta]$.
Soundness reconstructs the valid counts $x_i=c_ib_i+(1-2b_i)f_i$. Equal surviving
signed types are merged. The child box bounds concern one copy of each type, so
counts cause no enlargement.

\emph{Cost.} Since $m\le|P|$ and $w\ge c_{\bar d}\tau^{\bar d}|P|$,
the temporary copies number at most $k_0m=O_{\bar d}(w)$, and
$\sum_{\vec x}w^{\rm loc}_{\vec x}=O(w)$. For the deterministic variant, first merge
equal signed heavy types (types $\vec v_i=-\vec v_{i'}$ with opposite bits),
adding their capped counts, so that the types are distinct; then
Lemma~\ref{lem:det-bounded-sums} with this $k_0$ and the coordinate
budgets of Step~4 retains every required heavy sum, at cost
$\widetilde O(m+k_0^3w_{\vec x}^{\rm loc})$ per patch.
\end{proof}

\begin{lemma}[Interior reduction with counts]\label{ms:deep}
In the setting of Lemma~\ref{lem:deep}, let the items be $m$ types
$\vec v_i\in P\cap\Lambda_0$ with counts $c_i$, $Z=\sum_i[0,c_i\vec v_i]$, and let
the GAP data consist of reserve counts $0\le r_i\le c_i$ with
$|R|=\sum_ir_i=O(h)$ satisfying the filling
\eqref{eq:reserve-core-fill} for the reserved multiset, and the robustness that
for every $0\le e'_i\le c_i-r_i$ with $\sum_ie'_i\le h$, the vectors $\vec u_i$ with
$c_i-r_i-e'_i>0$ span $\R^r$. Then the conclusions of Lemma~\ref{lem:deep} hold,
with a child of whole surviving types and counts, and with cost $\softO{m+w}$.
\end{lemma}

Condition~(3) of Theorem~\ref{thm:carrier} is exactly this
robustness hypothesis, and the paragraph ``Interior hypotheses'' of
Section~\ref{sec:recursion} applies verbatim to $e$ exceptional copies.

\begin{proof}
We list the changes to the proof of Lemma~\ref{lem:deep}. Put $c'_i=c_i-r_i$ and
$J=\vec z_0+\sum_i[0,c'_i\vec u_i]$; the count version of
Lemma~\ref{lem:filled-zonotope} gives \eqref{eq:certified-core} for
$V_{\vec c}$, with corrections drawn from the reserved copies only. All support
functions are weighted, e.g.\
$h_{FJ}(\vec a)=\vec a\cdot F\vec z_0+\sum_ic'_i\max(0,\vec a\cdot\vec v_i)$, and $|R|$
now means $\sum_ir_i$. $G$ and $\alpha$
carry the weights $c'_i$. For counts $0\le x_i\le c'_i$ the deficit at
$\vec z_0+\sum_ix_i\vec u_i$ equals $\sum_i|q_i|f_i$, with flip counts $f_i=x_i$ if $b_i=0$ and
$f_i=c'_i-x_i$ otherwise.

Instead of $h$ items, sort the types by nonincreasing $|q_i|$ and take the first
$h$ copies in that order, i.e.\ counts $h_i\le c'_i$ with $\sum_ih_i=h$ (robustness
gives $\sum_ic'_i>h$). The forced-core argument gives $\sum_ih_i|q_i|\le\kappa$. If
$s_i=c'_i-h_i>0$, each of the $h$ removed copies has price at least $|q_i|$, so
$h|q_i|\le\kappa$. The child uses $\vec b=\sum_ib_ic'_i\vec v_i$, the target
update $\mathcal T'=\mathcal T-\mathcal S(R)-\vec b-\mathcal S(V'_{\vec h})$
(with $\vec h=(h_i)$ and $V'$ the signed types $\vec v_i'$), and types
$(\vec v_i',\zeta_i)$ with counts $s_i$.

Soundness: child counts $\ell_i\le s_i$, reserve counts $\bar r_i\le r_i$, and removed
flip counts $j_i\le h_i$ give original counts $\bar r_i+j_i+\ell_i$ or
$\bar r_i+c'_i-j_i-\ell_i$, both in $[0,c_i]$. Completeness: given parent counts
$y_i\le c_i$, put $\bar r_i=\min(y_i,r_i)$ and $x_i=y_i-\bar r_i\le c'_i$; then
$\vec z_0+\sum_ix_i\vec u_i\in J$ and its image is within
$(2\sigma\tau+2|R|)\cdot P$ of $\vec y_0$, so \eqref{eq:reflection-band}
applies and $\sum_i|q_i|f_i\le2\kappa$. Split $f_i=j_i+\ell_i$ with
$j_i=\min(f_i,h_i)$ and $\ell_i\le s_i$; copies of a type are interchangeable, so
this realizes the same vector sum, and the surviving part has
$z-\vec\beta\cdot\vec y=\sum_i\ell_i|q_i|\in[0,2\kappa]$. Robustness shows that the
types with $s_i>0$ span $\R^r$. Only $|R|+h$ copies are materialized,
and all oracles scan $m$ types.
\end{proof}

\subsection{Vector input and the dimension bound}\label{sec:vector-invariant}
Fix
\begin{equation}\label{eq:det-vector-parameters}
 0<\delta\le\min\left\{\frac{\varepsilon}{4(d+3)},\frac1{20d}\right\},
 \qquad k=\lfloor \nu ^{\delta/3}\rfloor,
 \qquad \eta=\frac{\delta}{8(d+1)};
\end{equation}
constants and the threshold on $\nu$ may depend on $d,\varepsilon,\delta$.

\begin{lemma}[Comparison of repeated vector sumsets]
\label{lem:vector-sum-comparison}
If $S\subseteq\prod_{j=1}^d[-N_j,N_j]\cap\Z^d$ contains zero, then
$|kS|\le \gamma _{d,\delta}|\lfloor k/2\rfloor S|$ for sufficiently large $\nu$.
\end{lemma}

\begin{proof}
After dropping coordinates with $N_j=0$, the mixed-radix map with radices
$t_j=8kN_j+1$ is injective on $\prod_j[-4kN_j,4kN_j]\cap\Z^d$. Translate the
image of $S$ by its minimum to a set $X\ni0$ of diameter
$L\le O_d(k^{d-1}\prod_j(N_j+1))\le\nu^{d+1}$. With
$s=\lfloor k/2\rfloor\ge L^{\eta}$ (or $X=\{0\}$),
\eqref{eq:carrier-doubling} gives $|4sX|\le\gamma_{d,\delta}|sX|$, and
$kX\subseteq4sX$. The encoding is injective on these sumsets.
\end{proof}

Each active type keeps its oriented original vector $\vec a_i$, and the
linear map $\varphi$ of Section~\ref{sec:representation}, now vector valued,
satisfies $\varphi(\vec v_i)=\vec a_i$. With
$S_I=\{\vec 0\}\cup\{\pm\vec a_i:i\in I\}$ and $\mu(I)=|kS_I|$, the
\emph{vector short-relation invariant} asks that $\varphi$ induce a bijection
\begin{equation}\label{eq:vector-short-invariant}
 k\bigl(\{\vec 0\}\cup\{\vec v_i,-\vec v_i:i\in I\}\bigr)\longrightarrow kS_I,
 \qquad 1\le\mu(I)\le(2k\nu +1)^d.
\end{equation}

\begin{lemma}[Vector rank and total rank loss]
\label{lem:vector-rank}
If every GAP used in a lift is proper on a scale exceeding $2k$, the recursion
preserves \eqref{eq:vector-short-invariant}, also under merging. With
$d_*=\lfloor d(1+3/\delta)\rfloor+2$, the current rank is less than $d_*$, the
total rank loss along a path is less than $4d_*$, and there are fewer than $9d_*$
coordinate introductions and rank-losing events on a path.
\end{lemma}

\begin{proof}
The proofs of Lemmas~\ref{lem:freiman} and~\ref{lem:rank} apply with the
following changes. The potential depends only on the distinct original vectors,
so counts and merging play no role (a rank loss must exhaust types). The upper
bound $2kN+1$ becomes $(2k\nu+1)^d$, and
Lemma~\ref{lem:vector-sum-comparison} replaces \eqref{eq:carrier-doubling}.
Then $\binom{k+d_*}{d_*}\gtrsim\nu^{\delta d_*/3}$ exceeds
$(2k\nu+1)^d=O_d(\nu^{d+d\delta/3})$, so the rank is less than $d_*$; and
$\sum q\le2d\log(2k\nu+1)/\log k<4d_*$. A path starting in rank at most $d$ thus
has fewer than $5d_*$ introductions and $4d_*$ rank-loss events.
\end{proof}

\subsection{Derandomization}\label{sec:derandomize}

\begin{lemma}[Deterministic bounded corrections]
\label{lem:det-bounded-sums}
Fix $d$. Let $\vec w_1,\ldots,\vec w_m\in\Z^d$ be distinct types with counts
$c_i\ge1$, and let $k_0\ge1$ and $\Delta_j\ge0$ be integers. One can
deterministically compute a set
$S\subseteq\mathcal{S}(W_{\vec c})\cap B$, where $B=\prod_{j=1}^d[-\Delta_j,\Delta_j]$,
containing every sum with a witness $0\le x_i\le c_i$ satisfying
$\sum_ix_i\le k_0$ and $\sum_ix_i|w_{ij}|\le\Delta_j$ for all $j$, in
$\widetilde O_d(m+k_0^{3}|B|)$ time.
\end{lemma}

\begin{proof}
Discard types with some $|w_{ij}|>\Delta_j$, cap counts at $k_0$, and split
the types into the $2^d$ coordinate orthants (zero counts as nonnegative). In
one orthant put $\vec u_i=(|w_{i1}|,\ldots,|w_{id}|)\in B_+=\prod_j[0,\Delta_j]$;
there are at most $|B_+|$ types. Working modulo
$(X_1^{\Delta_1+1},\ldots,X_d^{\Delta_d+1})$, let $E_a(\vec X)$ be the coefficient of
$Y^a$ in $\prod_i(1+Y\vec X^{\vec u_i})^{c_i}$ and
$\Pi_\ell(\vec X)=\sum_i c_i\vec X^{\ell\vec u_i}$, and compute
$E_0,\ldots,E_{k_0}$ by Newton's identities
\begin{equation}\label{eq:det-newton}
 E_0=1,\qquad
 aE_a=\sum_{\ell=1}^a(-1)^{\ell-1}\Pi_\ell E_{a-\ell}
 \quad(1\le a\le k_0).
\end{equation}
They hold for the list of individual copies, divisions by $a$ are exact over
$\Z$, and since all exponents are nonnegative, truncation never affects a
retained exponent. The coefficients of $E_a$ count selections of $a$ distinct
copies; they are at most $(k_0|B_+|)^{k_0}$, so each of the
$O(k_0^2)$ products costs $\widetilde O_d(k_0|B_+|)$ by exact
convolution. The positive coefficients of $E_0+\cdots+E_{k_0}$ give the
orthant's support (zero tests are over $\Z$, not modulo a prime). Restore signs
and add the $2^d$ supports by Boolean convolution, clipping to
$B$; every partial combination of a required witness obeys
the absolute budgets, so clipping keeps it, and every output uses available
copies.
\end{proof}

\begin{lemma}[Deterministic Freiman construction in a bounded interval]
\label{lem:det-freiman}
Fix $K\ge1$. Given $X\subseteq[0,L]\cap\Z$ with
$|X+X|\le K|X|$, one can deterministically construct a containing
GAP of rank $O_K(1)$ and coefficient volume $O_K(|X|)$, with
$O_K(\log(L+2))$-bit entries, in
$\widetilde O_K(L+1)$ time.
\end{lemma}

\begin{proof}
Assume $|X|\ge2$. Compute $Y=8X-8X\subseteq[-8L,8L]$ by
convolution; $|Y|\le K^{16}|X|$ by Pl\"unnecke--Ruzsa. Sieve primes
$8(|Y|+1)<p<16(|Y|+1)$ and $q>\max\{16L,2p\}$ with
$q=O(L+1)$. For each $e\in Y\setminus\{0\}$ mark the multipliers
\begin{equation}\label{eq:det-forbidden-multipliers}
 \lambda\equiv tp\,e^{-1}\pmod q,
 \qquad 1\le t\le\lfloor(q-1)/p\rfloor.
\end{equation}
Fewer than $(q-1)/8$ marks are made, so some $\lambda\in[1,q-1]$ is unmarked.
Let $\bar a\in[0,q-1]$ represent $\lambda a\bmod q$, split $[0,q-1]$ into eight
intervals of diameter less than $q/8$, let $X'\subseteq X$ be the values in a most
popular interval ($|X'|\ge|X|/8$), and put $\pi(a)=\bar a\bmod p$. Then $\pi$ is a
Freiman $8$-isomorphism on $X'$: for two $8$-term lists with original difference
$e$ and difference $z$ of the sums of the $\bar a$, we have $|z|<q$ and $z\equiv\lambda e\pmod q$; $e=0$
forces $z=0$, while $e\ne0$ with equal $\pi$-sums would make $|z|$ a nonzero
multiple of $p$ representing $\pm\lambda e$, so $\lambda$ would be marked.
Padding lists with a fixed element of $X'$ shows that $\pi$ is injective.

Apply the dense Fourier construction of~\cite[Appendix A.2]{RW} to
$\pi(X')\subseteq\Z/p\Z$, of density $\theta=|X'|/p=\Omega_K(1)$.
Fourier thresholds are compared with fixed slack (include all frequencies of
normalized amplitude at least $\theta^{3/2}/2$ and only ones of at least
$\theta^{3/2}/4$, using approximations with error $\theta^{3/2}/8$); Parseval bounds
their number by $16/\theta^2$ and the excluded mass in the Bogolyubov argument by
$\theta^4/4$, so $O(\log p)$-bit precision suffices. If no frequency is selected,
$\Z/p\Z\subseteq2\pi(X')-2\pi(X')$ and we take the whole group; otherwise apply
the lattice construction of~\cite[Appendix A.3]{RW}. Either way we obtain a
proper bounded-rank progression of size $\Omega(|X|)$ in $2\pi(X')-2\pi(X')$.
Lift it to a GAP $Q_0\subseteq2X'-2X'$ as in~\cite[Appendix A.5]{RW}, tabulating
the inverse of the modelling map on $2X'-2X'$. Finally choose a maximal family
of disjoint translates $Q_0+t$, $t\in T\subseteq X$; it has $O(1)$ members since
they lie in $Q_0+X$ of size $O(|X|)$, and $X\subseteq Q_0-Q_0+T$. The difference
GAP with one $\{0,1\}$ coefficient per element of $T$ is the output. Each step is
a dense convolution or a scan within $\widetilde O(L+1)$.
\end{proof}

\begin{corollary}[Deterministic common GAP and filling reserve]
\label{cor:det-carrier}
Theorem~\ref{thm:carrier} holds with a deterministic algorithm and the same
time bound.
\end{corollary}

\begin{proof}
Every Freiman call in the proof of Theorem~\ref{thm:carrier} is on a
bounded-doubling subset of $[0,10sL]$; Lemma~\ref{lem:det-freiman}
handles it in $\widetilde O(hL)$ time, and there are $O(\log(sL))$ calls.
All other steps are deterministic.
\end{proof}

\subsection{The recursion and its cost}

\paragraph{Preprocessing.}
Aggregate the input into types by an array indexed by the input box, in
$O_d(n+\prod_j(2N_j+1))$ time; delete zero types, drop coordinates with
$N_j=0$ after checking the target, and reject targets outside the attainable
intervals. For the finitely many $\nu$ below the threshold there are $O(1)$
types, and fixed-dimensional integer feasibility on their counts takes
$O(n+1)$ time. Counts are at most $n$ throughout.

\paragraph{The recursion.}
Start with $P_0=\prod_j[-N_j,N_j]$, target set $\{\vec t\}$, and
$\tau_0=\lfloor\nu^\delta\rfloor$, restricted to the span, so
\begin{equation}\label{eq:vector-root-size}
 w_0=O_d\Bigl(\tau_0^d\prod_{j=1}^d(2N_j+1)\Bigr).
\end{equation}
At a node of dimension $\bar d$ with item box $P$ use
$h=\lfloor\tau/\log^{c_{\mathrm{log}}}(\nu+2)\rfloor$ and the algorithm of
Section~\ref{sec:recursion}, with the count operations of
Section~\ref{sec:counts}, Step~2 applied when $\sum_ic_i\le\tau$, the
deterministic boundary reduction of Lemma~\ref{ms:normal} (cost
$\widetilde O(m+\tau^3w)$), Corollary~\ref{cor:det-carrier} and
Lemma~\ref{ms:deep} for interior patches, and the vector preprocessing applied
to the whole current vector. The map $F$ recovers the entire old vector and all
permitted values of the new coordinate are retained, so all coordinates refer to
one selection of original copies. As in \eqref{eq:encoding},
\begin{equation}\label{eq:det-local-size}
 m\le|P|,\qquad L=O(h^{\bar d-1}|P|),\qquad
 w\ge c_{\bar d}\tau^{\bar d}|P|,\qquad L=O(w/\tau),
\end{equation}
so the GAP construction costs $\widetilde O(mh+h^2L)=\widetilde O(hw)$, and the
other interior operations $\widetilde O(m+w)$.

\begin{lemma}[Validity of vector calls and termination]
\label{lem:det-vector-bootstrap}
With constants chosen uniformly over dimensions at most $d+9d_*+2$, every GAP
call satisfies $h\ge L^{\eta}$, every lifted GAP is proper on a scale
exceeding $2k$, scales never increase, and the recursion is exact, terminates,
and satisfies $\sum_{\text{nodes }u}w_u=\widetilde O(w_0)$.
\end{lemma}

\begin{proof}
Argue as in Lemma~\ref{lem:bootstrap} on a prefix ending at a putative
$(9d_*+1)$st exceptional transition, where ranks are at most $d+9d_*+1$ and
$\nu^\delta/\log^{O(1)}\nu\le h\le\tau\le\tau_0$ and
$w\le w_0\log^{O(1)}\nu$. By \eqref{eq:vector-root-size},
\eqref{eq:det-local-size}, and $d\delta\le1/20$, $L\le\nu^{d+1}$, so
$L^{\eta}\le\nu^{\delta/8}\le h$ and the properness scale $\Omega(h)$
exceeds $2k$. Exceptional sets have $O(h\log^2\nu)=o(\tau)$ copies, so interior
buffers survive their transfer, and rank-loss branches enlarge allocations by
$O(1)$ factors. Lemma~\ref{lem:vector-rank} then excludes the excess transition.
The summation \eqref{eq:totalalloc}, with $r_*$ replaced by $9d_*$, gives the
total allocation and termination. Exactness follows as in
Lemma~\ref{lem:correct}, now without any probabilistic event.
\end{proof}

\begin{proof}[Proof of Theorem~\ref{thm:main-vector}]
After preprocessing, every scan is over types and $m=O(w)$ at each node, so a
node costs $\widetilde O(\tau_0^3w)$ apart from bit-length factors, and by
Lemma~\ref{lem:det-vector-bootstrap} and \eqref{eq:vector-root-size} the
recursion costs
\begin{equation}\label{eq:det-vector-total}
 \widetilde O\left(\tau_0^3w_0\right)
 =\widetilde O\left(\nu ^{(d+3)\delta}\prod_j(2N_j+1)\right).
\end{equation}
Counts are at most $n$, so Lemma~\ref{lem:bits} applies with $\log(n\nu+2)$ in
place of $\log N$; $\varphi$ is recomputed by Cramer's rule from independent
current types, and the Newton coefficients were already charged in
Lemma~\ref{lem:det-bounded-sums}. Thus for a fixed $b$ the recursion costs
$O(\nu^{(d+3)\delta}\prod_j(2N_j+1)\log^b(n\nu+2))$. If $n\le\nu^{2(d+\varepsilon+1)}$
the logarithm is a power of $\log\nu$, absorbed by the slack in
\eqref{eq:det-vector-parameters}; otherwise
$\nu^{(d+3)\delta}\prod_j(2N_j+1)=O(\sqrt n)$ and $\log^b(n\nu+2)=O(\sqrt n)$, so
the work is $O(n)$. Adding the preprocessing proves the bound; $N_j=N$ gives
the box bound and $d=1$ gives Corollary~\ref{thm:main-multiset}.
\end{proof}

\section{A constructive John-type theorem}\label{sec:john}

For a multiset $U$ of vectors, $Z(U)=\sum_{\vec u\in U}[\vec 0,\vec u]$ is its
zonotope; zero items may be discarded throughout.

\begin{theorem}[John-type characterization]\label{thm:john}
Fix $0<\eta<1$. Let $N$ be sufficiently large in terms of $\eta$,
let $N^\eta\le h\le N$ be an integer, and let $A$ be a multiset of
$n$ nonzero integers in $[-N,N]$, with $m$ distinct values.
There is a partition $A=E\uplus B$ with
$|E|=O_\eta(h\log^2(N+2))$
such that either $B$ is empty or the following holds. There are
a rank $r=O_\eta(1)$, a symmetric box
\begin{equation}\label{eq:john-box-size}
 C=\prod_{j=1}^r[-b_j,b_j],\qquad b_j\in\Z_{\ge1},\qquad
 \prod_{j=1}^r b_j=O_\eta\!\left(\frac{N}{h^{r-1}}\right),
\end{equation}
a multiset $U\subset C\cap\Z^r$, and an integer linear map $\rho:\Z^r\to\Z$
mapping $U$ to $B$ with multiplicities preserved. There are also a
submultiset $R\subseteq U$ of size $|R|=O_\eta(h)$ and a point
$\vec{z}_0\in\mathcal{S}(R)$ such that $J=\vec{z}_0+Z(U\setminus R)$ satisfies
\begin{equation}\label{eq:john-sandwich}
 J\cap\Z^r\ \subseteq\ \mathcal{S}(U)
 \ \subseteq\ Z(U)\cap\Z^r,
 \qquad J\subseteq Z(U)\subseteq J+|R|\cdot C.
\end{equation}
Consequently,
\begin{equation}\label{eq:john-original}
 \mathcal{S}(E)+\rho(J\cap\Z^r)
 \ \subseteq\ \mathcal{S}(A)
 \ \subseteq\ \mathcal{S}(E)+\rho(Z(U)\cap\Z^r).
\end{equation}
The partition and the descriptions of $C,U,\rho,R,\vec{z}_0$ can be constructed
deterministically in $O(n)+\widetilde O_\eta(mh+h^2N)$ time,
storing repeated vectors by their multiplicities.
\end{theorem}

By \eqref{eq:john-box-size}, $|C\cap\Z^r|=\prod_j(2b_j+1)=O_\eta(N/h^{r-1})$, and
the correction box $|R|\cdot C$ has $O_\eta(hN)$ lattice points; the whole boundary
region $Z(U)\setminus J$ may still depend on the multiplicities. Taking
$h=\lceil N^\eta\rceil$ gives construction time $O(n)+\widetilde O_\eta(N^{1+2\eta})$
and exceptional shifts of size $\widetilde O_\eta(N^{1+\eta})$. The exact identity
is $\mathcal{S}(A)=\mathcal{S}(E)+\rho(\mathcal{S}(U))$; the lattice condition in
\eqref{eq:john-original} is essential, since projection creates holes even in
the middle of the range.

\begin{corollary}[Interior points are subset sums]\label{cor:john-interior}
In the nonempty outcome of Theorem~\ref{thm:john}, every $\vec{z}\in\Z^r$ with
$\vec{z}+|R|\cdot C\subseteq Z(U)$ belongs to $J$ and hence to $\mathcal{S}(U)$, so
$e+\rho(\vec{z})\in\mathcal{S}(A)$ for every $e\in\mathcal{S}(E)$. After the
construction, membership in $J$ and this containment can each be tested
deterministically in $O_\eta(\ell_{\vec{z}}+m\log^{O_\eta(1)}(nN+2))$ bit
operations per query, where $\ell_{\vec z}$ is the encoding length of $\vec z$.
\end{corollary}

\begin{proof}[Proof of Theorem~\ref{thm:john} and Corollary~\ref{cor:john-interior}]
First let all inputs be positive. Apply Corollary~\ref{cor:det-carrier} with
$L=N$; its exceptions form $E$, the coefficient vectors of
the other items form $U$, and the reserve forms $R$, with $|R|=O_\eta(h)$ and
$\vec{z}_0+(\gamma h\cdot C\cap\Z^r)\subseteq\mathcal{S}(R)$. Since $\rho$ is
injective on $\gamma h\cdot C\cap\Z^r$, the filled box gives distinct subset sums
of $|R|$ items, all in an interval of length $|R|N$, so
\[
 (\gamma h)^r\prod_{j=1}^r b_j
 \le |\gamma h\cdot C\cap\Z^r|\le |R|N+1=O_\eta(hN),
\]
which is \eqref{eq:john-box-size}. For large $N$, $\gamma h\ge r$, and
Lemma~\ref{lem:filled-zonotope} with $\Lambda=\Z^r$ gives
\eqref{eq:john-sandwich} and the first assertion of the corollary; the identity
$\mathcal{S}(A)=\mathcal{S}(E)+\rho(\mathcal{S}(U))$ gives \eqref{eq:john-original}.

For signed input, run the construction on absolute values, keeping the counts of
both signs, and assign selected exceptional and reserve copies to these counts.
Replace the coefficient vector of each remaining negative item by its negative;
this preserves $C$ and makes $\rho$ return the signed value. If $R$ is the reserve
before reflection and $R_-\subseteq R$ its part coming from negative items, the
reflected reserve has subset sums $\mathcal{S}(R)-\Sigma(R_-)$, so its filled box is centred at
$\vec{z}_0-\Sigma(R_-)$ and Lemma~\ref{lem:filled-zonotope} applies again.

The time bound is that of Theorem~\ref{thm:carrier} plus $O(n+N)$ to group the
input. $J$ is stored by its centre and its types with counts, which suffices for
its support oracle. Its generators have coordinates polynomial in $N$ and
multiplicities at most $n$, so the oracle scans $m$ types with
$O_\eta(\log(nN+2))$-bit data; after rejecting queries outside the coordinate
bounds of $Z(U)$, Theorem~\ref{thm:import-geometry} and the oracle construction
of Lemma~\ref{lem:bits} give the stated membership cost. Finally,
$\vec z+|R|\cdot C\subseteq Z(U)$ iff all $2^r$ vertices of this box lie in $Z(U)$.
\end{proof}

\section{Integer programming and further consequences}
\label{sec:ip-consequences}

Throughout this section $d$ counts the rows of the matrix; the variable bounds
are not counted.

\begin{importedtheorem}[Proximity for bounded integer programming;
  Eisenbrand--Weismantel {\cite[Theorem 7]{EW}}]
\label{thm:import-ip-proximity}
Let $\mathsf A\in\Z^{d\times n}$ have entries of absolute value at most
$N$, and let $\vec t\in\Z^d$ and $\vec\ell,\vec u\in\Z^n$ with
$\vec\ell\le\vec u$. If
$\mathsf A\vec x=\vec t$, $\vec\ell\le\vec x\le\vec u$, $\vec x\in\Z^n$
is feasible, then for every vertex $\vec x^*$ of its LP relaxation there is a
feasible integer $\vec z$ with $\|\vec z-\vec x^*\|_1\le d(2dN+1)^d$.
\end{importedtheorem}

The cited theorem compares optimal fractional and integral solutions; the
feasibility version follows with objective zero after translating the lower
bounds.

\begin{proof}[Proof of Corollary~\ref{cor:bounded-ip}]
Solve the LP relaxation; if it is empty, reject, and otherwise compute a vertex
$\vec x^*$ (it exists since the bounds are finite). Put $r=d(2dN+1)^d$,
$\ell'_i=\max\{\ell_i,\lceil x_i^*-r\rceil\}$, and
$u'_i=\min\{u_i,\lfloor x_i^*+r\rfloor\}$, rejecting if some $\ell'_i>u'_i$.
By Theorem~\ref{thm:import-ip-proximity} the smaller box preserves
feasibility, and $0\le u'_i-\ell'_i\le2r$. Let $B$ contain $u'_i-\ell'_i$ copies of
column $\vec a_i$ (with labels), and $\vec t'=\vec t-\mathsf A\vec\ell'$; then
$\vec t'\in\mathcal S(B)$ iff some integer $\vec x$ with
$\vec\ell'\le\vec x\le\vec u'$ satisfies $\mathsf A\vec x=\vec t$, via
$x_i=\ell'_i+y_i$. $B$ has at most $2nr$ copies with entries bounded by $N$;
after rejecting targets outside the attainable intervals,
Theorem~\ref{thm:main-vector} decides feasibility with
$O_{d,\varepsilon}(nN^d+N^{d+\varepsilon})$ operations on
$O_d(\log(nN+2))$-bit words, i.e.\
$O_{d,\varepsilon}(N^{d+\varepsilon}\ell_{\mathrm{in}}^{O(1)})$ bit operations, as
$n\le\ell_{\mathrm{in}}$. A feasible vector is found by bisecting the range of one
variable at a time with the decision procedure, using $O(n\log(r+2))$ calls.
\end{proof}

\begin{corollary}[Prescribed cardinalities]
\label{cor:vector-cardinality}
Fix $d,q\ge1$ and $\varepsilon>0$. Let
$A=\{\vec a_1,\ldots,\vec a_n\}$ be an explicitly listed multiset
with $|a_{ij}|\le N_j$, and let $G_1,\ldots,G_q\subseteq[n]$. Given
$\vec t\in\Z^d$ and integers $k_1,\ldots,k_q$, one can decide whether some
$I\subseteq[n]$ has $\sum_{i\in I}\vec a_i=\vec t$ and $|I\cap G_s|=k_s$ for all
$s$, in deterministic time
$O_{d,q,\varepsilon}(n+\nu ^\varepsilon\prod_{j=1}^d(2N_j+1))$,
$\nu =\max\{2,N_1,\ldots,N_d\}$.
\end{corollary}

\begin{proof}
Append the coordinates $\mathbf1_{i\in G_1},\ldots,\mathbf1_{i\in G_q}$ to each
$\vec a_i$ and $(k_1,\ldots,k_q)$ to the target; this multiplies the box volume
in Theorem~\ref{thm:main-vector} by $3^q$.
\end{proof}

\begin{corollary}[Inequalities and discrepancy]
\label{cor:ip-inequalities}
Corollary~\ref{cor:bounded-ip} also holds for $d$ constraints each of which is
an equality, a one-sided inequality, or a two-sided inequality, with finite
variable bounds and coefficient bound $N\ge2$. Consequently, for
$\mathsf A\in\Z^{d\times n}$ with $|\mathsf A_{ji}|\le N$, a vector
$\vec x\in\{0,1\}^n$ minimizing $\|2\mathsf A\vec x-\mathsf A\vec1\|_\infty$ can be
found in deterministic $O_{d,\varepsilon}(N^{d+\varepsilon}\ell_{\mathrm{in}}^{O(1)})$
bit time.
\end{corollary}

\begin{proof}
For $\vec a^{\,T}\vec x\le b$, let $m=\sum_i\min\{a_i\ell_i,a_iu_i\}$; reject if
$b<m$, and otherwise use $\vec a^{\,T}\vec x+s=b$ with $0\le s\le b-m$. Lower
bounds are negated, and $a\le\vec a^{\,T}\vec x\le b$ becomes an equality with a
slack in $[0,b-a]$. This adds unit columns but no rows. For exact vector
partition, query $\Sigma(A)/2$ in Theorem~\ref{thm:main-vector} directly. For the
discrepancy threshold $k$, apply the result to
$-k\vec1\le2\mathsf A\vec x-\mathsf A\vec1\le k\vec1$ (coefficient bound $2N$) and
binary search over $k\in[0,nN]$.
\end{proof}

\begin{corollary}[Bounded-coefficient objectives]
\label{cor:ip-bounded-objective}
In the setting of Corollary~\ref{cor:bounded-ip}, let $\vec c\in\Z^n$ satisfy
$|c_i|\le p$, and put $\nu =\max\{2,N,p\}$. One can optimize
$\vec c^{\,T}\vec x$ over the feasible integer vectors, or report infeasibility,
in deterministic $O_{d,\varepsilon}(\nu ^{d+1+\varepsilon}\ell_{\mathrm{in}}^{O(1)})$
bit time.
\end{corollary}

\begin{proof}
Binary search on the threshold $\vec c^{\,T}\vec x\le k$ (or $\ge k$), an extra
row with coefficient bound $\nu$, between the extreme values of the objective on
the variable box, using Corollary~\ref{cor:ip-inequalities} in dimension $d+1$.
\end{proof}

\paragraph{Partition and two-machine scheduling.}
Equal-sum Partition is the single query $t=\Sigma(A)/2$, so it takes
$O_\varepsilon(n+N^{1+\varepsilon})$ deterministic time. For makespan minimization
on two identical machines, $P2\Vert C_{\max}$ (equivalently minimum-discrepancy
number partitioning~\cite{KorfSchreiber}), let $w=\Sigma(A)$; a makespan bound
$b_{\max}\ge w/2$ is feasible iff some subset sum lies in $[w-b_{\max},b_{\max}]$.
Put $r=2b_{\max}-w$. If $r\ge N$, greedy assignment to the less loaded machine
certifies feasibility. If $0\le r<N$, append $O(\log(N+2))$ items at most $N$
whose subset sums are exactly $[0,r]\cap\Z$ and query $b_{\max}$. Binary search
gives the optimal makespan and minimum discrepancy in
$\widetilde O_\varepsilon(n+N^{1+\varepsilon})$ deterministic time. Dense subset-sum
structure was used earlier for approximating Partition and Knapsack by Deng,
Jin, and Mao~\cite{DJM}.

\paragraph{Binary-encoded multiplicities.}
For scalar input with multiplicities in binary, a proximity argument (compare
\cite[Lemma 2.1]{PRW}) reduces the instance before the recursion.

\begin{lemma}[Normalization and a bounded residual instance]\label{lem:multiset-cap}
A signed scalar multiset instance can be rejected or transformed into an
equivalent instance with target $r\in[0,N-1]$, at most $2N$ nonzero signed types
in $[-N,N]$, multiplicity at most $2N$ for each type, and hence at most $4N^2$
copies in total, together with a map from
residual solutions back to original counts. For explicitly listed input this
takes $O(n+N)$ word operations.
\end{lemma}

\begin{proof}
Let $c_a^\pm$ be the counts of $\pm a$, $1\le a\le N$. Replacing a selected
negative count $x_a^-$ by $c_a^--x_a^-$ turns the instance into positive types
$a$ with counts $c_a=c_a^++c_a^-$ and target $t^+=t+\sum_aac_a^-$ (a count
$y_a\le c_a$ splits back as $x_a^+=\min(y_a,c_a^+)$). Reject unless
$0\le t^+\le\sum_aac_a$. Take a prefix anchor $\vec b$: whole multiplicities of
the types in some order until the first type that does not fit entirely, as many
copies of it as fit, and none of the rest. Then $r=t^+-\sum_aab_a\in[0,N)$, and
solutions correspond to residual changes using at most $c_a-b_a$ copies of $+a$
and $b_a$ copies of $-a$.

A residual solution with the fewest changes has no nonempty zero-sum
submultiset. Order its changes by taking a positive term when the partial sum is
at most $0$ and a negative term otherwise (while both remain); all partial sums
lie in $[1-N,N]$ and are distinct, so there are at most $2N-1$ changes. Hence
capping every residual multiplicity at $2N$ preserves feasibility.
\end{proof}

\begin{corollary}[Binary-encoded multiplicities]\label{ms:binary}
If weights, multiplicities, and target of a scalar instance are given in binary
with total length ${\ell_{\mathrm{in}}}$, subset sum can be decided deterministically
in $\softO{{\ell_{\mathrm{in}}}}+O_{\varepsilon}(N^{1+\varepsilon})$ bit operations.
\end{corollary}

\begin{proof}
Merge equal signed values and compute the weighted totals with a balanced
addition tree, whose levels have total operand length $\softO{\ell_{\mathrm{in}}}$;
locate the crossing type of the prefix anchor by descending the tree, with
$O(\log(m+2))$ large-integer operations and one division. This performs the
reduction of Lemma~\ref{lem:multiset-cap} in $\softO{\ell_{\mathrm{in}}}$ bit
operations. We do not invoke the stated bound of Theorem~\ref{thm:main-vector},
whose term $n$ would here be $\Theta(N^2)$. Instead we run the count recursion of
Section~\ref{sec:deterministic} directly on the $m=O(N)$ residual types with their
counts. In that proof the term $n$ arises only from aggregating explicitly listed
copies and from the case split on $\log(n\nu+2)$; both are absent here, since the
total multiplicity is at most $4N^2$ and Lemma~\ref{lem:bits} applies with
$O(\log N)$-bit words. The recursion therefore takes
$\widetilde O(N^{1+(d+3)\delta})=O_\varepsilon(N^{1+\varepsilon})$ time with $d=1$.
\end{proof}

\section{Scope}
\label{sec:scope}
Theorem~\ref{thm:main} is the randomized scalar set case;
Section~\ref{sec:deterministic} adds multiplicities, original vector input, and
deterministic subroutines, proving Theorem~\ref{thm:main-vector} and
Corollary~\ref{thm:main-multiset}; Section~\ref{sec:ip-consequences} derives
bounded integer programming via LP proximity, with coefficient exponent
$d+\varepsilon$ and no claim about the dependence on $d$ alone.

The proof constructs one common GAP after excluding few items, and one reserve
whose sums fill a box; it does not need every random sample to fill the same
GAP, nor Fourier uniformity. Interior always refers to the current zonotope and
its buffer, integer points in GAP coordinates are accepted only through the
certified core inclusion, and a noninjective lift keeps the old coordinates and
every permitted value of the new coordinate.

\paragraph{Lean verification scope.}
An accompanying Lean~4 development machine-checks the paper-specific
scalar and multidimensional subset-sum proof chain, conditional on the
correctness of the cited lemmas and theorems and on the stated correctness
and cost contracts for external algorithmic primitives. For these
conditional endpoints, the checked complexity bounds use the
formalization's abstract storage and primitive-cost model. They do not
verify implementations of every cited primitive or native Lean execution
time.

\clearpage
\appendix
\section{Boolean products and signed color coding}
\label{app:boolean-operations}

This appendix proves Lemma~\ref{lem:quota} and supplies its two
auxiliary lemmas. The ingredients are standard balanced Boolean
convolution and two-stage color coding; the details below account for
signed coordinates, narrow array dimensions, and simultaneous retention
of all sums satisfying the absolute budgets.

For a finite integer set, its Boolean support is its indicator array.
We compute the support of $A+B$ by ordinary integer convolution followed
by replacing every positive coefficient by one. Signed coordinates are
handled by array offsets. In fixed dimension, these operations take time
nearly linear in the enclosing array size.

\begin{lemma}[Balanced Boolean products]\label{lem:boolean-product}
Fix $d$ and $q\ge1$, and let $A_1,\ldots,A_q$ be explicitly listed finite sets
contained in $\prod_{j=1}^d[-a_j,a_j]\cap\Z^d$, where $a_j\ge0$.
Writing $p=\sum_i|A_i|$, the support of $A_1+\cdots+A_q$ can be
computed exactly in
\[
 \widetilde O_d\left(p+q+\prod_{j=1}^d(1+qa_j)\right)
\]
time. The same bound suffices to retain all intermediate supports.
\end{lemma}

\begin{proof}
Pad the list by copies of $\{0\}$ until its length is a power of two;
this changes $q$ by at most a factor of two. Multiply the factors in a
balanced binary tree. At a level with $b$ nodes, each node represents
$q/b$ factors, and its sums lie in the box with coordinate bounds
$(q/b)a_j$. Put $e_j=qa_j$. The total array size at this level is
bounded, up to a dimension-dependent factor, by
\begin{equation}\label{eq:productbound}
 b\prod_{j=1}^d(1+e_j/b)
 =b+\sum_{\varnothing\ne J\subseteq[d]}b^{1-|J|}\prod_{j\in J}e_j
 \le b+\prod_{j=1}^d(1+e_j)-1.
\end{equation}
This estimate includes the constant terms from narrow or zero-width
coordinates. Sum it over the logarithmically many levels and add the
cost of reading the factors. Booleanizing after every convolution keeps
its coefficients bounded by an array size, so their bit lengths are
logarithmic. Exact integer convolution therefore has the stated cost.
Keeping the arrays at every level changes the space bound by only a
logarithmic factor.
\end{proof}

\begin{lemma}[Color coding within one size class]
\label{lem:one-class-color-coding}
Fix $d$, positive integer budgets $\Delta_1,\ldots,\Delta_d$, and
$B=\prod_j[-\Delta_j,\Delta_j]$. Let $\vec{w}_1,\ldots,\vec{w}_m\in\Z^d$ satisfy
$|w_{ij}|\le\alpha \Delta_j$ for some $0<\alpha\le1$. Let
$1\le k\le m$ be an integer with $k\alpha\le2d$.
For $0<\xi<1$, one can compute a set
$S\subseteq\mathcal{S}(\{\vec{w}_1,\ldots,\vec{w}_m\})\cap B$ in
\[
 \widetilde O_d(m+|B|)
\]
time, with polylogarithmic dependence on $1/\xi$, such that for every
fixed $X\subseteq[m]$ with $|X|\le k$ and $\sum_{i\in X}\vec{w}_i\in B$,
\[
 \Pr\left[\sum_{i\in X}\vec{w}_i\in S\right]\ge1-\xi.
\]
Every output has a witness using distinct input items.
\end{lemma}

\begin{proof}
Put
$\lambda=\lceil \gamma _d\log((m+|B|+2)/\xi)\rceil$, with $\gamma _d$ a
sufficiently large constant. Hash the items independently into
$b=\max\{1,\lceil k/\lambda\rceil\}$ buckets. Within each bucket,
color its items independently with $q=128\lambda^2$
colors. For each color form the factor consisting of $0$ and the
vectors carrying that color, and multiply these factors using
Lemma~\ref{lem:boolean-product}. Repeat the coloring
$\lceil\log_2(2b/\xi)\rceil$ times and take the union of the resulting
supports for that bucket. Finally add the supports of the different
buckets and intersect the result with $B$.

Fix the witness $X$ from the statement. Its expected load in each bucket
is at most $\lambda$. A Chernoff bound and a union bound over the
buckets show that all these loads are at most $8\lambda$ with probability
at least $1-\xi/2$. Conditional on these loads, the probability of
a collision among its items in one bucket is at most
$\binom{8\lambda}{2}/q\le1/4$. Independent repetitions and a union
bound over the buckets imply, with probability at least $1-\xi/2$,
that each bucket has a repetition in which all its witness items
receive distinct colors. Its selected sum is then in that bucket's
union. Adding the bucket supports retains $\sum_{i\in X}\vec{w}_i$, and
the final intersection retains it because it lies in $B$.

For the cost, every color factor lies in $\alpha\cdot B$, and every bucket
support lies in $q\alpha\cdot B$. We have
\[
 bq=O(k\lambda+\lambda^2),\qquad
 bq\alpha=O_d(\lambda^2).
\]
Equation~\eqref{eq:productbound}, applied at each level of the bucket
trees and then to their roots, bounds their total array cost by
\[
 \widetilde O_d\left(m+bq+
       \prod_j(1+bq\alpha \Delta_j)\right)
 =\widetilde O_d(m+|B|).
\]
Repetitions, unions, and the final product of bucket supports introduce
only logarithmic factors. The input items are scanned once per
coloring, within the same bound. Each product selects at most one item
per color, and products between buckets use disjoint items. Taking
unions and intersecting with $B$ preserve soundness. This is the
two-stage color-coding construction of \cite{Bringmann}, applied to
signed vectors with the indicated coordinate bounds.
\end{proof}

\paragraph{Proof of the signed-budget lemma.}
\begin{proof}[Proof of Lemma~\ref{lem:quota}]
If $\Delta_j=0$, every admissible witness omits items nonzero in coordinate
$j$; discard those items and suppress that coordinate. Discard also any
item with $|w_{ij}|>\Delta_j$ in a remaining coordinate, and ignore zero
vectors. These operations preserve all required sums. If no nonzero
item remains, return $\{0\}$. Henceforth all budgets are positive and
all remaining items are nonzero. Write
\[
 B=\prod_j[-\Delta_j,\Delta_j],\qquad |B|=\prod_j(2\Delta_j+1).
\]

\emph{Partition by normalized size.}
For each item put $\theta_i=\max_j |w_{ij}|/\Delta_j$ and place it in the class
$I_s$ determined by $2^{-s-1}<\theta_i\le2^{-s}$, where $s\ge0$.
Since a nonzero integer vector has $\theta_i\ge1/\max_j\Delta_j$, the number
$g$ of nonempty classes is $O(\log(2+|B|))$. A witness satisfying
\eqref{eq:absolutequota} has
\[
 \sum_{i\in X}\theta_i
 \le\sum_{j=1}^d\frac{\sum_{i\in X}|w_{ij}|}{\Delta_j}\le d.
\]
Thus $|X\cap I_s|\le k_s$, where
$k_s=\min\{|I_s|,\lceil d2^{s+1}\rceil\}$.
With $\alpha_s=2^{-s}$, we have $k_s\alpha_s\le2d$ and
$|w_{ij}|\le\alpha_s\Delta_j$ for $i\in I_s$.

\emph{Compute and combine the class supports.}
Apply Lemma~\ref{lem:one-class-color-coding} to each nonempty class,
with quota $k_s$ and failure allowance $\xi/(g|B|)$, obtaining
$S_s\subseteq\mathcal{S}(\{\vec{w}_i:i\in I_s\})\cap B$. Combine the $S_s$ in a
balanced Boolean product tree, intersecting every intermediate support
with $B$. A convolution of two supports in $B$ needs an array in $2\cdot B$,
of size $O_d(|B|)$, before this intersection. The class computations and
their combination therefore cost $\widetilde O_d(n+|B|)$ in total,
with polylogarithmic dependence on $1/\xi$.

\emph{Soundness and simultaneous completeness.}
Every class support uses distinct items, and different classes have
disjoint items. Consequently every computed sum is an actual subset
sum; the intersections only remove sums. For completeness, choose one
admissible witness $X_{\vec{z}}$ for every required output sum $\vec{z}$ before
drawing any randomness. There are at most $|B|$ such sums. For every
class, $X_{\vec{z}}\cap I_s$ has at most $k_s$ items and its sum lies in
$B$, so Lemma~\ref{lem:one-class-color-coding} retains that sum with
probability at least $1-\xi/(g|B|)$. A union bound over at most $g|B|$ pairs
$(\vec{z},s)$ shows that all these class sums are retained with probability
at least $1-\xi$. On this event, every intermediate sum belonging to
any $X_{\vec{z}}$ also lies in $B$: its items form a subset of $X_{\vec{z}}$, so its
absolute contribution in every coordinate is at most $\Delta_j$. None is
lost when intermediate supports are intersected with $B$, and every
required $\vec{z}$ reaches the final support. Reinsert any suppressed zero
coordinates to obtain the asserted output in the original space.
\end{proof}

\section{Lattice boxing and properization}
\label{app:lattice-tools}

The statements and algorithmic roles of these tools appear in
Sections~\ref{sec:carrier} and~\ref{sec:normal}. We collect their
constructive proofs here. They use the standard discrete John,
lattice-reduction, and quotient arguments attributed in the main text;
the proofs make explicit the bases, lattice-point counts, and bit costs
required by our recursion.

\subsection{Boxing in a basis of the full lattice}
\label{app:full-lattice-boxing}

\begin{proof}[Proof of Lemma~\ref{lem:carrier-full-lattice-box}]
Choose vertices $\vec{v}_1,\ldots,\vec{v}_r$ of $K$ maximizing the absolute
determinant of their column matrix $A$.  The maximum is nonzero, and
the enumeration has constant size.  Put $K'=A^{-1}K$ and
$\Lambda'=A^{-1}\Lambda$.  Replacing a column of $A$ by any vertex and
using Cramer's rule bounds every coordinate of that vertex in the new
coordinates by one in absolute value.  Symmetry and convexity give
\[
 [-1/r,1/r]^r\subseteq\conv\{\pm \vec{e}_1,\ldots,\pm \vec{e}_r\}
 \subseteq K'\subseteq[-1,1]^r.
\]
Compute an LLL-reduced basis $\vec{d}_1,\ldots,\vec{d}_r$ of the rational lattice
$\Lambda'$, with matrix $D$.  Its orthogonality defect obeys
\[
 \omega=\frac{\prod_j\|\vec{d}_j\|_2}{|\det D|}\le\omega_r
\]
for a fixed dimension-dependent constant; the basis is computable in
polynomial bit time by Theorem~\ref{thm:import-lattice-reduction}.
If $\vec{\ell}_j$ is row $j$ of $D^{-1}$, the cofactor estimate yields
\[
 \|\vec{\ell}_j\|_2
 \le\frac{\prod_{i\ne j}\|\vec{d}_i\|_2}{|\det D|}
 \le\frac{\omega_r}{\|\vec{d}_j\|_2}.
\]
Compute the rational coordinate maxima
\[
 r_j=\max_{\vec{x}\in K'}|\vec{\ell}_j\vec{x}|,
 \qquad n_j=\lceil r_j\rceil.
\]
In each basis coordinate, at least one of the independent lattice
points in $K'$ has a nonzero coefficient.  That coefficient is an
integer, so $r_j\ge1$ and $n_j\le2r_j$.  Since
$K'\subseteq[-1,1]^r$, the cofactor estimate implies
\[
 n_j\|\vec{d}_j\|_2\le2\sqrt r\,\omega_r.
\]
The full real coefficient box contains $K'$ by construction and has
Euclidean radius at most $2r\sqrt r\,\omega_r$.  Because $K'$ contains
$[-1/r,1/r]^r$, this box is contained in $\gamma_r\cdot K'$ for a constant $\gamma_r$.
Transforming back and setting $\vec{b}_j=A\vec{d}_j$ gives a basis of the entire
original lattice and the two real containments.

The following residue-class argument is the usual doubling estimate for
convex lattice sets; compare~\cite[Lemma 3.3]{TV}.
For the cardinality bound choose an integer $p\ge2\gamma_r$.  Partition
$a\cdot C\cap\Lambda$ into its at most $p^r$ residue classes modulo
$p\cdot\Lambda$.  In a nonempty class choose $\vec{x}_0$.  For every $\vec{x}$ in that
class, symmetry and convexity give
\[
 (\vec{x}-\vec{x}_0)/p\in\Lambda\cap(2\gamma_r/p)a\cdot K
 \subseteq\Lambda\cap a\cdot K.
\]
Within a class this map is injective.  Thus the constant $p^r$ works uniformly
for all $a>0$.  The same argument shows that a fixed dilation of any
origin-symmetric convex body has at most a factor depending only on the
dimension and the dilation more lattice points than that body.

All operations use rational linear algebra and fixed-complexity
optimization in fixed dimension, together with polynomial-time lattice
reduction after clearing denominators.  These operations prove the
claimed bit bounds.  If $K$ is a coordinate box, normalize by its
diagonal matrix of its coordinate bounds instead of enumerating vertices; then each
$r_j$ is the sum of the absolute values of the entries of $\vec{\ell}_j$.
\end{proof}

\subsection{Constructive properization}
\label{app:properization}

\begin{proof}[Proof of Lemma~\ref{lem:carrier-properization}]
We give the algorithmic version of Theorem~\ref{thm:import-properization}
needed here.
Represent the progression by an integer homomorphism $a:\Z^r\to\Z$
and a symmetric box $B=\prod_j[-\beta _j,\beta _j]$, $\beta _j\in\Z_{\ge0}$; discard zero-width
coordinates.  Fixed-dimensional integer programming detects a nonzero
$\vec{z}\in 2t\cdot B\cap\ker a$, or certifies there is none; the nonzero condition
is a union of $2r$ fixed-dimensional integer programs.  The latter
certificate implies $t$-properness.  In the former case replace $\vec{z}$ by its
primitive divisor and quotient $\Z^r$ by $\Z \vec{z}$.  A unimodular
coordinate change computes this torsion-free quotient and its induced
integer homomorphism in polynomial bit time.  The image $K$ of $B$ is a
symmetric rational polytope in the quotient space, with constantly many
vertices.  A rank-zero quotient is handled directly by the progression
$\{0\}$.  Otherwise the images of the vectors $\pm \vec{e}_i$ lie in $K$
and generate the whole quotient lattice, so $K$ contains a full
independent lattice set.

Every quotient lattice point in $K$ has an integer representative in
$B+[-1/2,1/2]\vec{z}\subseteq(1+t)\cdot B$: start with a real representative in
$B$ and round its displacement along the primitive kernel direction.
Thus $|K\cap\Lambda|\le |(1+t)\cdot B\cap\Z^r|=O_{r,t}(|B|)$.
Apply Lemma~\ref{lem:carrier-full-lattice-box} to $K$ and the quotient
lattice.  It returns a containing coefficient box in a basis of the
entire quotient lattice, with cardinality
$O_r(|K\cap\Lambda|)=O_{r,t}(|B|)$.  Composing this basis with the induced
integer scalar homomorphism gives an integer GAP containing the
original one.  This gives an explicit polynomial-bit implementation
of the quotient step.

Each detected relation reduces the rank.  Repeating the operation at
most $r$ times gives the asserted proper progression; the finitely many
cardinality factors depend only on $r,t$.  All lattice bases, quotient
maps, and rational polytopes have polynomial bit length throughout.
\end{proof}

\subsection{Local lattice boxing with an oracle}
\label{app:local-boxing}

\begin{proof}[Proof of Lemma~\ref{lem:boxing}]
Fixed-dimensional convex rounding computes an origin-centered rational
ellipsoid $E$ such that $M\subseteq E\subseteq a_d\cdot M$, for a fixed $a_d$;
see \cite[Section 1 and Theorems 2.13--2.14]{GLS}.  Apply lattice basis
reduction in the quadratic norm defining $E$.  For the analysis only,
normalize $E$ to the Euclidean unit ball and write $\vec{b}_1,\ldots,\vec{b}_d$ for the
resulting basis of the transformed lattice, with column matrix $B$.
The reduced-basis guarantee gives the orthogonality-defect bound
\[
 \omega:=\frac{\prod_j\|\vec{b}_j\|_2}{|\det B|}\le\omega_d,
\]
where $\omega_d$ depends only on $d$; the construction and guarantee are
\cite[Theorems 2.9--2.10]{GLS}.
Let $r_j$ be the Euclidean norm of row $j$ of $B^{-1}$.
Cofactor estimates give
\[
 r_j\le\frac{\prod_{i\ne j}\|\vec{b}_i\|_2}{|\det B|}
       =\frac{\omega}{\|\vec{b}_j\|_2},\qquad
 \prod_jr_j\le\frac{\omega^{d-1}}{|\det B|}.
\]
In basis coordinates the unit ball satisfies $|z_j|\le r_j$,
with equality attainable in each coordinate.  For each basis coordinate $j$, at least one of the $d$ independent
lattice vectors in the hypothesis has a nonzero integer coefficient in
that coordinate. Since its $\ell$-multiple belongs to $M\subseteq E$,
the corresponding coordinate has absolute value at least $\ell$ inside
$E$. Hence $r_j\ge\ell\ge1$ for every $j$.
Set $n_j=\lceil r_j\rceil$, so $n_j\le2r_j$, and define $D$
by $|z_j|\le n_j$ in the corresponding original lattice basis.  The box contains
$E$, and in the normalized coordinates every one of its points has norm
at most
\[
 \sum_j n_j\|\vec{b}_j\|_2\le2d\omega_d.
\]
Thus $M\subseteq D\subseteq2d\omega_d\cdot E\subseteq \gamma _d\cdot M$.
Also $2n_j+1\le5r_j$, whence
\[
 \begin{aligned}
 |D\cap\Lambda|=\prod_j(2n_j+1)
 &\le\frac{5^d\omega_d^{d-1}}{|\det B|}\\
 &=O_d\!\left(\frac{\vol(E)}{\det\Lambda}\right)
  =O_d\!\left(\frac{\vol(M)}{\det\Lambda}\right).
 \end{aligned}
\]
Translation does not change widths.  Outward rounding translated interval
endpoints adds at most two entries in each coordinate, absorbed by
$r_j\ge1$.  Fixed dilations are treated in the same way, with their fixed
scale included in the constant.

The normalization is conceptual, so no irrational matrix is required in
the implementation.  Basis reduction uses the rational quadratic form of
$E$ directly.  In the resulting rational basis each $r_j^2$ is rational;
its integer square-root ceiling is computed by rational comparisons.
The oracle rounding, lattice reduction, and these fixed-dimensional
operations have polynomial bit cost and polynomial output bit length.
This also proves the asserted constructive bounds.
\end{proof}

\section{Bit representations and oracle costs}
\label{app:bit-details}

The objects and the bit-cost guarantee are stated before
Lemma~\ref{lem:bits}. The proof below applies fixed-dimensional
rational geometry and integer linear algebra to those objects, then
checks that coordinate descriptions can be recomputed from current
bounded vectors. These representation details are used by both the
set and multiset recursions.

\begin{proof}[Proof of Lemma~\ref{lem:bits}]
We first bound the oracle and GAP data, then the rational computations
for an interior reduction, and finally the coordinates stored at the next node.

\paragraph{Zonotope and lattice oracles.}
For a zonotope $Z=\vec{z}_0+\sum_i[0,\vec{u}_i]$, separation at rational $\vec{x}$ reduces to maximizing
\[
 \vec{a}\cdot(\vec{x}-\vec{z}_0)-\sum_i\max(0,\vec{a}\cdot \vec{u}_i)
 \quad\text{over }\|\vec{a}\|_\infty\le1.
\]
This is a fixed-dimensional concave piecewise-linear problem with an $O(n)$-time hypograph separator. Its maximum is zero precisely for $\vec{x}\in Z$. If positive, it is bounded below by an inverse number with polynomial bit length: an optimum is attained at a vertex of the arrangement of the hyperplanes $\vec{a}\cdot \vec{u}_i=0$ and the cube faces, so fixed-size determinant bounds apply. Approximate optimization to below this rational gap supplies exact membership or a valid separating normal without enumerating the arrangement. The rational-polytope oracle algorithms in Theorem~\ref{thm:import-geometry} supply feasibility, affine sections, and the integer-point-or-thin-direction dichotomy in fixed dimension with polynomial dependence on bit length. All their oracle calls cost $n\polylog N$ in total. Fixed-dimensional HNF, lattice basis reduction, and constant-complexity convex-body boxing have the same bit bounds. A weighted zonotope replaces each summand by $c_i\max(0,\vec{a}\cdot \vec{u}_i)$, so its oracle scans types rather than individual copies.

\paragraph{GAP data and the old-vector map.}
In the finite prefix of Lemma~\ref{lem:bootstrap}, $h$ and the encoded scalar bound $L$ are polynomially bounded in $N$. Let $s=a_\eta h$ be the repeated-sumset scale used in the proof of Theorem~\ref{thm:carrier}. Its phase measure obeys ${m_{\mathrm{ph}}}\le sL+1=N^{O_\delta(1)}$, and its constructive Freiman input lies in $[0,10sL]$. The construction in \cite[Appendix A.5]{RW} first obtains a progression $\mathcal Q\subseteq2A-2A$ and a bounded covering set $X\subseteq A$, with $A\subseteq\mathcal Q-\mathcal Q+X$. An active progression step is the difference of two neighboring progression points, so its magnitude is $O(sL)$. Incorporating $X$ as constantly many binary directions retains that bound and only changes the fixed rank and volume constants. Thus this construction starts with $O_\delta(\log N)$-bit scalar descriptions, rather than an unspecified unit-cost output on arbitrarily long integers.

Properization uses only constantly many fixed-dimensional quotient and lattice operations, with polynomial logarithmic temporary bit cost. The returned coefficient box has polynomial cardinality and positive integer coordinate bounds, so every GAP coordinate of an item is polynomially bounded in $N$. The coefficient vectors span their full real rank. Recompute their scalar homomorphism from a fixed set of independent coefficient vectors and their scalar encodings. Cramer's rule bounds its numerator and denominator sizes by $O_\delta(\log N)$; the map is integral on the supplied lattice in GAP coordinates. The same computation bounds the integer map $F$ from independent coefficient vectors and their old vectors, using the existence and integrality proved in Lemma~\ref{lem:carrier-old-map}. The reserve translation $\vec{z}_0$ has a witness using $O_\delta(h)$ coefficient vectors of the reserved items and therefore has the same bit bound. No historical coordinate matrices need to be multiplied together.

\paragraph{Fiber geometry and the thin functional.}
A zonotope with polynomially bounded integer generators, multiplicities, and center has rational facets of polynomially bounded coefficient magnitude in fixed dimension. A facet normal is given by fixed-size minors of generator vectors, and its support value is a weighted type sum. The number of facets can be large; these are coefficient bounds, not a requirement to enumerate them. Fixed-dimensional integer linear algebra gives a polynomially bounded particular point and basis for the affine kernel lattice of $F$. Substituting that parametrization into the facet inequalities retains polynomial coefficient bounds. Thus a full-dimensional fiber in its kernel coordinates has outer radius $N^{O_\delta(1)}$ and contains a ball of radius $N^{-O_\delta(1)}$: a full-dimensional simplex of its rational vertices has inverse-polynomial volume and polynomial diameter, by fixed-size determinant bounds. These provide the quantitative radius bounds required by the oracle algorithms.

Suppose the point-or-thin-direction algorithm returns an integer functional of width at most $f_r$ on such a fiber, and let $r_{\rm fib}$ be the supplied inverse-polynomial inner-radius bound in integer kernel coordinates. Its Euclidean coefficient norm is at most $f_r/(2r_{\rm fib})$, because the fiber contains a radius-$r_{\rm fib}$ ball. Its coefficients therefore have $O_\delta(\log N)$ bits. The kernel lattice is saturated in the lattice in GAP coordinates. Extending the functional by fixed-dimensional HNF or Smith normal form uses polynomially bounded integer coefficients, so the ambient integer functional $\psi$ has the same bit bound.

\paragraph{A bounded rational grid for the interior dual.}
In the notation of the interior reduction, write $c'_i=1$ for nonreserve set items; in the multiplicity reduction, $c'_i=c_i-r_i$. The dual objective is
\[
 G(\vec{\beta})=\psi(\vec{z}_0)+\vec{\beta}(\vec{y}_0-F\vec{z}_0)
       +\sum_i c'_i\min(0,\psi(\vec{u}_i)-\vec{\beta} \vec{v}_i).
\]
Set $\ell_{\mathrm{lip}}=\|\vec{y}_0-F\vec{z}_0\|_1+\sum_i c'_i\|\vec{v}_i\|_1$. It is polynomially bounded in $N$, and
\[
 |G(\vec{\beta})-G(\vec{\beta}')|\le \ell_{\mathrm{lip}}\|\vec{\beta}-\vec{\beta}'\|_\infty.
\]
The interior buffer bounds a maximizing vector within a ball of polynomial radius, as in Lemma~\ref{lem:deep}. Optimize on that ball to additive accuracy $1/2$ with the type-scan hypograph oracle. Put ${q_{\mathrm{grid}}}=\lceil\max(1,4\ell_{\mathrm{lip}})\rceil$ and round every coordinate to the nearest multiple of $1/{q_{\mathrm{grid}}}$. The objective changes by at most $\ell_{\mathrm{lip}}/(2{q_{\mathrm{grid}}})\le1/8$. Thus the rounded vector has dual slack less than one, and its numerator and denominator are polynomially bounded in $N$. Define $\alpha$ afterwards by the exact displayed minimum in the interior reduction. This supplies the required rational $\vec{\beta}$ with $O_\delta(\log N)$-bit entries and $\softO{n}$ bit cost, or $\softO{m}$ in the multiplicity version.

\paragraph{Storing the next instance.}
After reboxing, each coordinate of a current item vector is bounded by
the side length of the allocated array, hence by $O(w)$. For each
coordinate $j$, retain only targets whose $j$th coordinate lies in
\[
 \left[\sum_i\min(0,v_{ij}),\ \sum_i\max(0,v_{ij})\right].
\]
Every subset sum satisfies these bounds. Thus this filtering preserves
all solutions and bounds every retained target coordinate by $O(nw)$.
In the count representation, use
\[
 \left[\sum_i c_i\min(0,v_{ij}),\
       \sum_i c_i\max(0,v_{ij})\right],
\]
which gives the bound $O((\sum_i c_i)w)$. The set recursion has
$n=O_\delta(w)$ by \eqref{eq:nscan}; in the count representation,
total multiplicity is polynomially bounded by hypothesis. These item
and target coordinates therefore require $O_\delta(\log N)$ bits.

Choose a retained target as the new array center, using the fixed-factor
box enlargement already included in the normal-reduction constants.
Recompute a basis of the lattice generated by the current items by
fixed-dimensional Hermite normal form. Its entries have the same bit
bound. To recover the original scalar functional, solve a fixed-size
rational system using independent current vectors and their oriented
original values $a_i$. Cramer's rule bounds its entries by
$O_\delta(\log N)$ bits. The earlier parts of this proof provide the
same bound for the GAP homomorphism, $F$, the reserve translation,
and the rounded dual vector. Only the coordinate changes needed for the current reduction
are retained; their historical matrix products are unnecessary.
This proves the storage and bit-cost claims.
\end{proof}

\end{document}